\documentclass[conference]{IEEEtran}

\usepackage{ifpdf}
\usepackage{amsthm}
\usepackage{cite}
\usepackage{tabularx}
\usepackage{comment}
\usepackage{xcolor}

\newtheorem{proposition}{Proposition}
\ifCLASSINFOpdf
   \usepackage[pdftex]{graphicx}
 
\else
 
   \usepackage[dvips]{graphicx}
 
\fi
\usepackage[cmex10]{amsmath}
\usepackage{algorithmic}
\usepackage{array}
\usepackage{mdwmath}
\usepackage{mdwtab}
\usepackage{eqparbox}
\usepackage[font=footnotesize]{subfig}
\usepackage{fixltx2e}
\usepackage{stfloats}
\usepackage{url}

\begin{document}
\title{SecBeam: Securing mmWave Beam Alignment  against Beam-Stealing Attacks} 

\author{
\IEEEauthorblockN{Jingcheng Li, Loukas Lazos, Ming Li}
\IEEEauthorblockA{\textit{University of Arizona}\\
\{jli2972, llazos, lim\}@arizona.edu}
}

\maketitle
\begin{abstract}
Millimeter wave (mmWave) communications employ narrow-beam directional communications to compensate for the high path loss at mmWave frequencies. Compared to their omnidirectional counterparts, an additional step of aligning the transmitter's and receiver's antennas is required. In current standards such as 802.11ad, this beam alignment process is implemented via an exhaustive search through the horizontal plane known as beam sweeping. 
However, the beam sweeping process is unauthenticated. As a result, an adversary, Mallory, can launch an active beam-stealing attack by injecting forged beacons of high power, forcing the legitimate devices to beamform towards her direction. Mallory is now in control of the communication link between the two devices, thus breaking the false sense of security given by the directionality of mmWave transmissions.

Prior works have added integrity protection to beam alignment messages to prevent forgeries. In this paper, we demonstrate a new beam-stealing attack that does not require message forging. We show that Mallory can amplify and relay a beam sweeping frame from her direction without altering its contents. Intuitively, cryptographic primitives cannot verify physical properties such as the SNR used in beam selection. We propose a new beam sweeping protocol called SecBeam that utilizes power/sector randomization and coarse angle-of-arrival information to detect amplify-and-relay attacks. We demonstrate the security and performance of SecBeam using an experimental mmWave platform and via ray-tracing simulations.

\end{abstract}

\maketitle

\section{Introduction}
Millimeter Wave (mmWave) communications are a key enabling technology for the Next-Generation of wireless networks \cite{7959169,7999294}. The abundance of available bandwidth offers unprecedented opportunities for  high data rates, ultra-low latency, and massive connectivity. However, the short wavelength poses new challenges at the physical layer due to the substantial signal attenuation (which is proportional to the center frequency) and susceptibility to blockage \cite{7999294}. To compensate for this high attenuation, mmWave devices employ high-gain directional transmissions which are possible due to the small antenna factor necessary at small wavelengths. New highly-directional antenna systems pack many antenna elements into a small area which can be electronically steered to form high-gain pencil like beams at any desired direction.

Compared to their omnidirectional counterparts, narrow-beam directional systems require continuous beam management \cite{wang2022beam}. The beam management process consists of the initial phase of beam alignment between the Tx and the Rx and the subsequent phase of beam tracking if any of the Tx or Rx are mobile, or the physical environment changes. Recent standards such as the IEEE 802.11ad \cite{IEEE:802.11ad} and its enhanced version IEEE 802.11ay \cite{ghasempour2017ieee}, describe beam alignment as an exhaustive search process called {\em sector-level sweep} (SLS) where all possible directions are scanned before the optimal Tx and Rx antenna orientations are decided. Specifically, one device acts as {\em the initiator $I$}, whereas the other acts as {\em the responder $R$.} The initiator sequentially sends sector sweep frames (beacons) in every candidate direction using narrow-beam transmissions while the responder receives in quasi-omni mode. 
\begin{figure}[t]
\centering
\includegraphics[width=0.80\linewidth]{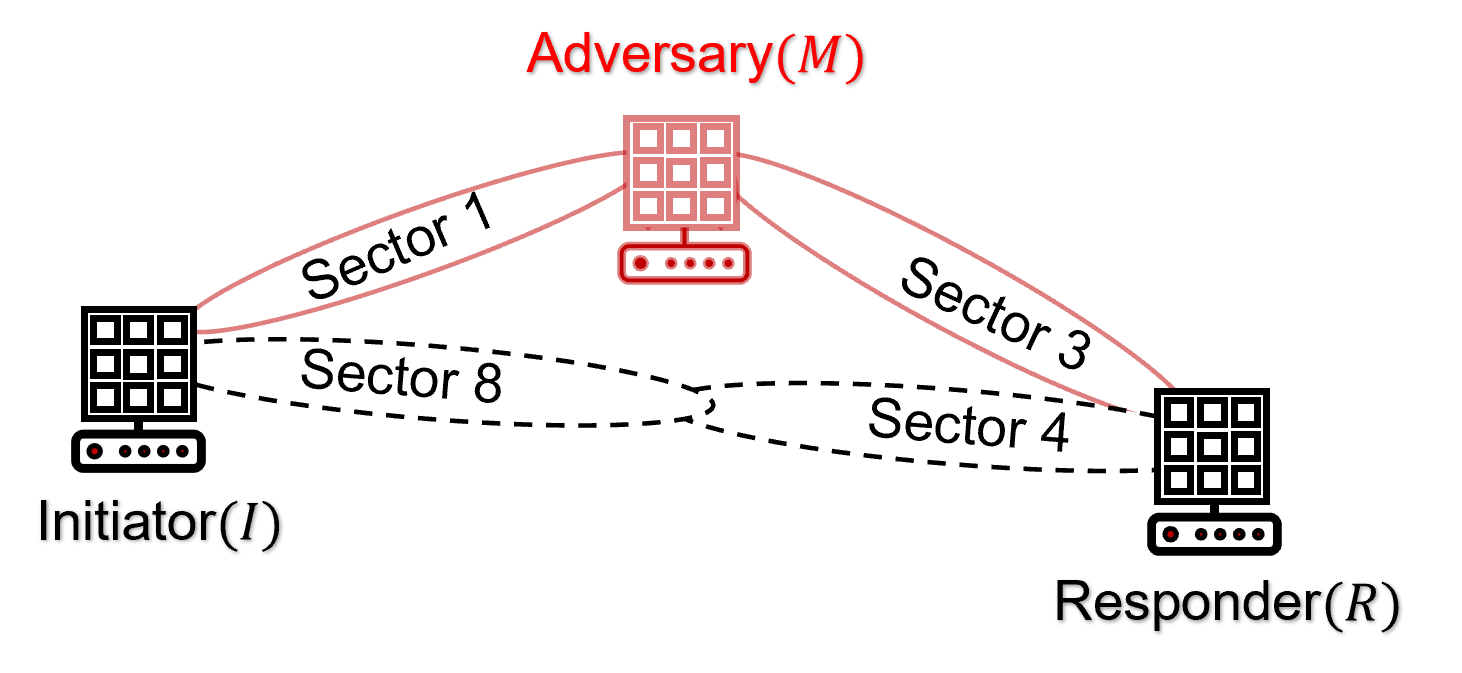}
\vspace{-0.1in}
\caption{Beam-stealing attack against the SLS protocol. $I$ and $R$ beamform towards the adversary instead of the LoS path.} 
\label{beam steal}
\vspace{-0.20in}
\end{figure}
The responder measures the received power from each receivable beacon and determines the initiator's direction (sector) with the highest signal-to-noise ratio (SNR). The process is repeated in the opposite direction to define the responder's sector with the highest SNR. Upon completion, the initiator's and the responder's beams are aligned to the direction that maximizes the cumulative antenna gain which can be a line-of-sight (LoS) path or a reflection path. 

The beacons in the SLS protocol, however, are unauthenticated and can be easily forged. An adversary, Mallory, can launch an active beam-stealing attack by injecting forged beacons of high power, forcing the legitimate Tx and Rx to beamform towards her direction \cite{steinmetzer2018beam}. This attack is shown in Figure~\ref{beam 
 steal} where the initiator selects sector 1 instead of sector 8 and the responder selects sector 3 instead of 4. Mallory is now in control of the communication link between $I$ and $R$, thus breaking the false sense of security given by the directionality of mmWave transmissions \cite{kim2017analysis, zhu2017secure}. She can eavesdrop on the link, apply traffic analysis to infer information from encrypted communications \cite{papadogiannaki2021survey},  inject and modify messages, or selectively drop them.

A few methods have been proposed to prevent beam-stealing attacks \cite{steinmetzer2018authenticating, yang2020man, tian2019secure, balakrishnan2019physical, wang2021exploiting, wang2020machine}. Steinmetzer $et\ al.$ introduced an authenticated beam sweep protocol \cite{steinmetzer2018authenticating}.  The main idea is to prevent the forging of beacons by guaranteeing their authenticity and freshness via cryptographic means (e.g., encrypt and sign beacon transmissions). However, {\em we postulate that a beam-stealing attack is still feasible against an authenticated beam sweep protocol.} Specifically, we demonstrate a new amplify-and-relay attack that does not require access to cryptographic primitives on behalf of Mallory.   

{\bf Amplify-and-relay beam-stealing attacks}. To achieve beam-stealing when beacons are authenticated, Mallory can launch a relay attack as shown in Figure~\ref{beam steal}. When $I$ transmits a beacon toward Mallory (sector 8), Mallory swiftly amplifies and  relays the beacon toward the responder. Mallory does not need to decode the beacon or modify its contents, as the beacon only needs to be received at higher power. Because the beacon is generated by $I$ and contains the correct authenticator, it passes authentication at $R$. Due to the signal amplification, $R$ is led to believe that sector 8 yields the highest SNR and provides this feedback to $I$. The relay attack is repeated in the $R$ to $I$ direction, leading $I$ to believe that the third sector of $R$ yields the highest SNR. 

\medskip

{\em The amplify-and-relay attack is feasible because the added cryptographic protections validate the contents of the beacons, but fail to authenticate the physical property used for sector selection which is the SNR. Practically, the attack distorts the wireless environment and cannot be dealt with by cryptographic means alone.}   
\medskip

Existing relay attacks that aim at breaking proximity constraints (e.g., \cite{francillon2011relay,francis2011practical}) are different from amplify-and-relay beam-stealing attacks in terms of goals and assumptions. 
Other known attacks of similar nature are signal cancellation attacks where a carefully placed adversary relays a transmitted signal with the purpose of causing destructive interference at the intended receiver \cite{popper2011investigation,moser2019digital}, which are harder to realize. An intuitive defense against relay attacks  is to differentiate between the legitimate transmitter and the relay. Fingerprinting techniques such as those in  \cite{balakrishnan2019physical, wang2021exploiting, wang2020machine} extract unique fingerprints from the devices such as antenna patterns, frequency offset and DC offset characteristics, etc. However, these methods require training, may utilize multiple access points and are not applicable when devices meet for the first time. Others, such as the angle-of-arrival-based detection method  \cite{yang2020man} and SNR fingerprinting methods \cite{wang2021exploiting} assume that the locations of the devices are fixed and known.

In this paper, we address the problem of securing the SLS protocol of 802.11ad against  amplify-and-relay beam-stealing attacks. {\bf Our main contributions are as follows:}
\begin{itemize}
    \item We demonstrate a new amplify-and-relay attack against the beam alignment process of mmWave communications. This attack defeats cryptographic defenses that rely on beacon authentication to detect beacon forging. 
    \item We develop an attack detection method that exploits the power-delay profile (PDP) of the mmWave channel. Assuming that the LoS path is available, any amplified relay path will have a longer delay and a higher SNR compared to the LoS path. However, a signal traveling through the LoS path should incur the least attenuation due to its free space propagation over a shorter distance. We use the discrepancies in the PDP to detect amplify-and-relay attacks over paths longer than the LoS path. 
    \item The PDP detection method assumes the presence of the LoS path, which may not be true. To generalize attack detection to all environments, we propose the SecBeam protocol that derives security from a combination of power randomization and coarse AoA detection. The focal idea is to hide the sector identity and randomize the beacon transmit power to prevent Mallory from fine-tuning her amplification power. Without knowledge of the sector ID and transmit power, Mallory has to amplify the beacons from all sectors she can hear. This leads to a disproportionate number of beacons arriving from the same direction, violating the geometrical mmWave channel model \cite{jameel2018propagation}. SecBeam requires minimal modifications to the original SLS protocol of 802.11ad.
    \item We analyze the security of SecBeam and show that it can detect an amplify-and-forward attacker that relays incoming beacons at fixed transmit power or at fixed amplification. We experimentally evaluate the performance and security of SecBeam on a 28GHz mmWave testbed in typical indoor environments. Our experiments demonstrate that SecBeam allows legitimate devices to optimally align their beams while detecting amplify-and-relay attacks with high probability.   

\end{itemize}

\section{The SLS protocol in  IEEE 802.11ad and attack}
\subsection{The SLS protocol in  IEEE 802.11ad}
\label{sector sweep}
The beam alignment process in the IEEE 802.11ad, also known as ``sector-level sweep (SLS)", aims at maximizing the combined Tx-Rx antenna gain \cite{IEEE:802.11ad}. This is done by sweeping the plane using fine-beam and quasi-beam antenna configurations and identifying the directions (sectors) of optimal antenna alignment. Sweeping is achieved electronically by controlling the individual antenna elements \cite{gao2021experimental}.  An example of a steerable antenna pattern is shown in Figure~\ref{fig:quasi}. The antenna can be set to quasi-omni mode for wider coverage, or to a fine-beam mode for the highest gain. 

\begin{figure}[t]
\centering
\includegraphics[width=0.8\linewidth]{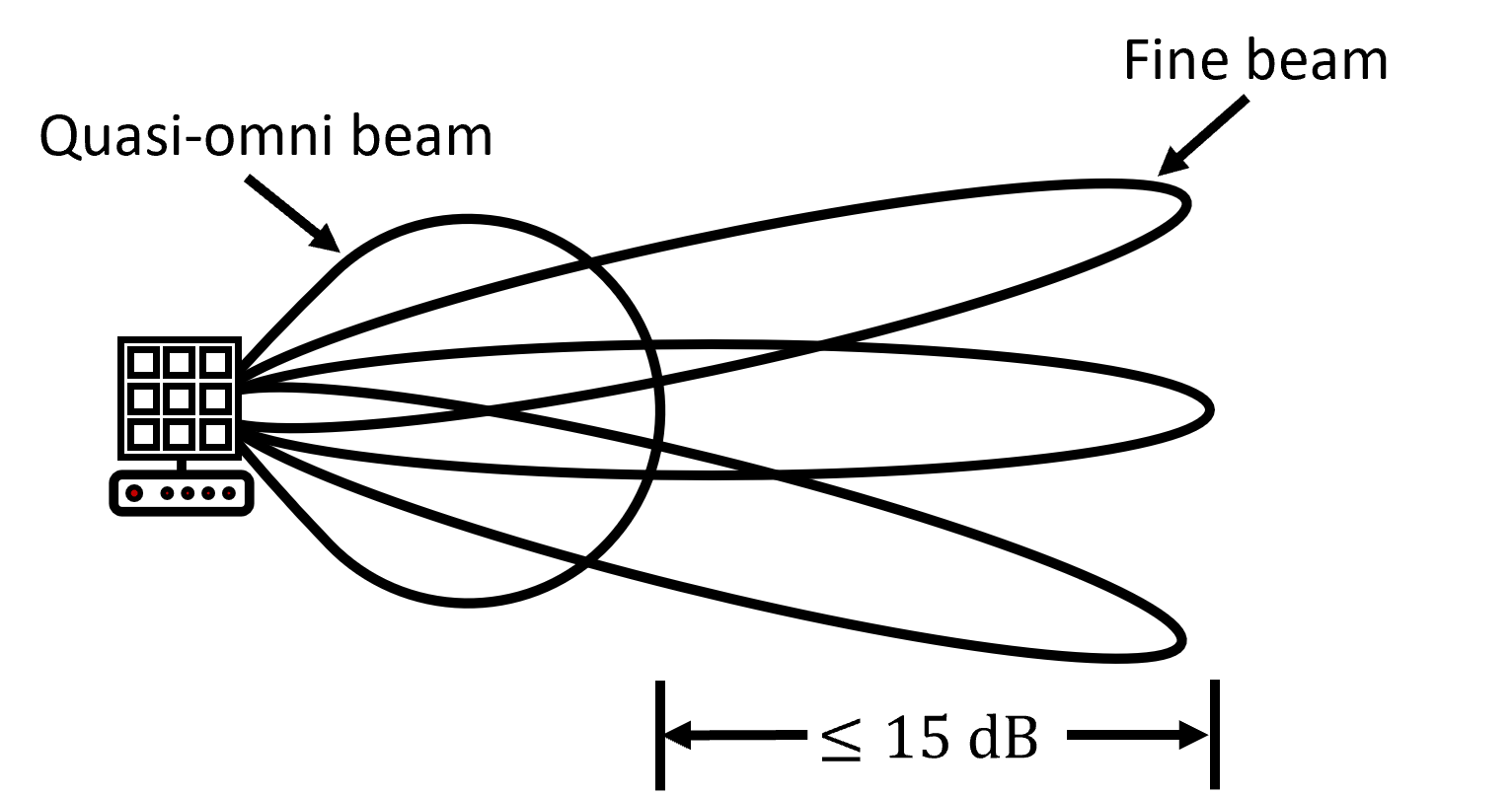}
\vspace{-0.1in}
\caption{A steerable mmWave antenna that can operate in fine-beam mode or quasi-omni mode.}
\label{fig:quasi}
\vspace{-0.20in}
\end{figure}

The main idea of SLS is shown in Figure~\ref{System model}. One device assumes the role of the initiator ($I$) whereas the other assumes the role of the responder ($R$).  The initiator sets its antenna in fine-beam mode and sweeps through the plane transmitting sector sweep frames (SSF). The receiver operates in quasi-omni mode, recording the received SNR from each sector as shown in Figure~\ref{System model}(a). In the next phase, the two devices switch roles. The responder transmits in fine-beam mode while the initiator is in quasi-omni mode. Additionally, the responder provides feedback about the optimal initiator's sector (highest SNR), as recorded from the previous phase. Finally, the initiator provides feedback to the responder about the optimal responder's sector and the two antennas are aligned as shown in Figure~\ref{System model}(b).  In detail, the SLS protocol consists of the three phases shown in Figure~\ref{802.11 and attack}(a). 

\begin{figure}[t]%
\centering
\setlength{\tabcolsep}{-2pt}
\begin{tabular}{c}
  \includegraphics[width=0.8\columnwidth]{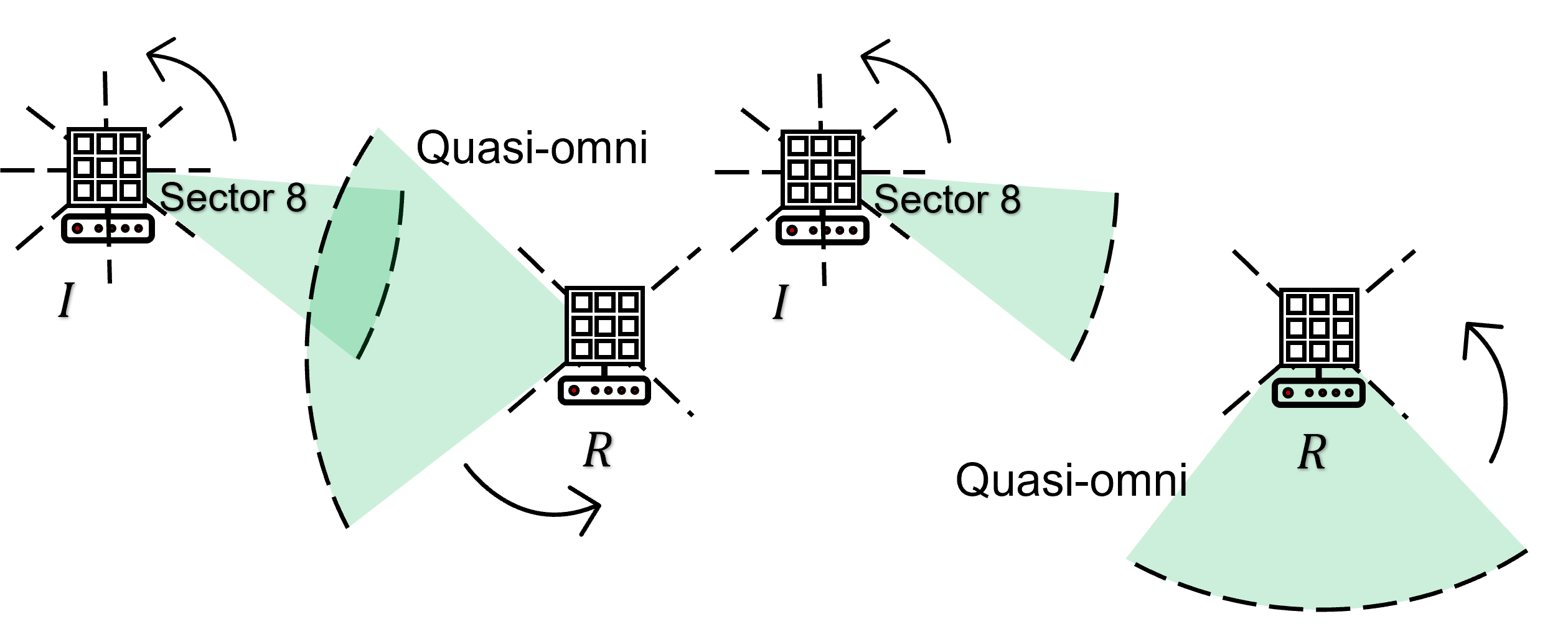} \\
  (a) \hspace{90pt}  (b)\\
  \includegraphics[width=0.8\columnwidth]{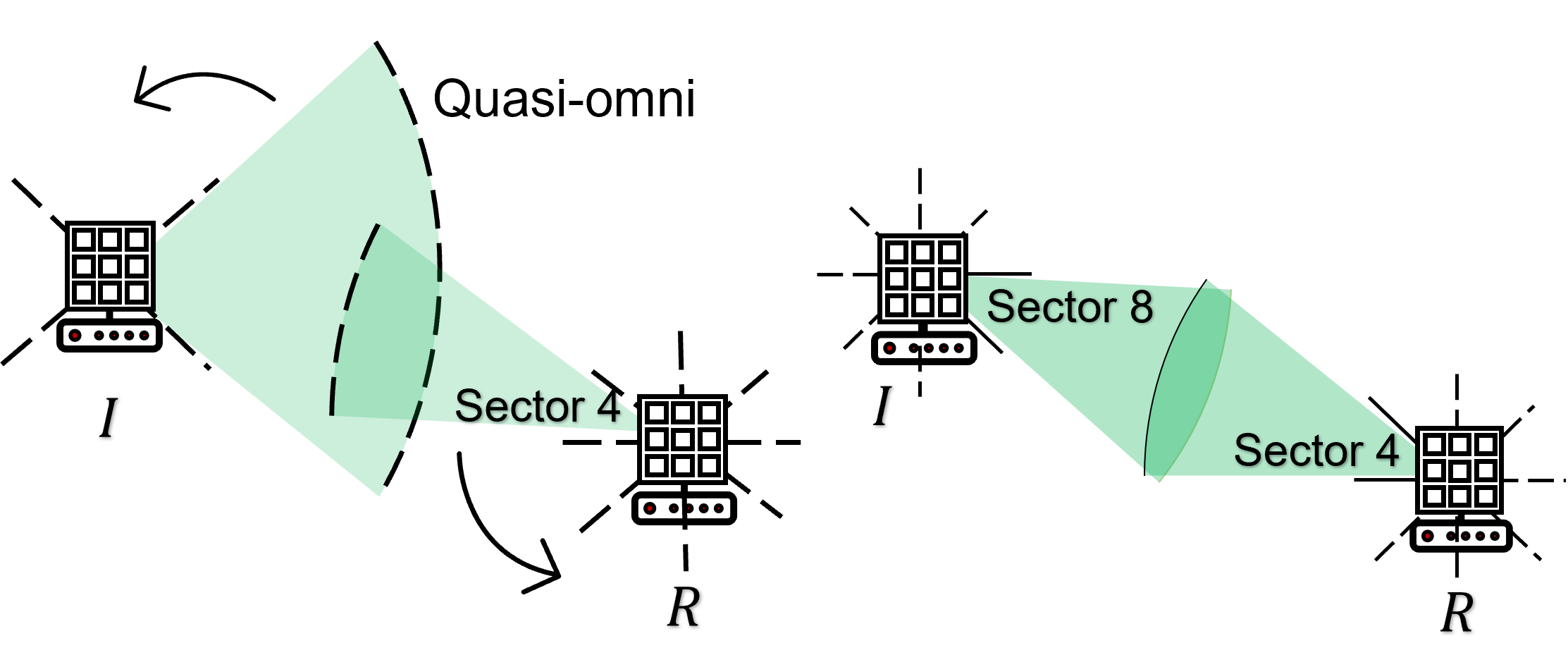} \\ 
  (c) \hspace{90pt} (d)
\end{tabular}
\vspace{-0.1in}
\caption{The SLS process of 802.11ad: (a) $I$ sweeps through the plane in fine-beam mode while $R$ receives in quasi-omni mode, (b) $I$ repeats the sector sweep for every quasi-omni pattern of $R$, (c) $R$ sweeps through the plane in fine-beam mode  while $I$ receives in quasi-omni mode, (d) The antennas of $I$ and $R$ are optimally aligned.}
\label{System model}
\vspace{-0.1in}
\end{figure}


\noindent
\textbf{\textit{Initiator sector sweep phase.}} 
\begin{enumerate}
\item{The plane is divided into $N$ sectors $S_1, S_2,\ldots S_N.$ The initiator $I$ sweeps through each $S_i$ and transmits sector sweep frame (SSF) $SSF_I(i)$, using the fine-beam mode.} 


\item{The responder $R$ measures the SNR for each received SSF, while in quasi-omni mode.}
\item{Steps 1 and 2 are repeated for every quasi-omni antenna orientation in $R$ to cover $R$'s plane. The responder records the sector $S^{\ast}_I$ with the highest SNR.}
\end{enumerate}
\noindent
\textbf{\textit{Responder sector sweep phase.}} 
\begin{enumerate}
\setcounter{enumi}{3}
\item{The responder switches to fine-beam mode whereas the initiator switches to quasi-omni mode. The responder $R$ sweeps through each sector $S_i$ and transmits $SSF_R(i)$, using the fine-beam mode. Further, the responder populates the sector sweep (SSW) feedback field of the SSF with the $I$'s optimal transmitting secctor  $S^{\ast}_I,$ as identified in the previous phase.}  
\item{The initiator records the sector $S^{\ast}_R$ with the highest received SNR and also learns $S^{\ast}_I$ from $R$'s feedback.}
\end{enumerate}
\noindent
\textbf{\textit{Acknowledgement phase.}} 
\begin{enumerate}
\setcounter{enumi}{5}
\item{The initiator sends a SSW-Feedback frame to $R$ indicating $R$'s optimal transmitting sector $S^{\ast}_R$. The feedback is transmitted in fine-beam mode using sector $S^{\ast}_I$, whereas the responder is in quasi-mode.}
\item{At the last step, the responder transmits a sector sweep ACK to the initiator using sector $S^{\ast}_R$.} 
\item{The two antennas are aligned at the $S^{\ast}_I$-$S^{\ast}_R$ direction.} 
\end{enumerate}

\begin{figure}[t]%
\centering
\setlength{\tabcolsep}{-3pt}
\begin{tabular}{cc}
  \includegraphics[width=0.5\columnwidth]{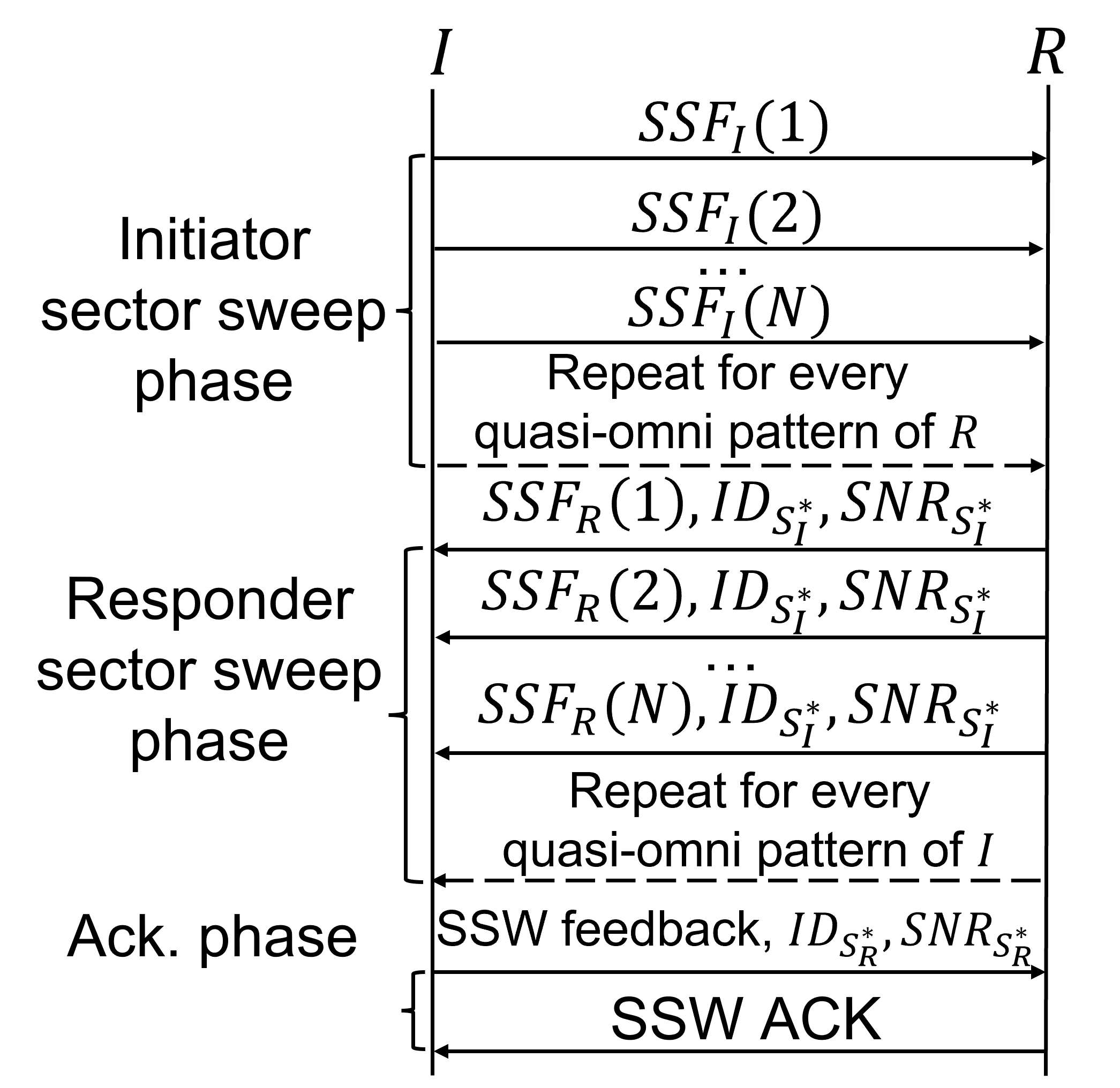} &
  \includegraphics[width=0.5\columnwidth]{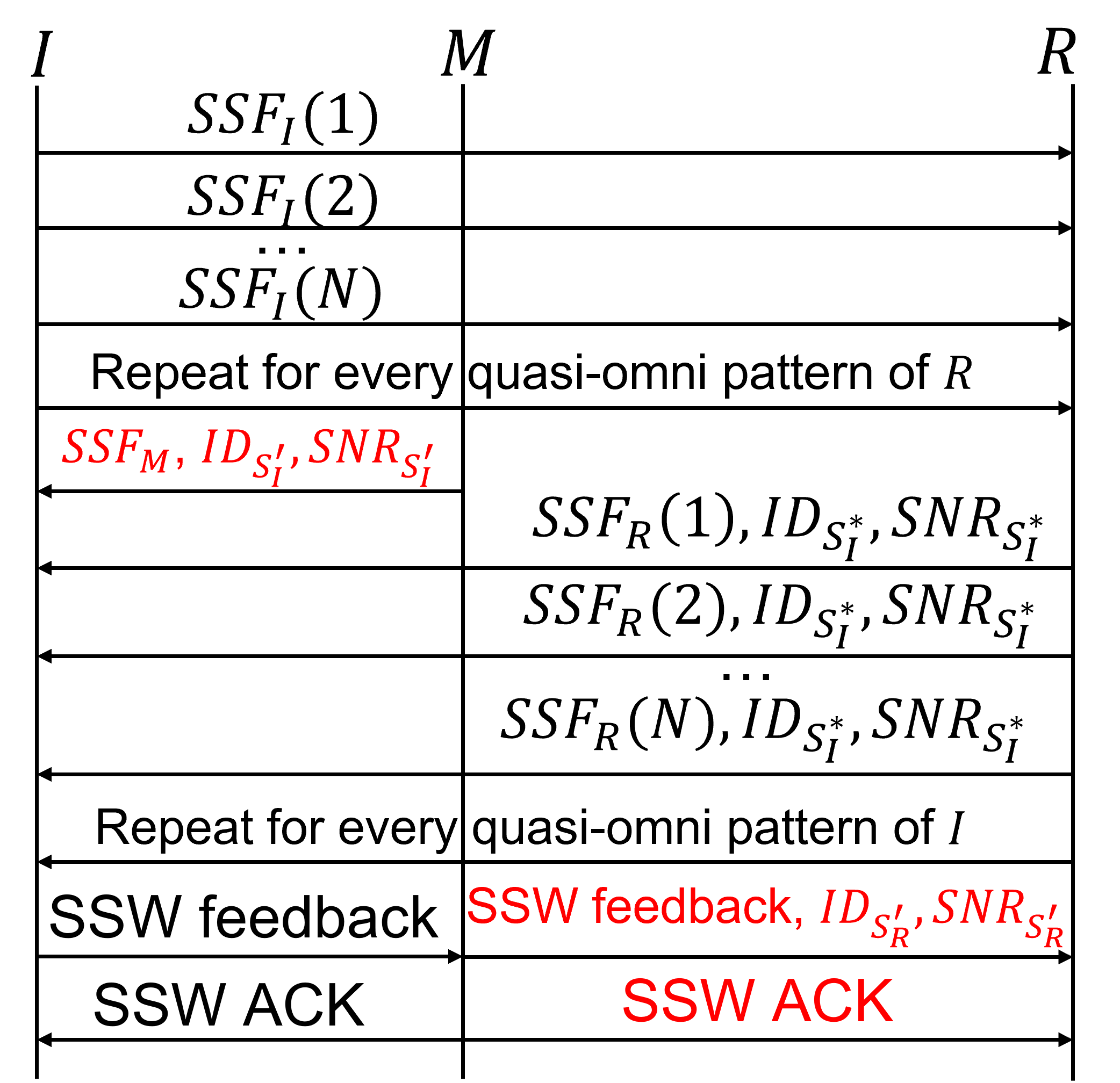} \\ 
  (a) & (b) 
\end{tabular}
\vspace{-0.1in}
\caption{(a) SLS protocol in IEEE 802.11ad standard, (b) beam-stealing attack against SLS protocol}
\label{802.11 and attack}
\vspace{-0.1in}
\end{figure}


In the example of Fig. \ref{System model}, the responder identifies $S_I^{\ast}  = S_8$ whereas the initiator identifies $S_R^{\ast}  = S_4$.  The two devices use the $S_8-S_4$ antenna orientation for communications.

\subsection{Beam-Stealing Attack Against the SLS Protocol}
\label{beam steal SLS}
In this section, we describe the beam-stealing attack against SLS protocol in IEEE 802.11ad presented by Steinmetzer et al. \cite{steinmetzer2018beam}. The attack is launched by $\mathcal{M}$ who forges the SSF frames. The timeline of the attack is shown in Fig.~\ref{802.11 and attack}(b). 

Mallory positions herself within range of both $I$ and $R$ (though the attack can be also launched by two colluding devices). During the initiator sector sweep phase, Mallory measures the SNR from all sectors of $I$ and chooses the one that yields the highest SNR, denoted by $S^\prime_I$. Because the communication protocol lacks an authentication mechanism, Mallory impersonates $R$ by forging $R$'s SSFs and indicating $S^\prime_I$ as the best sector. The same process is repeated during the responder sector sweep phase, where $\mathcal{M}$ chooses the best sector $S^\prime_R$ and informs $R$ through forged SSF frames on behalf of $I$. Ultimately, $I$ and $R$ choose sectors $S^\prime_I$ and $S^\prime_R$ that beamform towards $\mathcal{M}$, as illustrated in Figure~\ref{beam steal}. $I$ and $R$ choose sectors $S_1$ and $S_3$ instead of $S_8$ and $S_4$. Mallory gains full control over the communication between the victim devices, enabling a range of potential attacks such as eavesdropping, traffic analysis, data modification and injection.

\subsection{Authenticated Beam Sweeping} 

It is evident that the beam stealing attack is possible due to the  lack of authentication and integrity protection on the SSF frames. To counter this vulnerability, Steinmetzer et al. proposed an authenticated version of SLS where $I$ and $R$ are assumed to possess public/private key pairs \cite{steinmetzer2018authenticating} which are used to securely establish a session key $s$. The session key is used for mutual authentication and freshness to protect against message forgeries and replay attacks. Specifically, during the initiator sector sweep phase, $I$ generates a unique cryptographic nonce $v_I$ that is sent with every SSF frame. R receives the SSF frames, determines the best sector $S^{\ast}_I$, and computes an authenticator $\alpha_R$ of length $l_\alpha$ as 
\[
\alpha_R = \text{auth}_{l_{\alpha}}(S^{\ast}_I, v_I,s)=\text{trunc}(h(S^{\ast}_I,v_I,s),l_\alpha),
\]
where $\text{auth}_{l_{\alpha}}(S^{\ast}_I,v_I,s)$ computes an authenticator of size $l_\alpha$ as a truncated hash function of the concatenation of $S^{\ast}_I, v_I$ and $s$. In the responder sweep phase, $R$ generates a nonce $v_R$. The initiator computes its authenticator as $\alpha_I = \text{auth}_{l_{\alpha}}(S^{\ast}_R, v_R,s)$ for the optimally selected sector $S^{\ast}_R$. After both devices complete their sweeps, $I$ sends  $S^{\ast}_R$ and $\alpha_I$ in a feedback frame and $R$ acknowledges with $S^{\ast}_I$ and $\alpha_R$. Finally, both devices verify the received authenticators. If both verifications are successful, $I$ and $R$ set their sectors accordingly. Without access to the shared secret $s$, $M$ cannot create a forgery. Moreover, the nonces $v_I$ and $v_R$ guarantee freshness. 

\begin{figure}[t]
\centering
\includegraphics[width=0.7\linewidth]{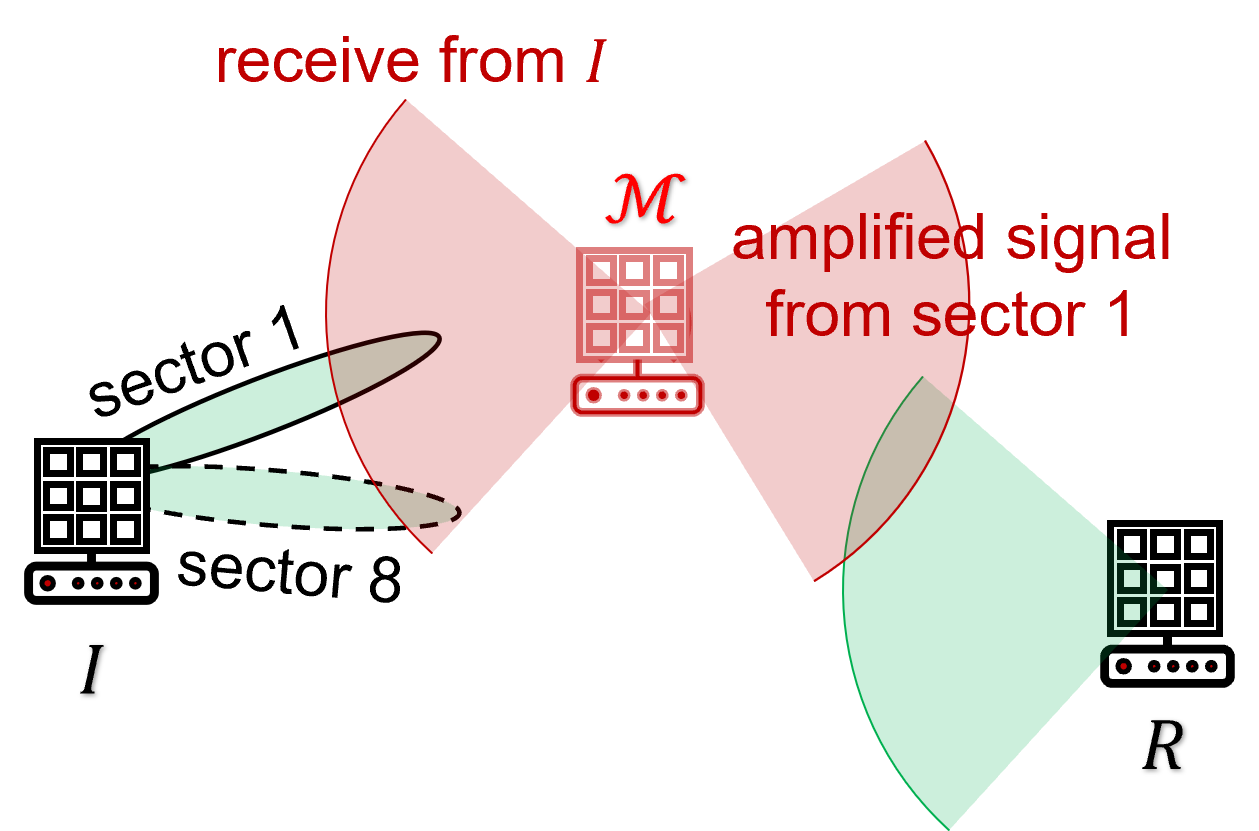}
\vspace{-0.1in}
\caption{Amplify-and-relay attack. $\mathcal{M}$ receives an authentic SSF while $I$ is scanning sector 1. $\mathcal{M}$  amplifies and relays the SSF without decoding. The responder authenticates the SSF and perceives sector 1 to yield the highest SNR. } 
\label{Beam steal attack}
\vspace{-0.20in}
\end{figure}

\section{Amplify-and-relay Beam-Stealing Attack}
\label{Amp-and-relay beam stealing}
The authenticated beam sweep protocol proposed in \cite{steinmetzer2018authenticating} is still vulnerable to beam stealing attacks. This is because the cryptographic protections cannot prevent Mallory from {\em forging the wireless environment} by simply amplifying authentic SSFs.  We present a new type of attack that we call {\em amplify-and-relay beam stealing} that allows Mallory to attract $I$'s and $R$'s beams without access to the shared secret $s.$

The attack operates as follows. Mallory positions herself at some location between $I$ and $M$ away from the LoS (we later examine Mallory's candidate attack locations). Mallory is equipped with two directional antennas (or a patch antenna that can operate two beams), which point towards the general direction of $I$ and $R$, respectively.  These antennas operate in quasi-omni mode to allow reception of $I$ and $R$'s transmissions, as Mallory may not know the exact position of the two devices and the wireless environment (e.g., blockage). During the initiator sweeping phase, Mallory receives an  authenticated SSF using the antenna pointing to $I$, amplifies it, and relays it intact with the antenna pointing to $R$.    The responder authenticates the SSF and records the highest SNR from $M$'s direction due to the applied amplification. The same amplification and relay is applied in the opposite direction during the responder sweep frame phase, leading $I$ to indicate $M$'s direction as the beast one for beam alignment.

We emphasize that although timely relay of the SSFs is helpful, it is not necessary. First, the SLS protocol does not specify stringent timing constraints between the transmission of successive SSFs at different sectors. Second, some SSFs are not ever received by the intended receiver ($I$ or $R$) due to antenna misalignment. Therefore, even if the relayed message arrives late, it is not a replay of the direct transmission and freshness is not violated.

Mallory can follow two amplification strategies when launching the beam stealing attack. The first is the \textit{fixed power (FP)} strategy in which $\mathcal{M}$ relays received signal(s) at a fixed transmit power, regardless of the received power. The closer distance of $\mathcal{M}$ to $I$ and $R$ (if $\mathcal{M}$ is in between) or a higher transmit power than the legitimate devices guarantees that $\mathcal{M}$'s relayed SSFs will be received  with the highest SNR. However, if Mallory transmits at fixed power for all incoming SSFs, the receiver will measure the same SNR for several transmit sectors leading to easy attack detection. To succeed in her attack, Mallory may relay at high power only one SSF that she received.  An alternate strategy is to apply \textit{fixed amplification (FA)} to any incoming SSF leading at varying received power at the receiver. The main benefit of the $FA$ strategy is that the received SNR is maximized when the transmitter ($I$ or $R$) is best aligned with $\mathcal{M}$ (due to the highest received power at $\mathcal{M}$) leading to a natural alignment of $I$ and $R$ with $\mathcal{M}.$

\begin{figure}[t]%
\centering
\setlength{\tabcolsep}{-3pt}
\begin{tabular}{c}
  \includegraphics[width=0.9\columnwidth]{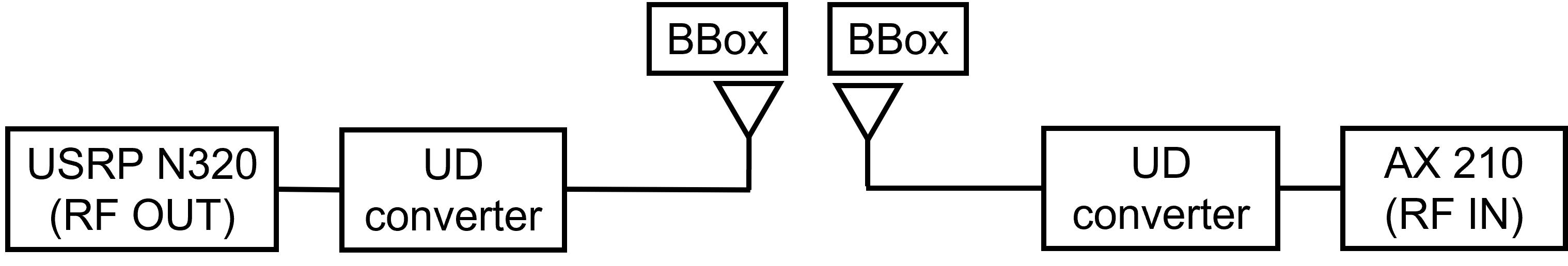} \\
  (a) RF chain of $I$ and $R$ \\
  \includegraphics[width=0.7\columnwidth]{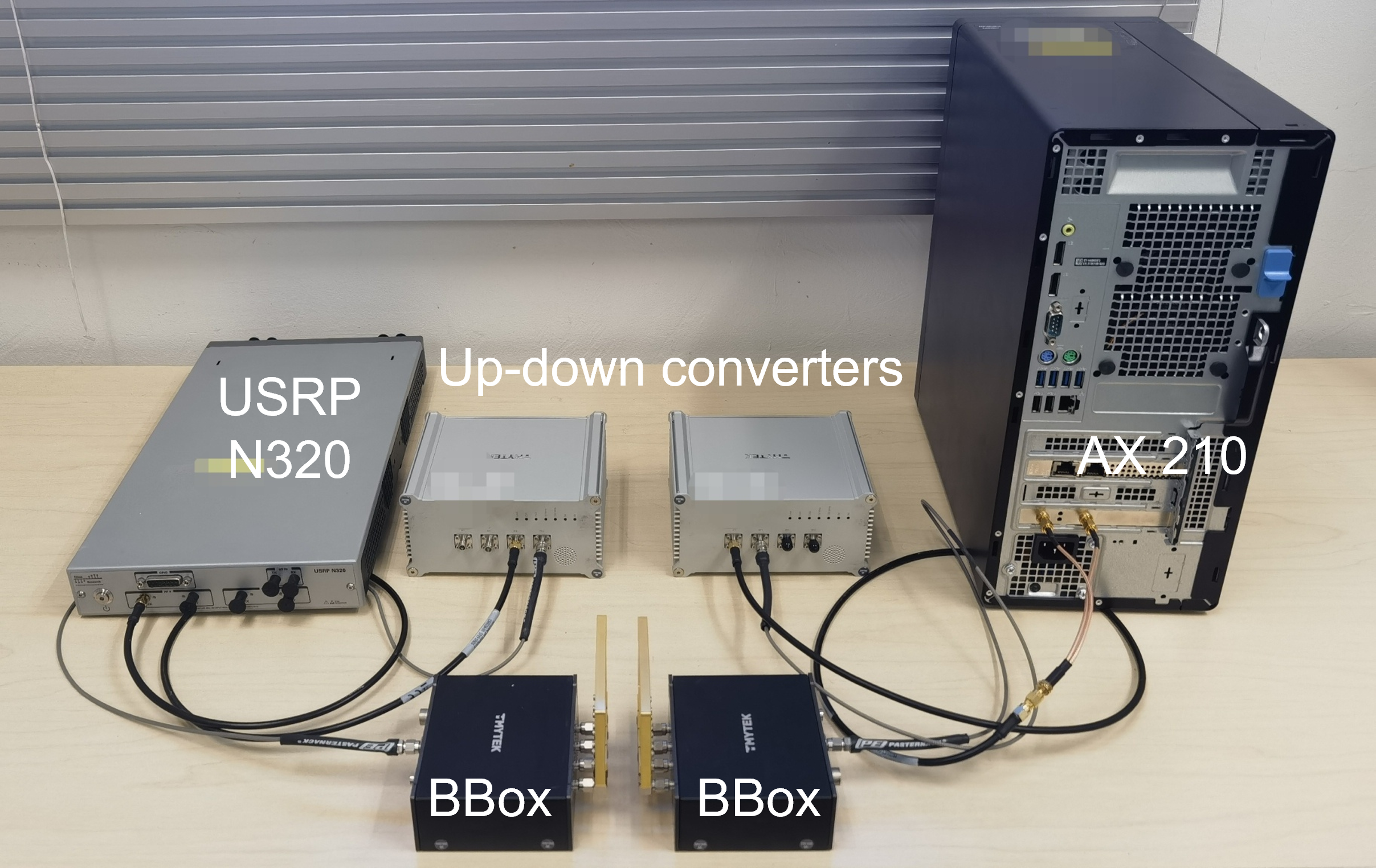} \\ 
  (b) Hardware used in the platform \\
  \includegraphics[width=0.7\columnwidth]{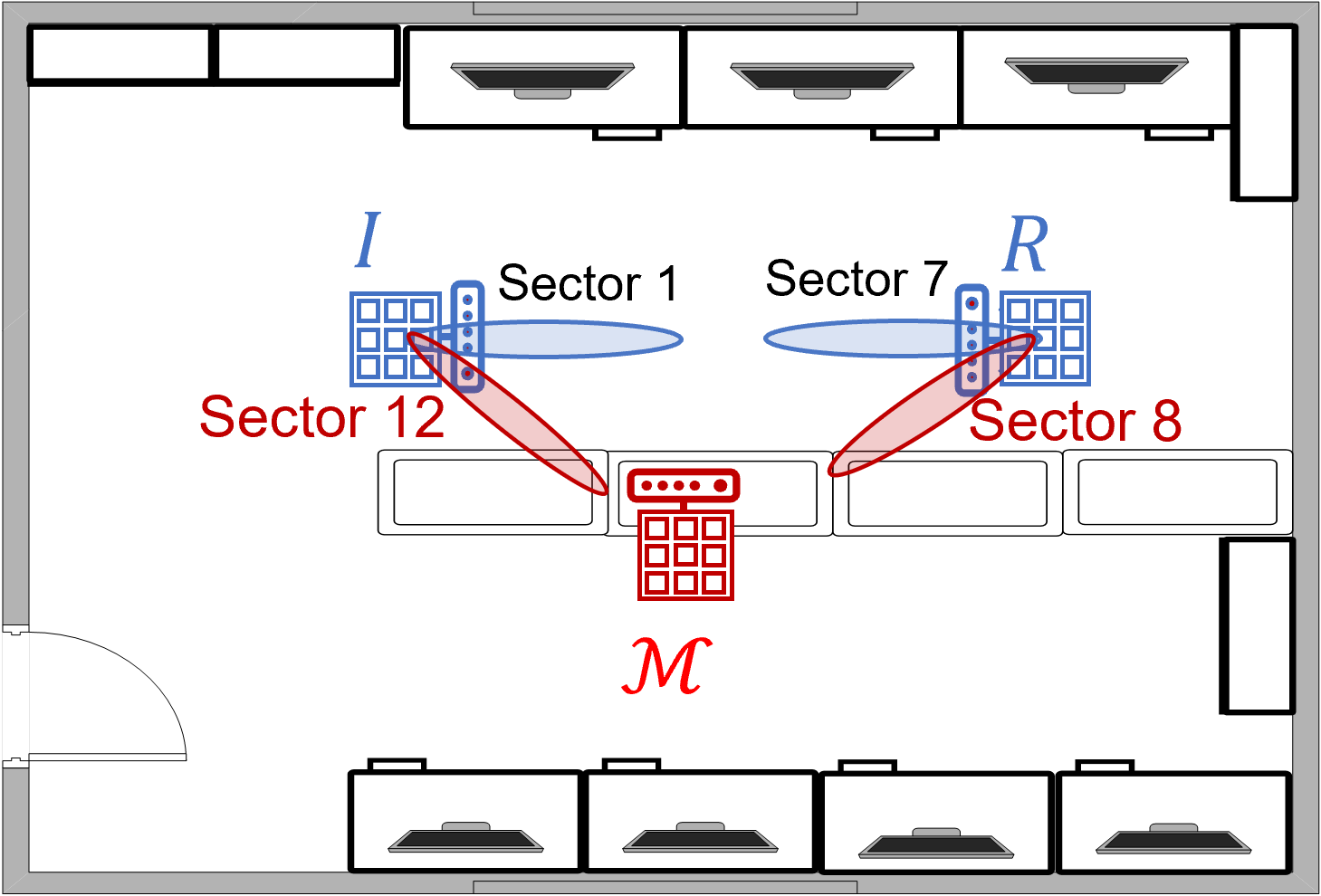} \\
  (c) Experimental topology
\end{tabular}
\caption{The mmWave experimental platform.}
\label{fig:platform}
\vspace{-0.20in}
\end{figure}

\begin{figure*}%
\centering
\setlength{\tabcolsep}{-6pt}
\begin{tabular}{cccc}
  \includegraphics[width=0.57\columnwidth]{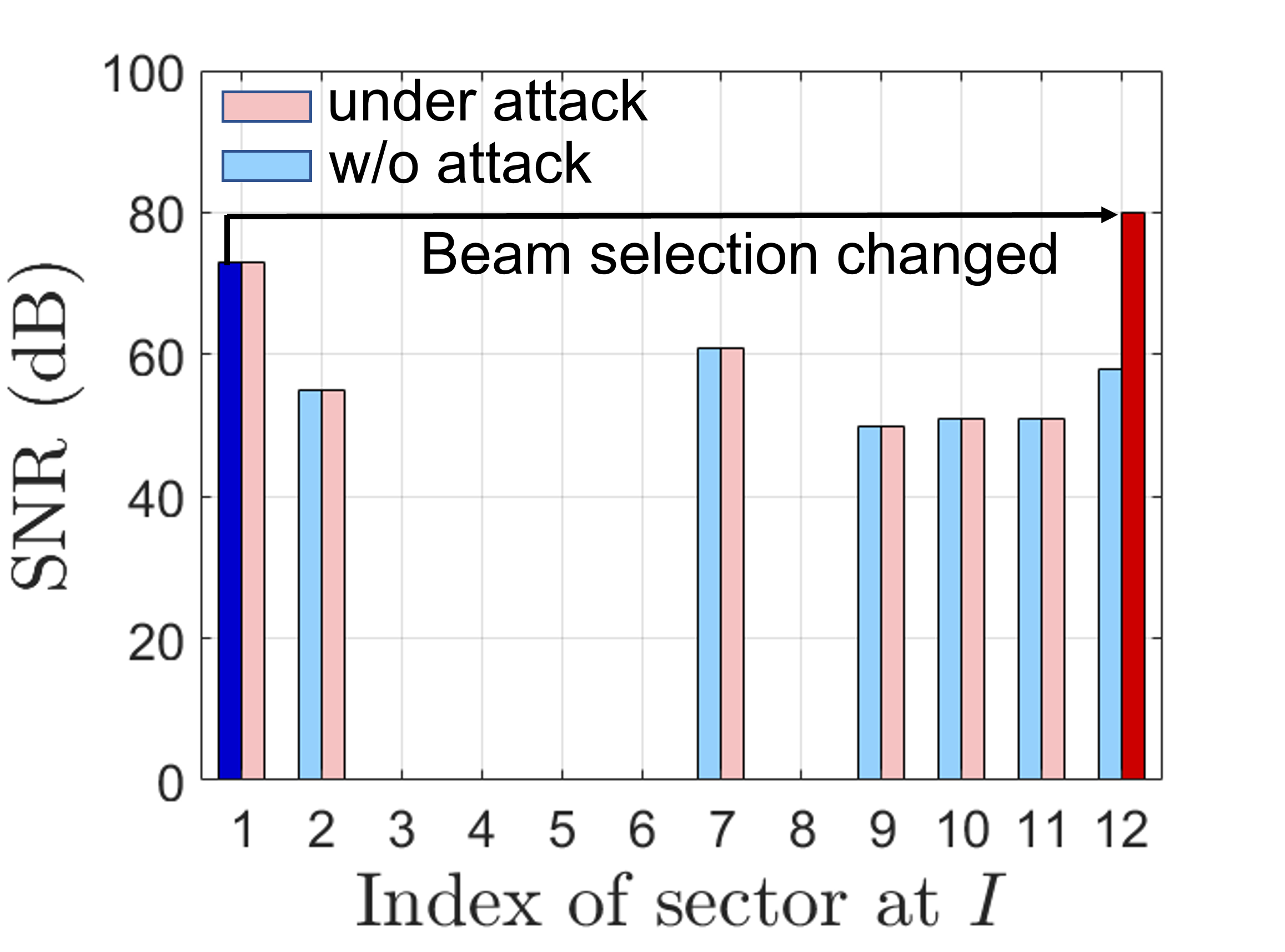} &
  \includegraphics[width=0.57\columnwidth]{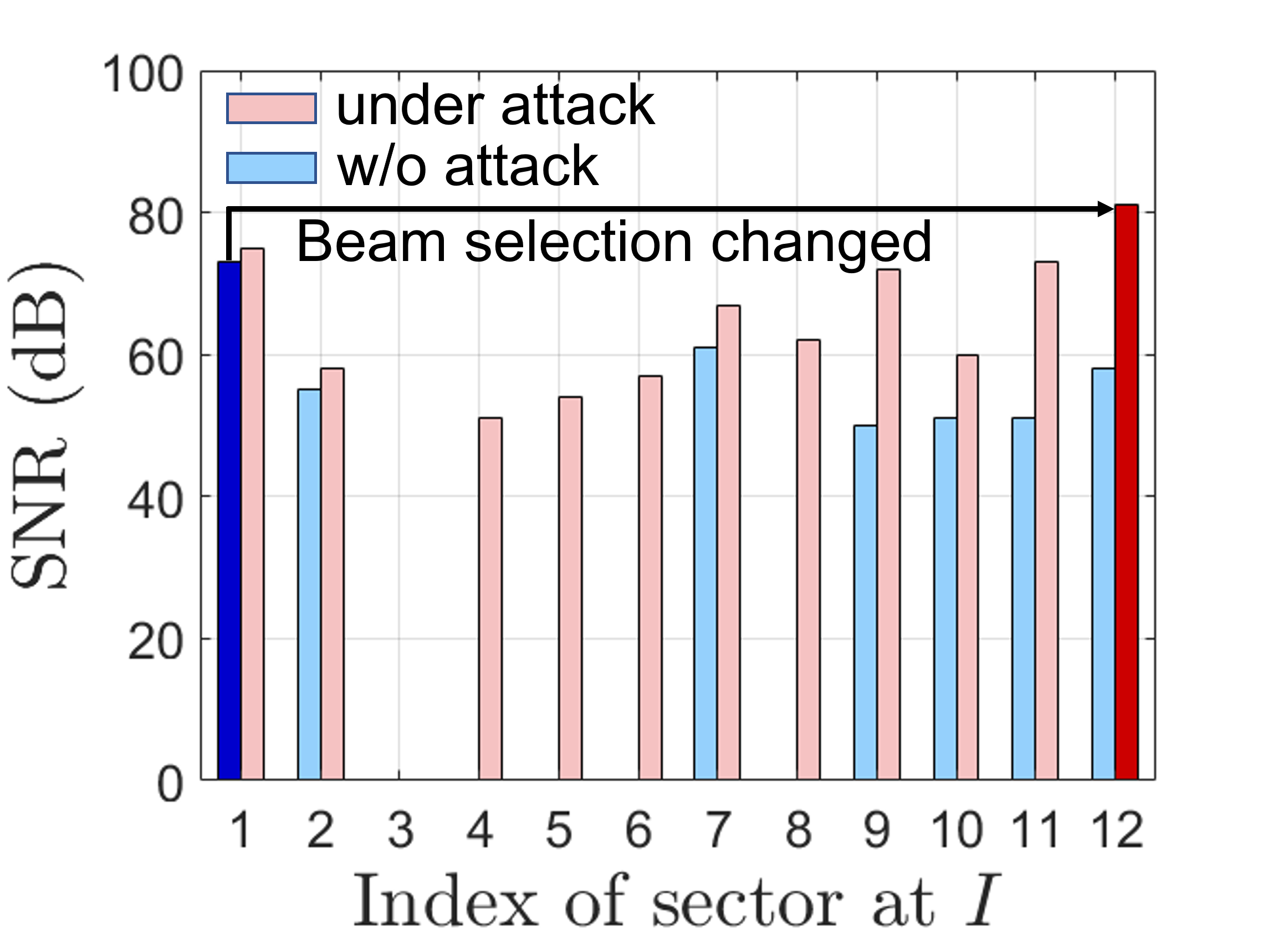} & 
  \includegraphics[width=0.57\columnwidth]{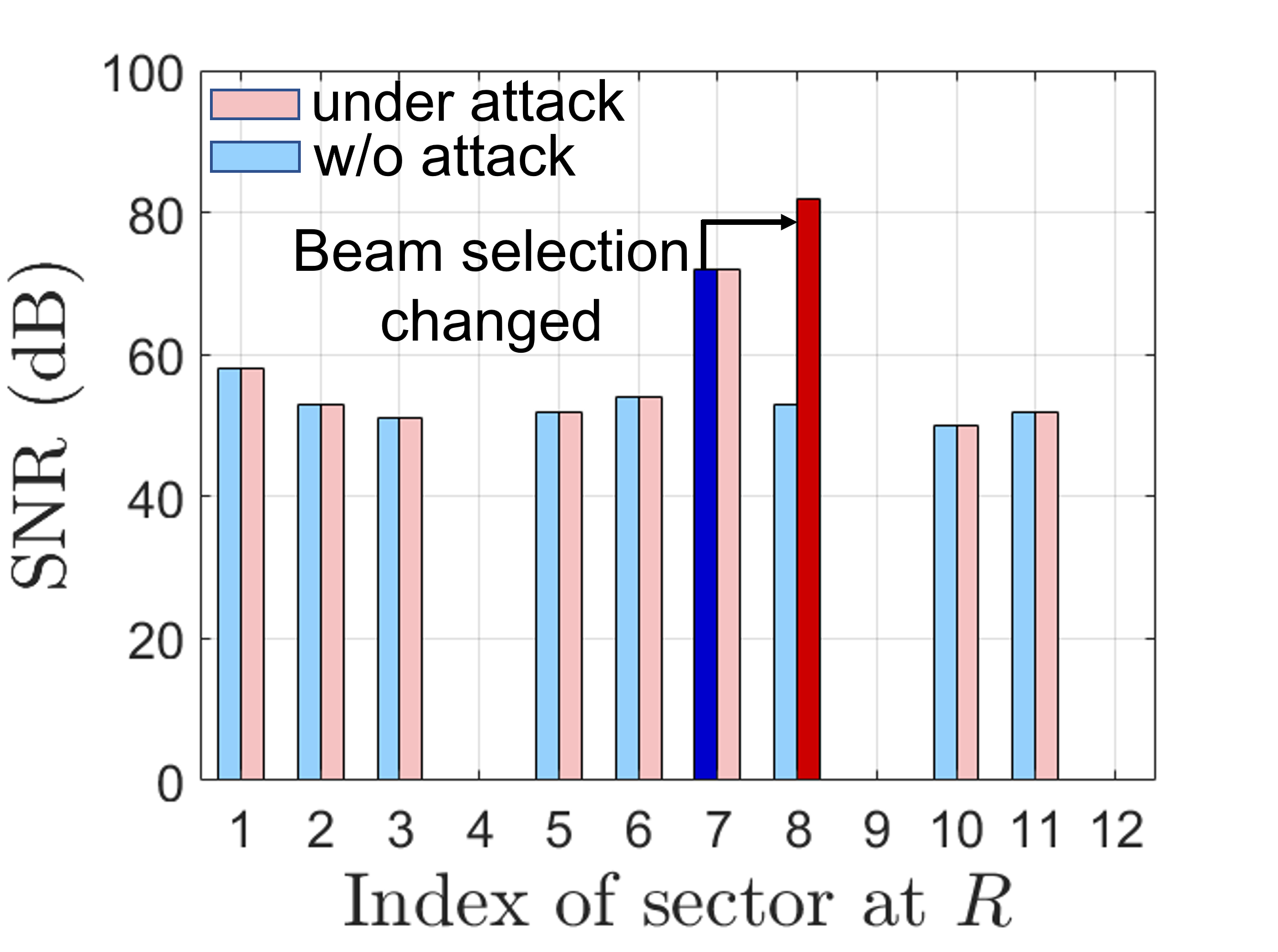} & 
  \includegraphics[width=0.57\columnwidth]{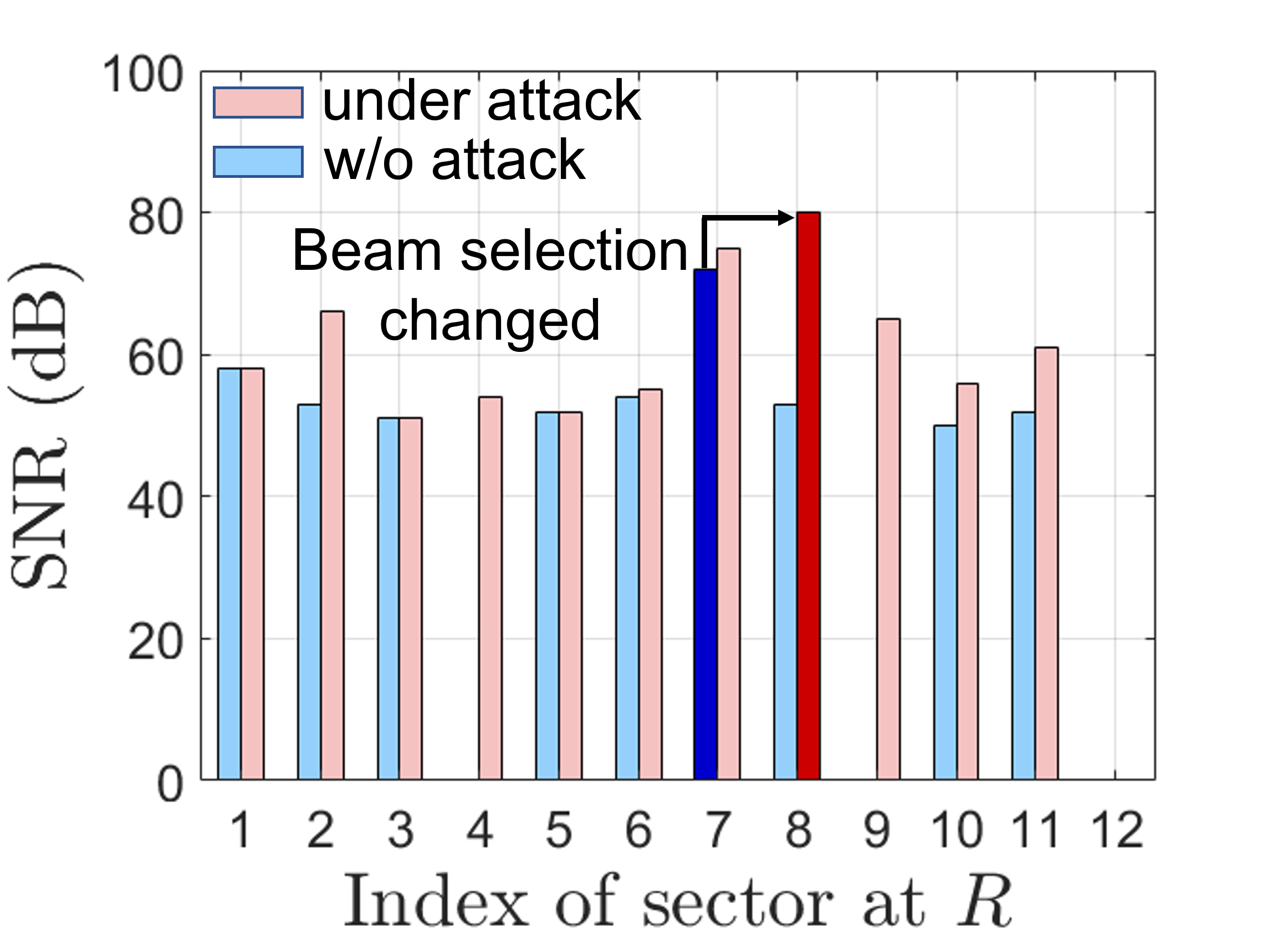} \\
  (a)  & (b) & (c) & (d)
\end{tabular}
\caption{(a)Amplify-and-relay beam-stealing attack on the initiator sector sweep using $FP$ strategy, (b) attack on the initiator sector sweep using $FA$ strategy, (c) attack on the responder sector sweep using $FP$ strategy, (d) attack on the responder sector sweep using $FA$ strategy.}
\label{fig:setup2 and result}
\vspace{-0.2in}
\end{figure*}

\subsection{Demonstration of the Amplify-and-Relay Attack}
\label{feasibility amplify attack}
In this section, we experimentally demonstrate the amplify-and-relay attack. 

{\bf Experimental Setup.} Our experimental setup consisted of USRPs and Wi-Fi cards connected to the TMYTEK mmWave platform. Specifically, we implemented the baseband functionality of the initiator and Mallory on two Ettus USRP N320s \cite{N320}. Each USRP was connected to a TMYTEK up-down (UD) converter \cite{UDconverter} that upscaled the signal from the sub-6GHz to a center frequency of 28GHz with 160MHz bandwidth. Transmission took place via a directional BBox lite patch antenna with 16 elements. The antenna supports electronically steerable sectors with half-power beamwidth of $30^\circ$, covering the plane with 12 non-overlapping sectors. The maximum transmit power of the USRPs was 30dBm.  The platform configuration is shown in Fig.~\ref{fig:platform}.

The responder was implemented with an AX 210 Wi-Fi card to allow for the extraction of CSI measurements by the PicoScence CSI tool \cite{jiang2021eliminating}. The signal was received using the BBox lite antenna and downconverted to the sub-6 GHz band using the UD converter. The card extract CSI measurements and computed the RSS. Moreover, we used the  Wi-Fi card to measure the noise floor, which remained consistent at -128dBm throughout all the experiments. All hardware was controlled by a desktop computer and synchronized using the Ettus OctoClock-G module.  

The experimental topology is shown in Fig.~\ref{fig:setup2 and result}(a). We positioned all three devices in a laboratory of rectangular shape with typical office furniture (bookcases, desks with workstations, whiteboards, etc.) $I$ and $R$ were placed 4.3m apart whereas $\mathcal{M}$ was placed in between (but not in the LoS) at distances 2.4m and 2.7m from $I$ and $R$, respectively. To mimic the SLS protocol in IEEE 802.11ad as we described in Section~\ref{sector sweep}, we fixed the $R$'s beam and employed the TMXLAB Kit to electronically steer $I$'s beam. The transmitter conducted a $360^\circ$ sweep of the entire plane. Upon completion of one sweep, $R$ switched to the next fixed receive direction while $I$ repeated the sweep. This sector sweep process was repeated in the other direction (with $I$ and $R$ switching roles) to determine the optimal beam of $R$. 

During a fixed power attack, the legitimate devices and Mallory transmitted at 10dBm, whereas during a fixed amplification attack, $\mathcal{M}$ amplified its received signal by 70dB. Mallory was synchronized with the initiator and the responder via the the Ettus OctoClock-G module and transmitted the same SSF frame at the same time as the initiator or the responder (depending on the round).

{\bf Attack Evaluation.} In Fig.~\ref{fig:setup2 and result}, we show the highest measured SNR at the responder during the initiator sweep phase. Figure~\ref{fig:setup2 and result}(a) corresponds to a fixed power attack. Without an attack, the responder measures the highest SNR when the initiator transmits at sector 1. During the attack, Mallory  only amplified beam sector 12 that is best aligned towards her. Therefore, the responder measures the same SNR in all other sectors except sector 12, where the SNR is about 10dB higher than that of sector 1. This results in informing $I$ that its best transmitting sector is sector 12, thus beamforming towards $\mathcal{M}.$ Figure~\ref{fig:setup2 and result}(b) shows the fixed amplification attack case that causes the same optimal sector change as the $FP$ attack. One main difference is that during an $FA$ attack, all sectors are amplified at the same gain thus causing an elevated SNR at $R$ compared to the absence of the attack. 

Figures~\ref{fig:setup2 and result}(c) and \ref{fig:setup2 and result}(d) show the $FP$ and $FA$ attacks launched during the responder sector sweep phase. In both attacks, the beam selection shifts from sector 7 to sector 8 pointing to $\mathcal{M}$ instead of the LoS. {\em We note that when $I$ and $R$ use sectors 12 and 8 respectively, they can no longer directly communicate.} The only way to communicate is via $\mathcal{M}.$ To verify the generality of the attack with respect to the adversary positioning, we repeated the experiments with $\mathcal{M}$ placed behind the initiator. The results of the second setup verify the attack feasibility and are presented in Appendix \ref{setup1}.

\section{A baseline Defense Method}

We first propose a baseline method for detecting amplify-and-relay attacks that exploits the invariable physical properties of propagation time. The main idea is to use the longer signal propagation and processing times of the relay path over the legitimate path. This method is specifically applicable to the mmWave environment where multipath is limited.  We use power delay profile (PDP) analysis to measure the timing of various paths and correlate it to the received power. As the mmWave signal attenuates rapidly with distance, it is expected that the LoS path is the strongest and takes the least delay to propagate to the receiver. Detecting high SNR at a path with higher delay is indicative of an amplify-and-relay attack, provided that the LoS path (or any other shorter path) is unobstructed.  

The PDP can be derived from the receiver's channel impulse response (CIR) \cite{rappaport2010wireless}. Let $X(f, t)$ and $Y(f, t)$ be the transmitted and received signals in the frequency domain and $H(f, t)$ be the CSI. Then $H(f, t)$ can be expressed as \[  
    H(f, t)= \frac{Y(f, t)}{X(f, t)} 
    = e^{-j2 \pi \Delta ft} \sum_{k=1}^{N}a_k(f, t)e^{-j2 \pi f \tau_k(t)},
\]
where $e^{-j2 \pi \Delta ft}$ is the phase shift caused by the carrier frequency difference $\Delta f$ between $I$ and $R$, $N$ is the number of paths, $e^{-j2 \pi f \tau_k(t)}$ is the the phase shift caused by the delay $\tau_k(t)$ of the $k^{th}$ multipath, and $a_k(f, t)$ is the amplitude of the $k^{th}$ multipath. Then the CIR, $h(t, \tau)$ can be calculated by applying an IFFT to  the CSI,\[  
    h(t, \tau)=\sum_{k=1}^{N}a_k(t, \tau)e^{-j \theta_k(t, \tau)}\delta[\tau-\tau_k(t)]
    \label{CIR_eq}, \]
where $e^{-j \theta_k(t, \tau)}$ is the phase shift and $\delta [\cdot]$ is a delta function. If choosing a time period equal or smaller than the channel's coherence time, the CIR can be considered to be constant. Then the amplitude of the CIR, $|h(\tau)|$, can be extracted as,\[
    |h(\tau)|=\sum_{k=1}^{N}a_k(\tau)\delta(\tau-\tau_k),
    \label{CIR amplitude}
\]
and the PDP is $|h(\tau)|^2$. Instead of using the PDP of each transmit sector directly, we apply an additional step to calculate the {\em collective power delay profile (CPDP)}. Specifically,  we extract the amplitude and time delay of {\em the highest amplitude path}  from the PDP of each SSF transmission and combine them in a single CPDP. This combination is possible only if time can be measured at high accuracy (which is possible at mmWave due to the available bandwidth) and if SSFs are transmitted periodically at precisely known times. 
In benign scenarios, the first peak of the CPDP corresponds to the LoS path between legitimate users. The signal strength received from this path is expected to be the strongest among all possible sectors. If a path that appears later has a stronger RSS, an attack is detected. Note that any reflections in mmWave not only travel longer distances, but also experience high attenuation due to the lower reflection coefficients as the small wavelength is comparable to surface imperfections.

\begin{figure}[t]
\centering
\includegraphics[width=0.5\linewidth]{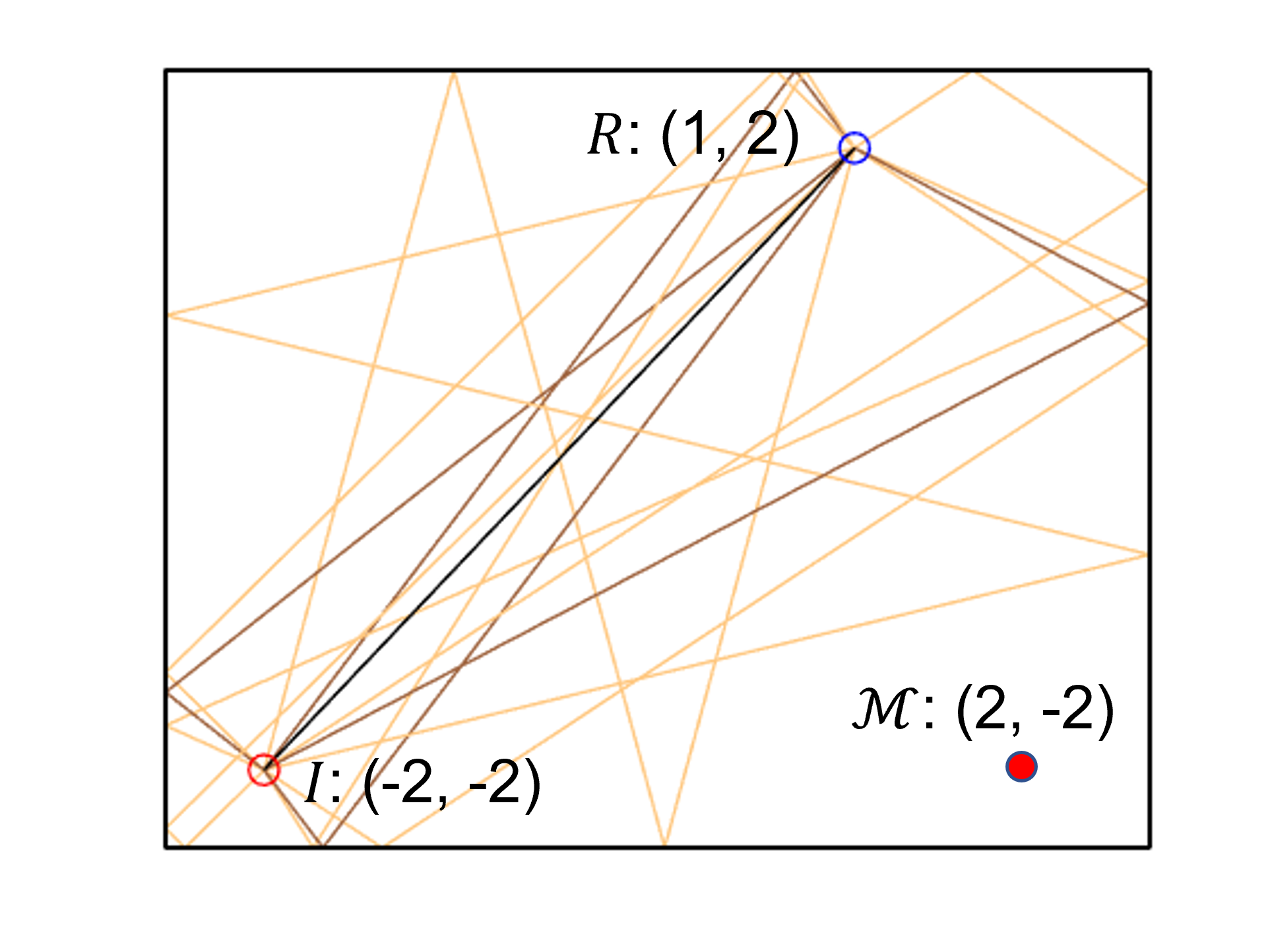}
\vspace{-0.1in}
\caption{Topology.} 
\label{cpdp layout}
\vspace{-0.20in}
\end{figure}

\begin{figure}[t]%
\centering
\setlength{\tabcolsep}{-4.5pt}
\begin{tabular}{cc}
  \includegraphics[width=0.55\columnwidth]{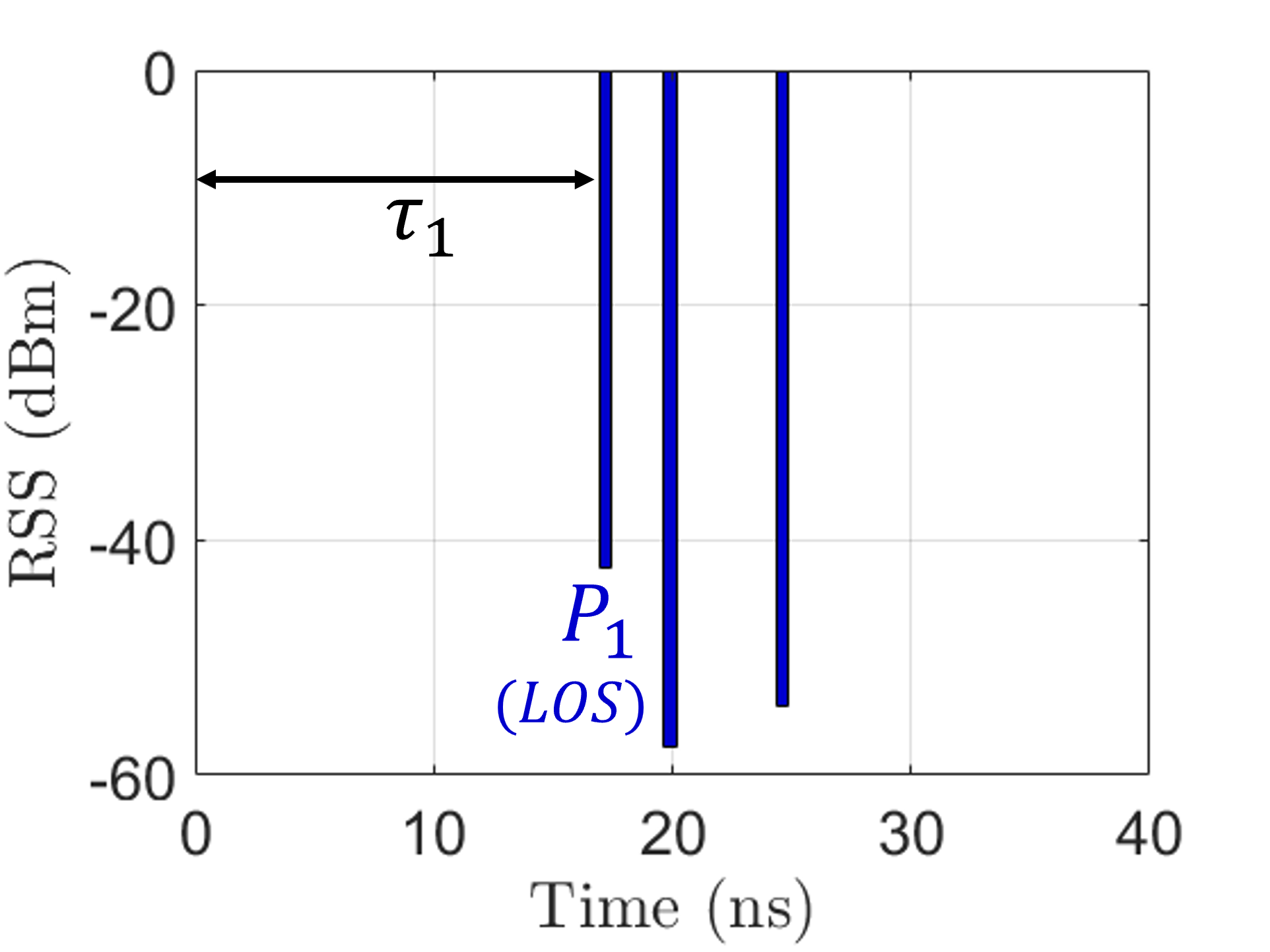} &
  \includegraphics[width=0.55\columnwidth]{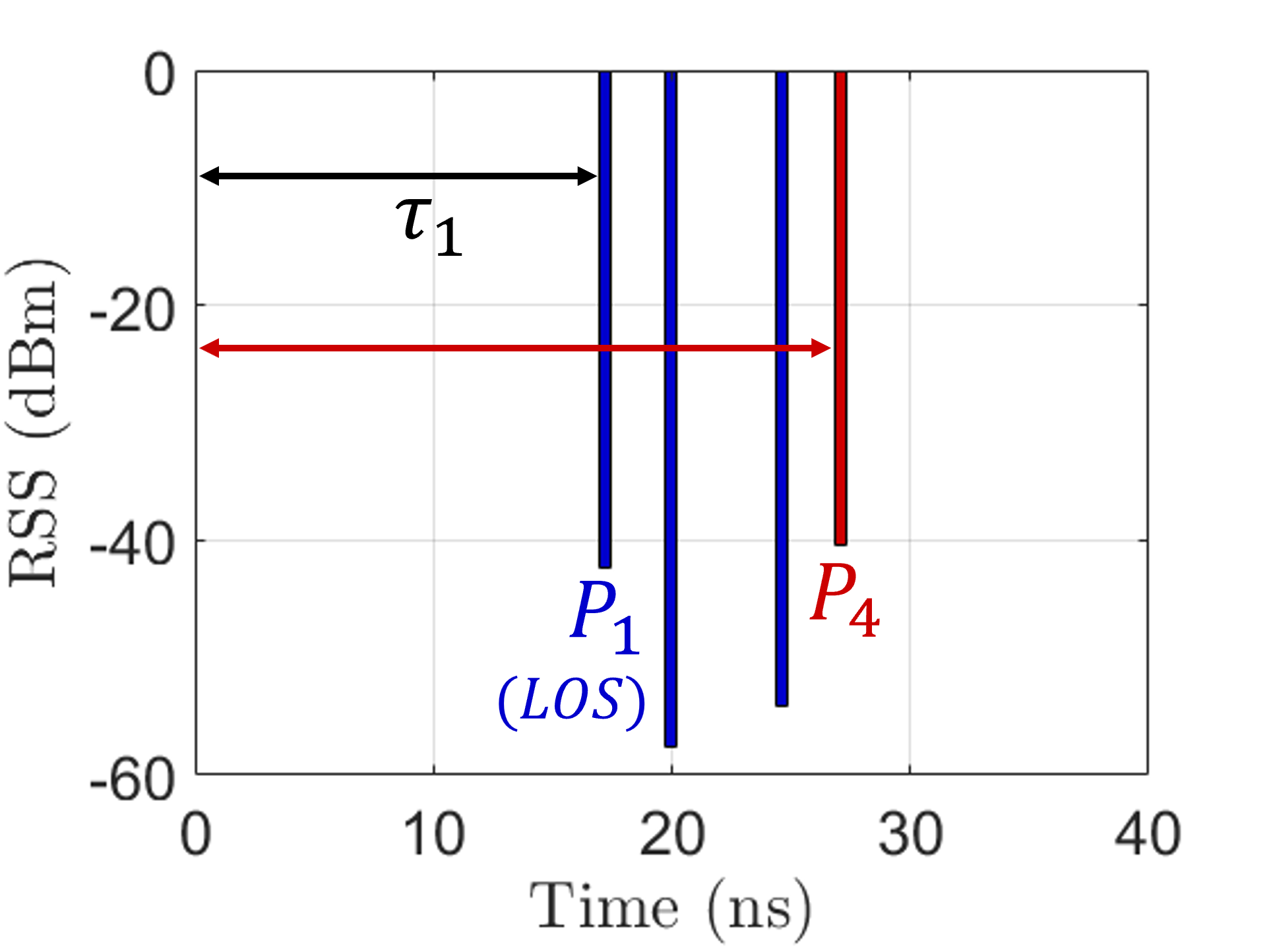} \\ 
  (a) & (b) 
\end{tabular}
\caption{CPDP: (a) without attack, the first peak is the strongest. (b) with a beam-stealing attack, the fourth peak is the strongest.}
\label{pdp w and w/ attack}
\vspace{-0.20in}
\end{figure}

{\bf Numerical Examples.} We utilized the ray-tracing-based simulator mmTrace \cite{steinmetzer2016mmtrace} to provide an illustration of how to employ CPDP to detect relay attacks in the initiator sector sweep. The topology  under consideration is depicted in Fig.~\ref{cpdp layout}. The distances between devices are $d_{IR}=5m$, $d_{MR}=4.1m$ and $d_{IM}=4m$. The initiator ($I$) is configured with $N=32$ sectors with a HPBW of $12^\circ$, while the responder ($R$) is equipped with six non-overlapping quasi-omni patterns with an HPBW of $60^\circ$. $I$ conducts an SLS sweep for each quasi-omni pattern, while $R$ measures the PDP for each transmit sector and extracts the highest amplitude path from each PDPs to create CPDP.  If $R$ observes multiple peaks with the same delay in different quasi-omni patterns, he chooses the one with the highest RSS. 

The results of CPDP without attack is illustrated in Fig. \ref{pdp w and w/ attack}(a). It is observed that the signals arriving later are weaker than the first peak as they travel along longer paths. However, in the presence of a beam-stealing attack (Fig. \ref{pdp w and w/ attack}(b)), the fourth peak is stronger than earlier peaks, even though it comes from a longer path. This is due to the adversary amplifying and relaying the signal on that path. In such a scenario, $R$ can detect the attack and reject the signal from $I$. Note that our results do not include any processing delays due to relaying. 

It should be noted that Mallory could avoid detection if it lies in the LoS path. However, the beam-stealing attack is not meaningful because the beams of $I$ and $R$ will lie in the optimal path and $I$ and $R$ can directly communicate unless $\mathcal{M}$ introduces physical blockage. Moreover, this method does assume the existence of a faster path (e.g. LoS) than the adversary's relay. Finally, the stringent synchronization requirement for computing the CPDP may require custom hardware \cite{samimi201628}. To address these limitations, we propose an alternative protocol based on path loss and AoA.

\begin{figure}[t]%
\centering
\begin{tabular}{c}
\includegraphics[width=0.9\columnwidth]{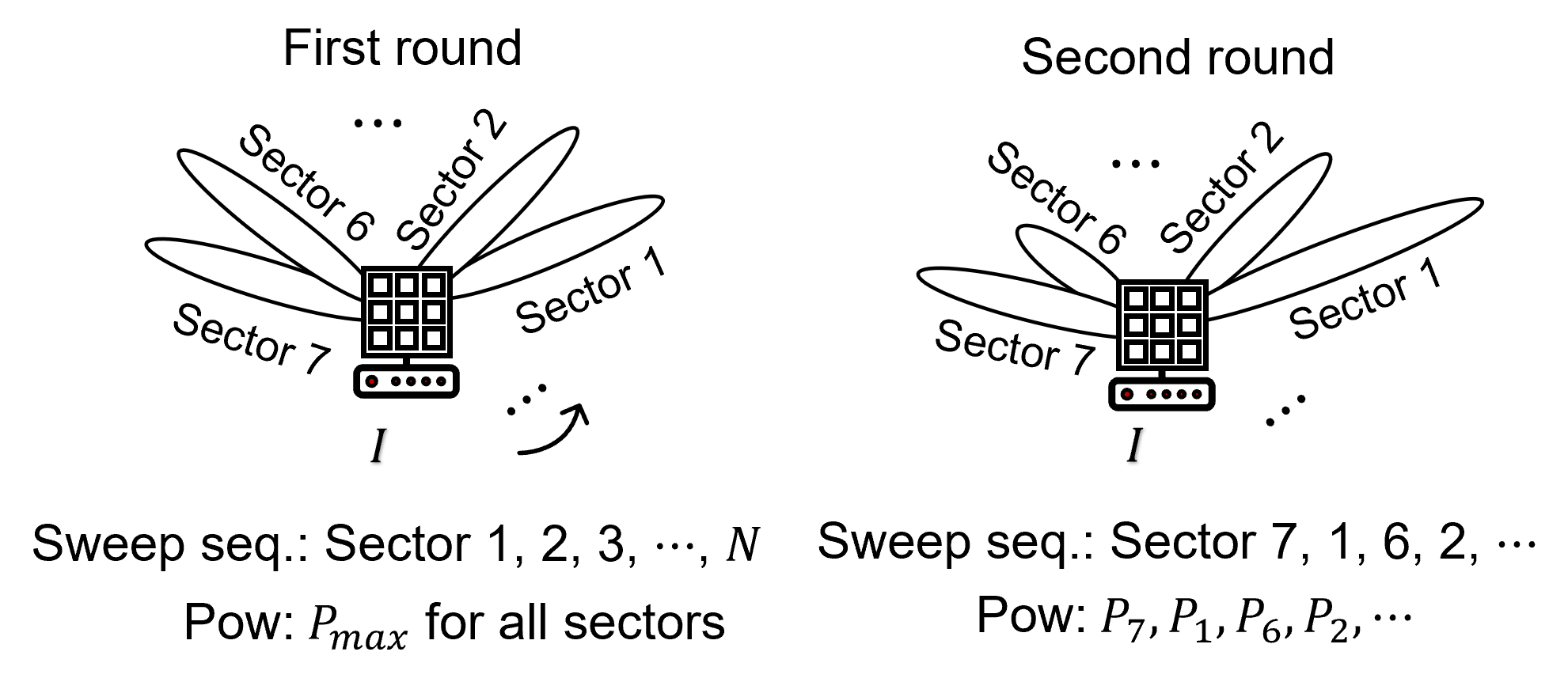} \\
(a) Physical layer commitment\\
\includegraphics[width=0.9\columnwidth]{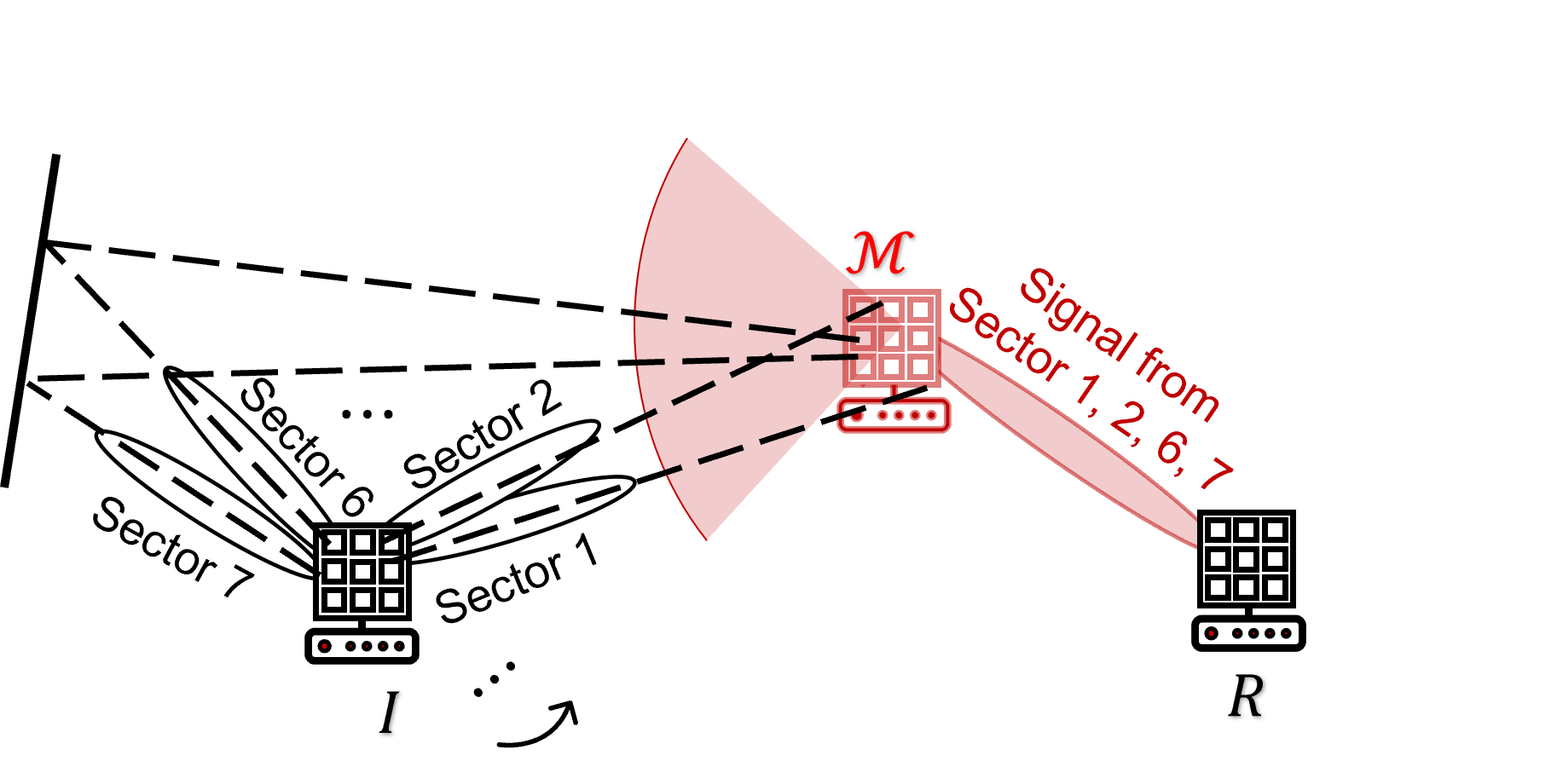} \\ 
(b) Coarse AoA detection
\end{tabular}
\caption{Overview of the SecBeam protocol: (a) the two sector sweep rounds of the SecBeam protocol and (b) when $\mathcal{M}$ relays all sectors it hears, $R$ receives multiple SSFs over the same receiving sector, indicating the same coarse AoA.}
\vspace{-0.2in}
\label{protocol demo}
\end{figure}

\section{The SecBeam Protocol}

We propose SecBeam, a secure beam sweeping protocol that detects amplify-and-relay attacks. To thwart fixed power amplification attacks, SecBeam implements a {\em physical-layer commitment scheme} were a transmitter commits to use the same power level over successive transmission rounds. Moreover, SecBeam incorporates coarse angle-of-arrival (AoA) information to detect fixed amplification attacks. To motivate the SecBeam design, we first provide an overview of the focal ideas and follow with a detailed protocol description.

\subsection{Overview}

The main challenges in detecting an amplify-and-relay attack are that (a) cryptographic operators are not violated during a signal relay and (b) $I$ and $R$ do not have any prior knowledge of their relative positions and of the wireless environment.  Without this information, a receiver cannot detect if a physical property such as SNR has been altered. 

{\bf Physical layer commitments.} To counter fixed power attacks, we construct a physical-layer security primitive that resembles a commitment scheme. We extend the SLS protocol to operate in two rounds. During the first round, $I$ sweeps across each beam sector $S_i$ sequentially using maximum power $P_{\max}$. The responder $R$ calculates the path loss at round 1 as $PL_i(1)= P^T_i(1)-P^R_{i}(1)$ by subtracting the received power from the transmit power (in dB) using $P^T_i(1)= P_{\max}, \forall i.$ This round serves as {\em a physical commitment} to the wireless environment even if this environment is distorted by Mallory. 

During the second round,$I$ repeats the sector sweep with two randomizations in place. First, $I$ randomizes the sector sweeping sequence so it is no longer sequential. Second, $I$ randomizes the transmit power on every sector. The sector ID $S_i$ and the transmit power $P^T_i(2)$ are included in the SSF (which is encrypted). The responder calculates the path loss at round 2  and verify that $PL_i(1) \approx PL_i(2),\forall i.$ This round serves as {\em a physical decommitment},  where the transmit power $P^T_i(2)$ and sector ID $S_i$ are used to physically verify the commitment of the first round. An example of the two-round sector sweep is shown in Fig.~\ref{protocol demo}(a). 

The power and sweeping sequence randomization have two effects on Mallory's attack. If Mallory amplifies both rounds at fixed power, the path loss computed by $R$ would be different between the two rounds, thus revealing the attack. Mallory could attempt to detect the potentially lower received power from $I$. However, the randomization of the sweeping sequence  and encryption of the sector ID does not allow Mallory to infer which sector it is receiving from at each round.  Since Mallory does not know which sector in the second round corresponds to the one she relayed in the first round she could not fine tune its transmit power to make the path loss consistent across  both rounds.  The only strategy that is viable is {\bf to amplify all sectors} overheard at $\mathcal{M}$ with fixed amplification, that is, to launch a fixed amplification attack. 

{\bf Coarse AoA detection.} To counter fixed amplification attacks, we exploit another physical layer phenomenon that arises in the presence of a relay.  Because Mallory is forced to apply a fixed amplification to all sectors (SSFs) it hears, an unusually high number of sectors will appear to arrive from the same direction. That is, although the initiator is sweeping the plane, the signal arrives to the responder from the same direction, thus indicating the presence of a relay. In the example of Fig.~\ref{protocol demo}(b), sectors 1, 2, 6, and 7 all appear to be received at the same sector of $R.$ To extract AoA information without requiring the application of complex AoA estimators such as MUSIC \cite{gupta2015music} during the sweeping process, $R$ uses the quasi-omni mode to roughly estimate the incoming direction of a received signal. As the quasi modes are switched, $R$ can calculate for how many transmitting sectors each quasi-omni receiving sector appears to be the optimal. It can then use those statistics for detecting relays.

\subsection{Protocol details}

We now present the detailed steps of the SecBeam protocol. The protocol is facilitated by a shared secret $s$ established between $I$ and $R$ using public keys, as described in \cite{steinmetzer2018authenticating}.

\medskip

\noindent
\textbf{\textit{First round initiator sector sweep.}} 

\begin{enumerate}

    \item The initiator and the responder derive keys $k_1$ and $k_2$ from the shared common secret $s.$ Key $k_1$ is used for integrity protection and $k_2$ is used for encryption. 

    \item The initiator sweeps through sectors $S_1, S_2,\ldots S_N$  in this order and transmits sector sweep frame $SSF_i^{I}(1)$ at each $S_i$, using the fine-beam mode. Each $SSF_i^{I}(1)$ contains an authenticator  $\alpha_i(1)=\text{MAC}(S_i, P_i^T(1), v_i(1), k_1)$, where MAC is a message authentication code function, $P_i^T(1)$ is the transmission power of $SSF_i^{I}(1)$ and $v_i(1)$ is a nonce. The initiator attaches $\alpha_i(1),P_i^T(1),v_i(1)$ to $SSF_i^{I}(1)$, encrypts it with $k_2$ and transmits it at power $P_i^T(1) = P_{\max}.$  
 
    \item The responder decrypts $SSF_i^{I}(1)$ using key $k_2$ and verifies the authenticator $\alpha_i(1)$ using $k_1$. The responder records the path loss $PL_i(1)=P_{\max} - P_i^R(1)$, where $P_i^R(1)$ is the received power from sector $S_i.$
    
    \item The responder finds the sector with the smallest path loss (highest SNR) as $S^{\ast}_I(1)=\arg\min_{i} PL_i(1)$.

\end{enumerate}

\textbf{\textit{First round responder sector sweep.}} 
\begin{enumerate}
    \setcounter{enumi}{4}
    
    \item The initiator and the responder switch roles and repeat steps 1) to 4). Each sweep frame $SSF_i^{R}(1)$ includes the best sector $S^{\ast}_I$ and the path loss vector $PL(1) =\{ PL_1(1), PL_2(1)\ldots Pl_N(1)\}$ in the SSW feedback field.
    
\end{enumerate}

\textbf{\textit{Second round initiator sector sweep.}} 
\begin{enumerate}
    \setcounter{enumi}{5}

    \item The initiator calculates the transmission power range for each sector $S_i$ as $r_i: [P_{i,\min},P_{\max}-\epsilon].$ Here, $P_{i,\min}$ is the minimum transmission power that allows decoding at $R$ and is calculated as $P_{i,\min} = P_f + PL_i(1),$ where $P_f$ is the receiver's sensitivity. Parameter $\epsilon$ is the expected path loss variation over the same path. 

    \item The initiator perturbs the sector sweep sequence $S =\{S_1,S_2,$ $\ldots,S_N\}$ using a pseudorandom permutation $\Pi.$ The initiator sweeps through the $N$ sectors using sequence $\Pi(S).$ 
    
    \item For each sector $S_i$, $I$ randomly selects the transmission power $P_i^{T}(2) \in r_i$, calculates authenticator $\alpha_i(2)$, attaches $\alpha_i(2), P_i^T(2), v_i(2),$ and $S_i$ to $SSF_i^{I}(2)$, encrypts it with $k_2$ and transmits it with power $P_i^{T}(2)$. 
    
    \item {\bf Path loss test:} The responder decrypts the $SSF_i^{I}(2)$ using key $k_2$ and verifies the authenticator $\alpha_i(1)$ using $k_1$. The responder records the path loss $PL_i(2)=P_i^T(2) - P_i^R(2)$, where $P_i^R(2)$ is the received power from sector $S_i$ during the second round. If 
    $|PL_i(1) - PL_i(2)|>\epsilon$, $R$ detects an attack and aborts the protocol. Parameter $\epsilon$ is the same expected path loss variation over the same path as in Step 6.

    \item {\bf Coarse AoA test:} The responder records the quasi-omni receive sector that experiences the least path loss from each transmit sector $S_i.$ It maintains the fraction of sectors heard at each receive sector as $F=\{f_1, f_2, \dots, f_L\}$, where $L$ is the total number of quasi-omni sectors that cover the plane.  If $f_j > \beta \cdot N, j=1, 2, \cdots, L$, $R$ detects an attack and aborts the protocol.

\end{enumerate}

\textbf{\textit{Second round responder sector sweep.}} 
\begin{enumerate}
    \setcounter{enumi}{10}
    
    \item The responder and the initiator switch roles and repeat steps (6) to (10). 
    
\end{enumerate}

\subsection{Security Analysis of SecBeam}
\label{security analysis}

{\bf Forgery attacks.} The original SLS protocol is susceptible to forgery attacks where Mallory can forge SSFs and change the feedback field on the optimal transmitting beam \cite{steinmetzer2018beam}. Such attacks are prevented in SecBeam due to the addition of integrity protection on the SSFs. Without access to the shared key $s,$ Mallory cannot recover $k_1$ to generate a valid MAC (under standard security assumptions for the MAC function). This security property is similar to the one provided by the authenticated SLS in \cite{steinmetzer2018authenticating}.

{\bf Detection of Fixed Power Attacks.} Fixed power attacks are countered by the power randomization method that is employed in the second round of SecBeam. This is because when relaying SSFs at fixed power, the path loss test of Step 9 is violated. We demonstrate this in Proposition 1. 

\begin{proposition}
\label{Prop1}
 A fixed power attack on sector $S_i$ is detected by the path loss test if the transmission power $P_i^T(2)$  in the second round satisfies 
\[
P_i^T(2) < P_i^{T}(1) - \epsilon,
\]
where $\epsilon$ is a noise-determined parameter. 
\end{proposition}

\begin{proof}
The proof is provided in Appendix B.
\end{proof}

The intuition behind Proposition 1 is that by transmitting at a power lower than $P_i^{T}(1) - \epsilon$ in the second round, the path loss difference between the two rounds will violate the path loss test (will exceed $\epsilon$) because the adversary relays at the same power, regardless of the transmission power of $I$. Parameter $\epsilon$ must be selected in such a way to account for the natural path loss variations between rounds. Such variations are expected to be primarily caused by noise, as the geometric properties of the mmWave channel lead to long coherence times in most environments of interest. We study the choice of $\epsilon$ in the evaluation via simulations and experiments.

{\bf Detection of Fixed Amplification Attacks.} In a fixed amplification attack, Mallory applies a fixed gain to each sector she relays as opposed to transmitting at fixed power.  Let Mallory receive from $N_M \leq N$ sectors during the first round initiator sector sweep, when $P_i^{T}=P_{\max},\forall i$. Let also $K \leq N_M$ denote the sectors that, when amplified, achieve an SNR higher than the optimal sector without attack. The probability of passing the path loss test in the second round initiator sector sweep is given by Proposition 2.

\begin{figure}[t]
\centering
\includegraphics[width=0.6\linewidth]{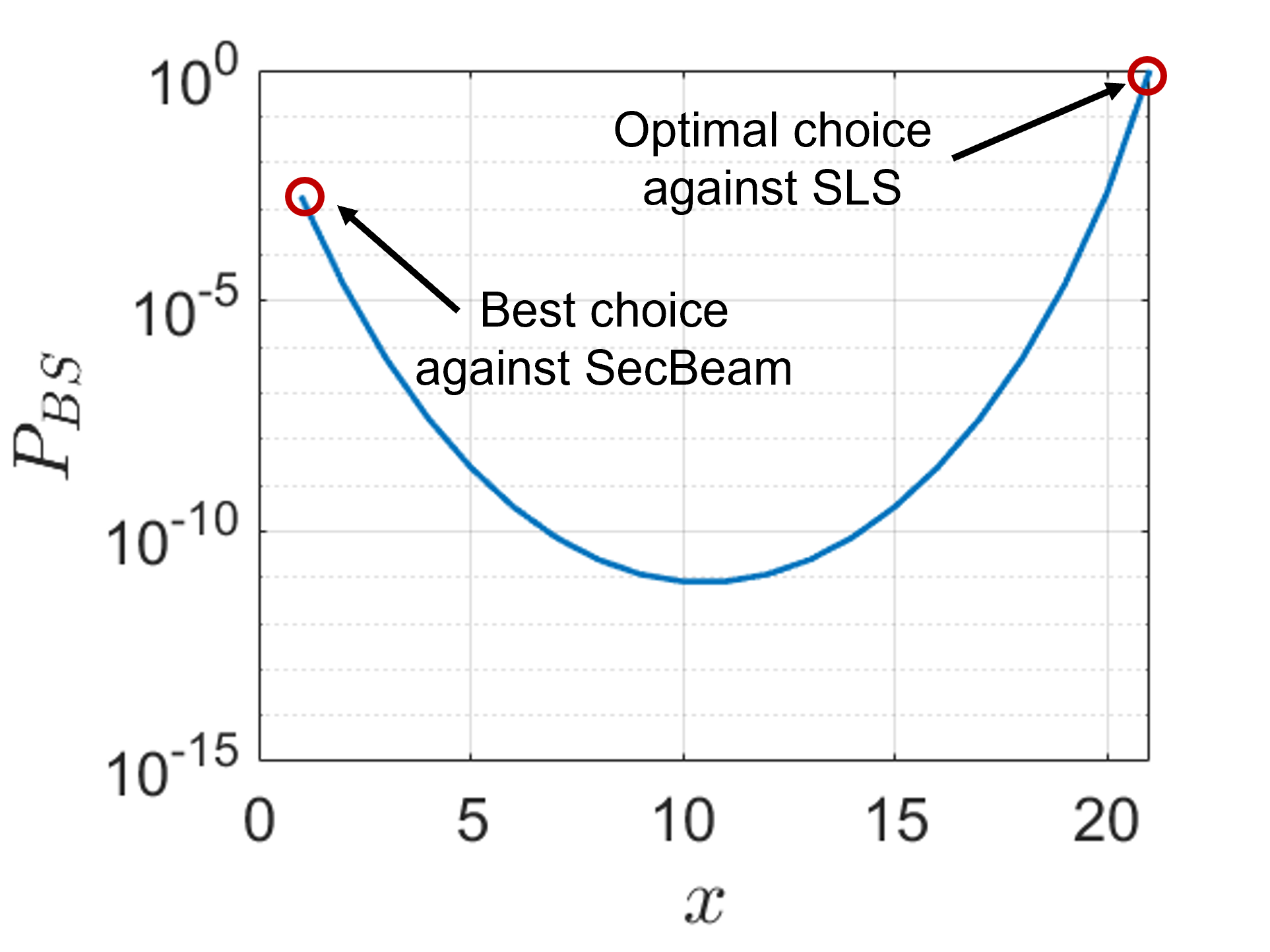}
\vspace{-0.1in}
\caption{ $P_{BS}$ when $N=32$, $N_M=21$ and $K=19$.}
\label{passing rate of M}
\vspace{-0.20in}
\end{figure}

\begin{proposition}
 Given the number of sectors $N_M$ heard at $\mathcal{M}$ and the number of sectors $K$  in $N_M$ with path loss smaller than the best path without attack under a fixed amplification strategy, the probability $P_{BS}$ of a successful beam stealing attack is given by 
 \[
    Pr_{BS}=\frac{1}{\binom{N_M}{x}^4} \cdot \left (\sum_{i=1}^{\min(x, K)}\binom{K}{i}\binom{N_M-K}{x-i}\right)^2.
\]
\end{proposition} 

\begin{proof}
The proof is provided in Appendix C.  
\end{proof}

{\bf Discussion.} The probability of success for the fixed amplification attack is a function of the number of sectors $N_M$ heard at $\mathcal{M}$ which is location- and wireless environment-dependent and the choice of $x$. It is evident that $Pr_{BS}$ becomes 1 when $x = N_M$ (i.e., Mallory amplifies all sectors it hears in both rounds). The coarse AOA test is introduced to prevent the success of this strategy. When Mallory amplifies and relays all $N_M$ sectors, the responder detects an unusual number of sectors arriving from the direction of the relay, thus detecting the attack. For values of $x$ smaller than the coarse AoA test $\beta N,$ the probability of success rapidly declines due to the requirement of relaying the same set of sectors on both rounds. In fact, based on the convexity of the binomial distribution \cite{ramsey2005convexity}, selecting $x=1$ becomes Mallory's most beneficial strategy. In this case, the probability of success is $Pr_{BS}= K^2/N_M^4.$ 

As an example, Fig~\ref{passing rate of M} shows $P_{BS}$, when $N=32,$ $N_M = 21$ and $K=19.$ With $\beta=0.5$, Mallory must select $x<16$ to pass the coarse AoA test  with certainty. In this case, $x=1$ yields the highest passing probability which is in the order of $10^{-3}.$ This is a fairly low probability for an online attack.

Another important remark for the analysis is that we have made the reasonable assumption of an adversary that does not know a priori the exact locations of $I$ and $R$. Therefore she uses quasi-omni beams to the general directions of $I$ and $R$ in order to receive, amplify, and relay. If the adversary could exactly pinpoint $I$'s and $R$'s locations, she could point fine beams at each party, therefore substantially reducing the number of sectors $N_M$ received by $\mathcal{M}.$ In this case, relaying all $N_M$ sectors passes the path lost test, while also defeating the coarse AoA test if $N_M < \beta N.$

\section{Parameter Selection and Evaluation}
\label{simu}

\begin{figure}[t]%
\centering
\setlength{\tabcolsep}{-5pt}
\begin{tabular}{cc}
  \includegraphics[width=0.55\columnwidth]{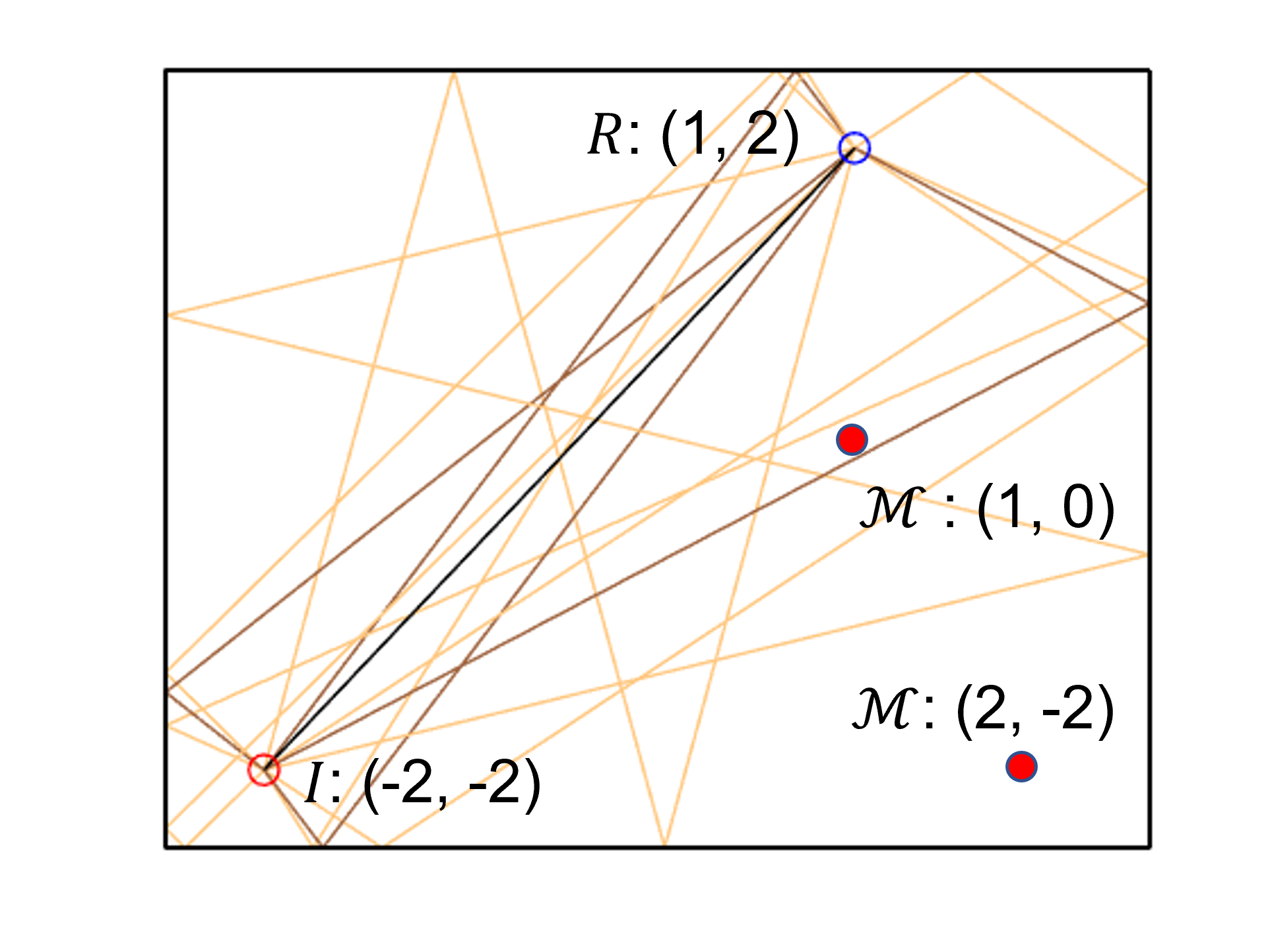} &
  \includegraphics[width=0.55\columnwidth]{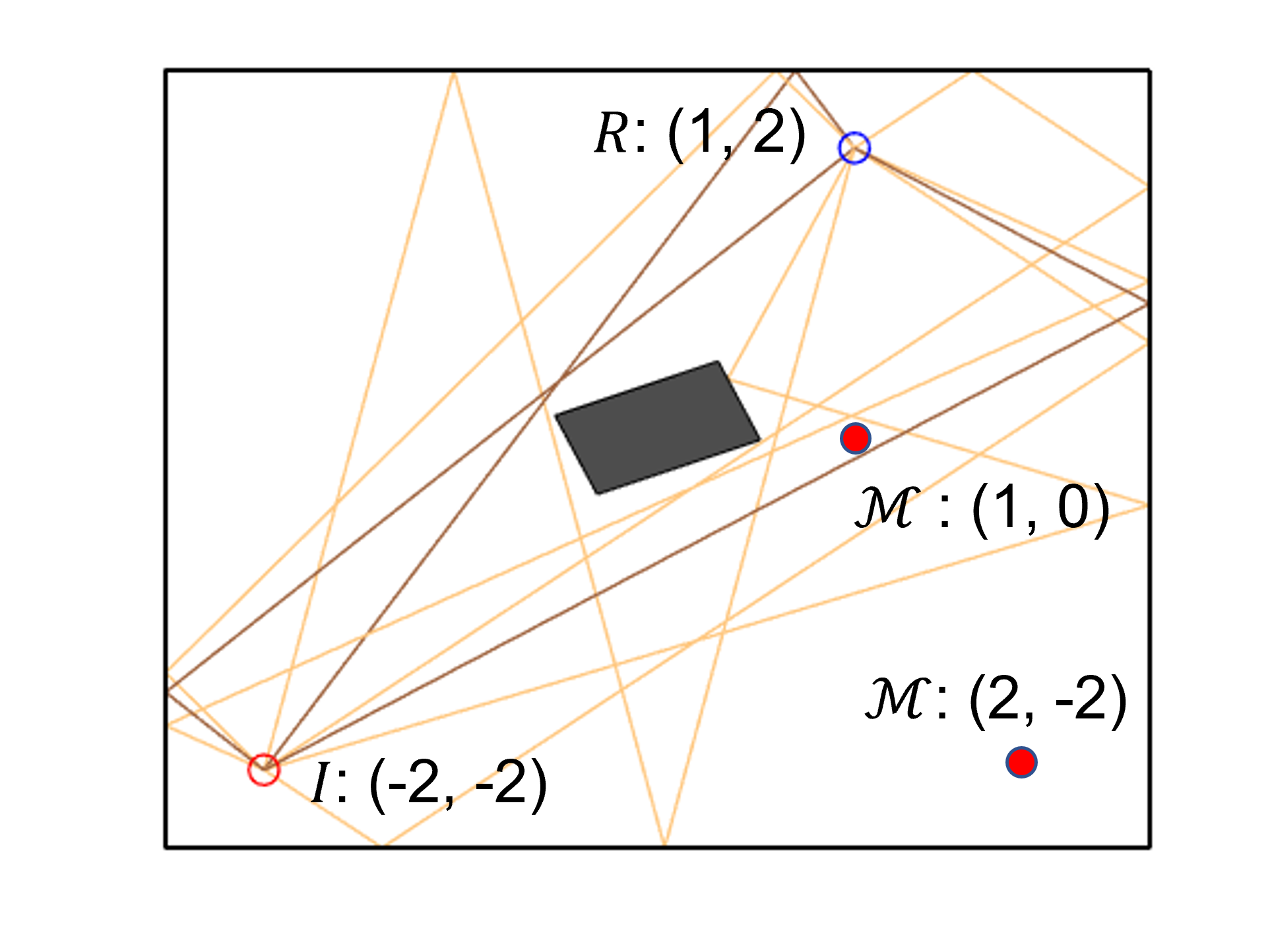} \\ 
  (a) & (b) \\
\end{tabular}
\vspace{-0.1in}
\caption{(a) LoS is unobstructed, (b) LoS is blocked by an obstacle with permittivity equal to 2. $\mathcal{M}$ was placed at two different locations.}
\label{fig:layout}
\vspace{-0.1in}
\end{figure}

\begin{figure}[t]
\centering
\includegraphics[width=0.6\linewidth]{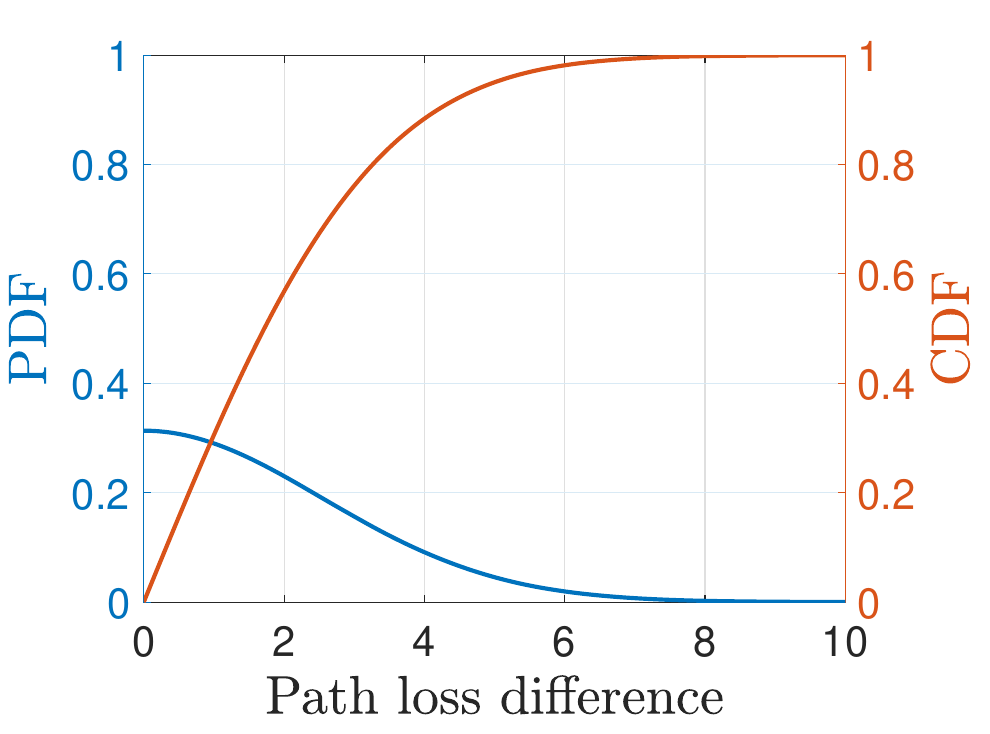}
\caption{PDF and CDF of $\Delta P$ given $\sigma=1.8$}
\label{PDF and CDF}
\vspace{-0.1in}
\end{figure}

In this section, we use simulation based on ray-tracing (mmTrace) \cite{steinmetzer2016mmtrace} to select parameters for SecBeam and evaluate its security. The insights from the simulations guide our experimental evaluations. 

\subsection{Simulation Setup}

We consider two scenarios in our simulations,  one where  the LoS path between $I$ and $R$ is unobstructed (Fig. \ref{fig:layout}(a)), and the other one with the  LoS   blocked by an obstacle (Fig. \ref{fig:layout}(b)).   
 All devices are placed in a self-defined room with a size of $5m \times 5m$, where the wall permittivity is set as 3.24. The center coordinate of the room is set to $(0, 0)$, and the distance between $I$ and $R$ is $d_{IR}=5m$. In addition, we consider two possible locations for $\mathcal{M}$: $(1, 0)$ and $(2, -2)$, which result in different distances between $I$, $R$, and $\mathcal{M}$, shown in Table II. 
 Each device has $N$ sectors in total, so the HPBW is $\lceil \frac{360^\circ}{N} \rceil$. Additionally, the HPBW of $I$ and $R$'s quasi-omni pattern is $60^\circ$, and their maximum transmit power is $P_{max}=0$ dBm. The initiator follows our protocol to perform the two-round sector sweep, using $P_{max}$ in the first round and randomized power in the second round. The adversary implements the amplify-and-relay attack strategy with fixed amplification.

\begin{table}[t]
\centering
\caption{Simulation scenarios}
\setlength{\tabcolsep}{2pt}
\begin{tabularx}{0.48\textwidth}{|X|X|X|X|X|}
\hline
   &
  $d_{IR}$&
  $d_{MR}$ &
  Room size & 
  Amp. \\
\hline 
  {Scenario 1}&
  5m&
  1m&
  $5 \times 5 \times 5$m&
  40 dB\\
    \hline 
  {Scenario 2}&
  5m&
  1m&
  $5 \times 5 \times 5$m&
  50 dB \\
  \hline
  {Scenario 3}&
  5m&
  1m&
  $5 \times 9 \times 5$m&
  40 dB \\
\hline
  {Scenario 4}&
  5m&
  $4.1$m&
  $5 \times 5 \times 5$m&
  70 dB \\
  \hline
  {Scenario 5}&
  5m&
  $5.4$m&
  $9 \times 9 \times 9$m&
  70 dB \\
\hline
\end{tabularx}
\label{diff simu setup}
\vspace{-0.15in}
\end{table}


\begin{figure}[t]
\centering
\includegraphics[width=0.6\linewidth]{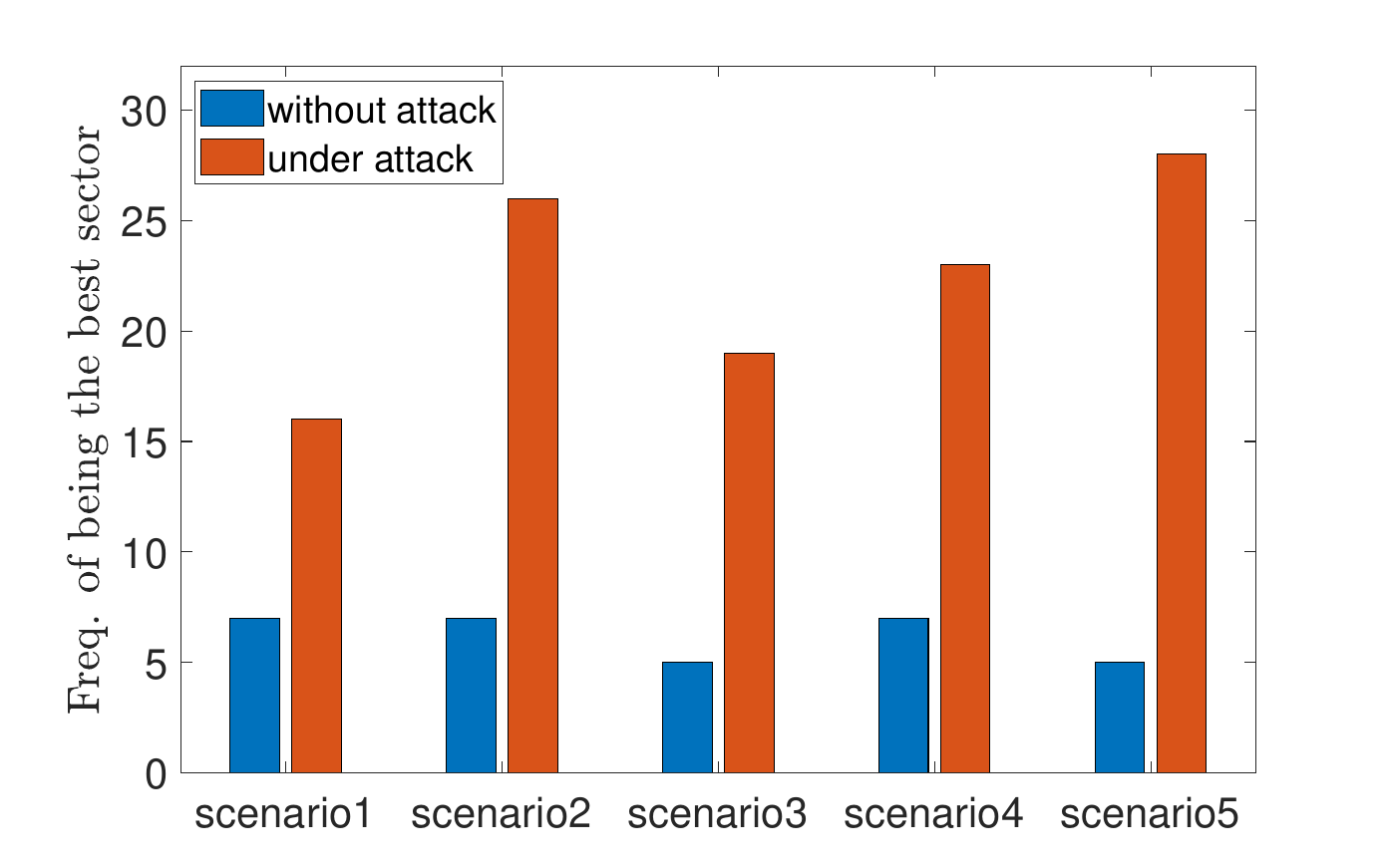}
\caption{The highest frequency of a sector being the best receive sector with and without $\mathcal{M}$ in different scenarios.}
\label{diff setup}
\vspace{-0.2in}
\end{figure}

\begin{figure}[t]%
\centering
\setlength{\tabcolsep}{-5pt}
\begin{tabular}{cc}
  \includegraphics[width=0.55\columnwidth]{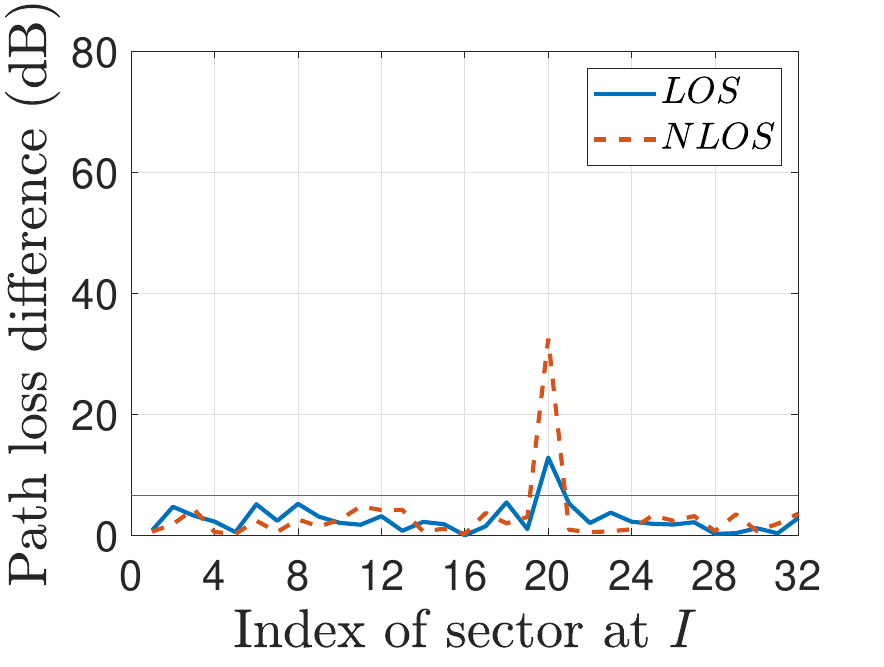} &
  \includegraphics[width=0.55\columnwidth]{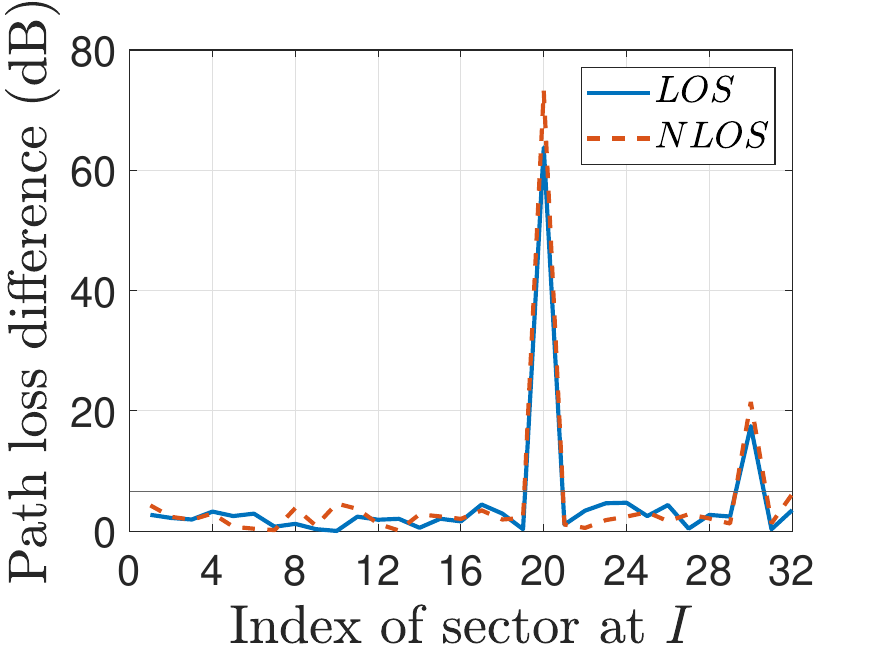} \\ 
  (a) & (b) 
\end{tabular}
\caption{Path loss difference with $\mathcal{M}$ in LoS and NLos scenarios: (a) $\mathcal{M}$ chooses the same sector in two rounds of SSW, (b) $\mathcal{M}$ chooses different sectors in two rounds of SSW.}
\label{path loss diff with M}
\vspace{-0.20in}
\end{figure}

\begin{figure}[t]%
\centering
\setlength{\tabcolsep}{-3pt}
\begin{tabular}{cc}
  \includegraphics[width=0.55\columnwidth]{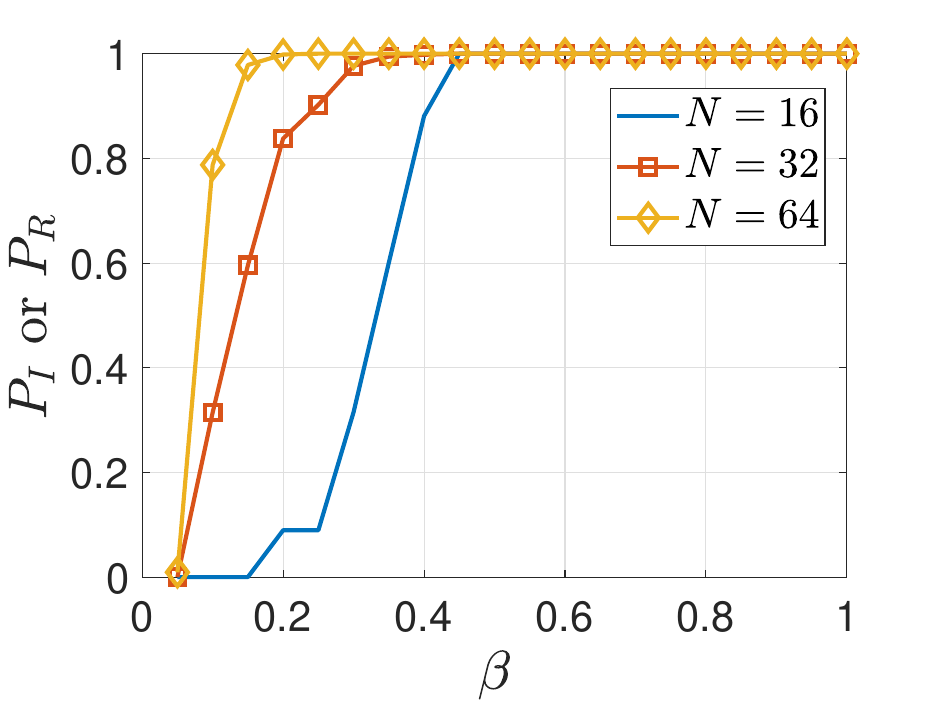} &
  \includegraphics[width=0.55\columnwidth]{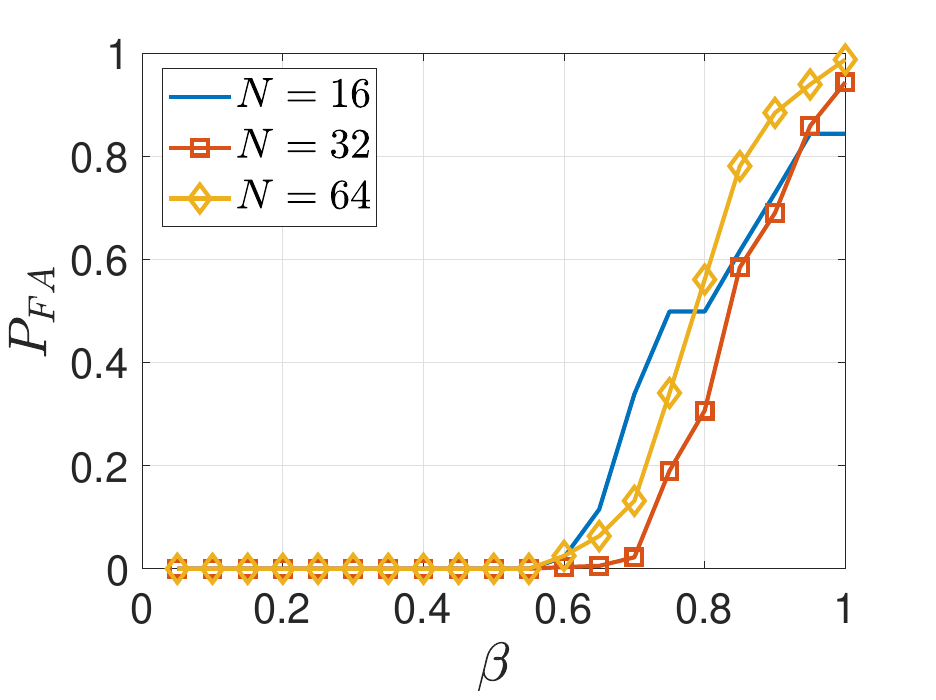} \\ 
  (a) & (b) 
\end{tabular}
\caption{Passing rates of a legitimate device ($P_I$ or $P_R$) and $\mathcal{M}$ ($P_{FA}$) versus $\beta$, for different $N$.}
\label{passing rate beta}
\vspace{-0.20in}
\end{figure}

\begin{figure*}%
\centering
\setlength{\tabcolsep}{-3pt}
\begin{tabular}{cccc}
  \includegraphics[width=0.54\columnwidth]{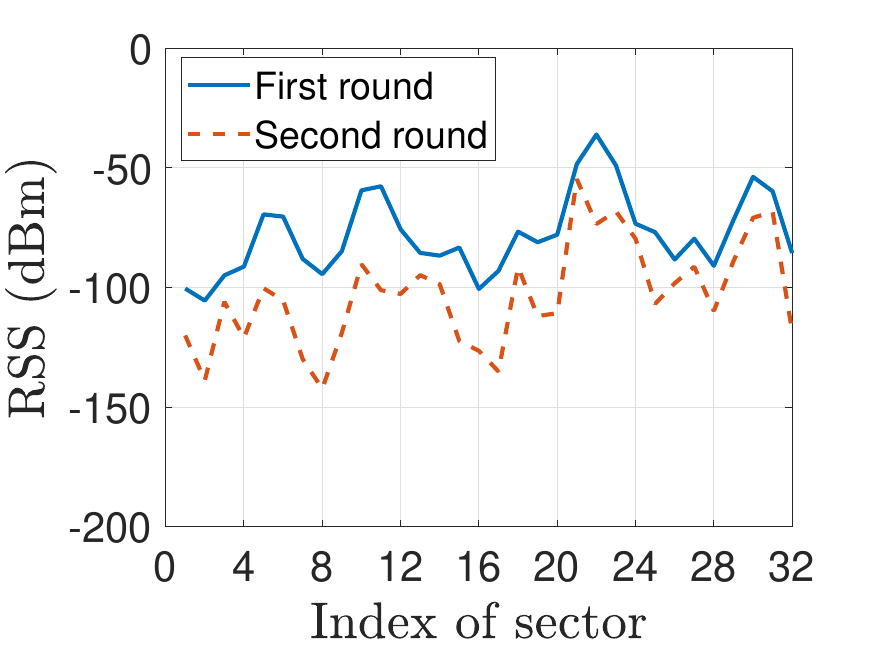} &
  \includegraphics[width=0.54\columnwidth]{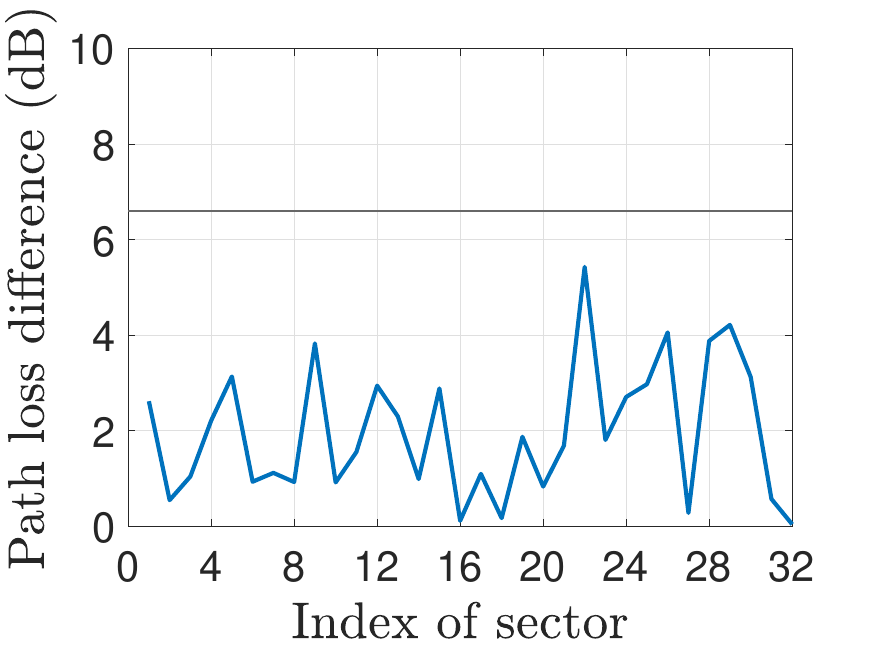} & 
  \includegraphics[width=0.54\columnwidth]{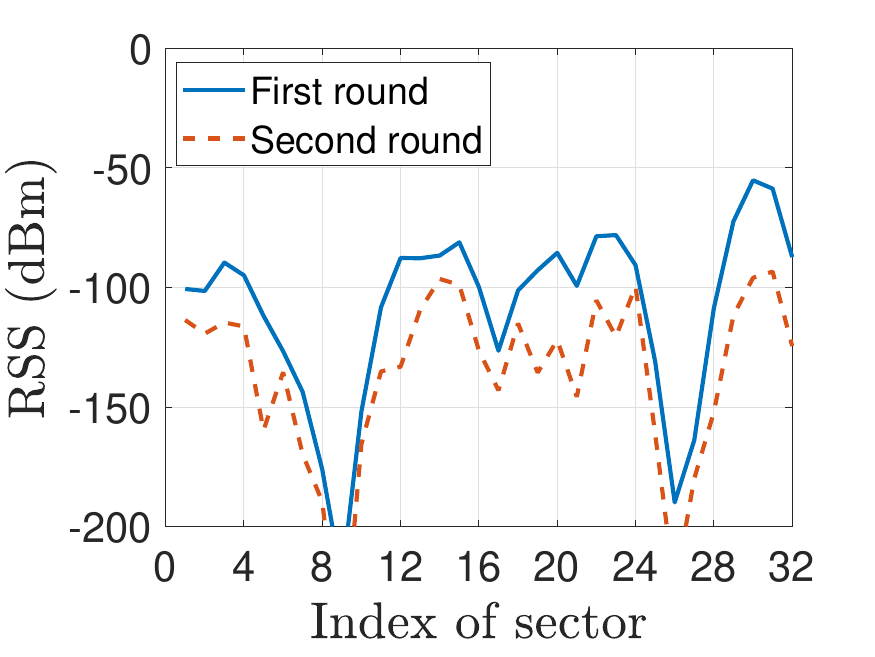} &
  \includegraphics[width=0.54\columnwidth]{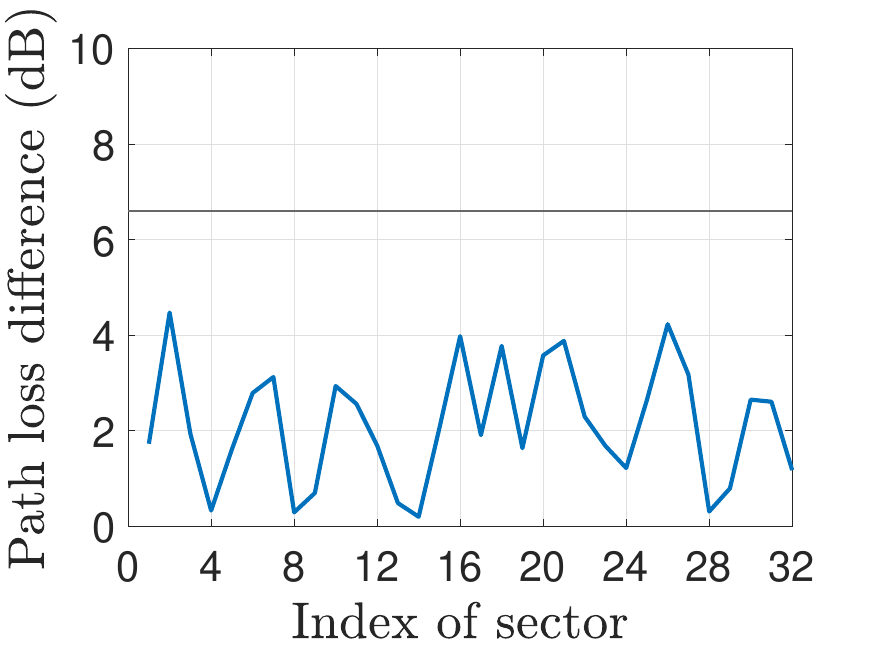} \\
  (a)  & (b)  & (c) & (d)
\end{tabular}
\caption{Simulation results without $\mathcal{M}$: (a) RSS measured by $R$ when LoS is unobstructed, (b) path loss difference between two rounds when LoS is unobstructed, (c) RSS measured by $R$ when LoS is blocked, (d) path loss difference between two rounds when LoS is blocked. The path loss difference is lower than the threshold in both LoS and NLoS scenarios.}
\label{fig:RSS and path loss difference}
\vspace{-0.20in}
\end{figure*}

\begin{figure*}%
\centering
\setlength{\tabcolsep}{-3pt}
\begin{tabular}{cccc}
  \includegraphics[width=0.55\columnwidth]{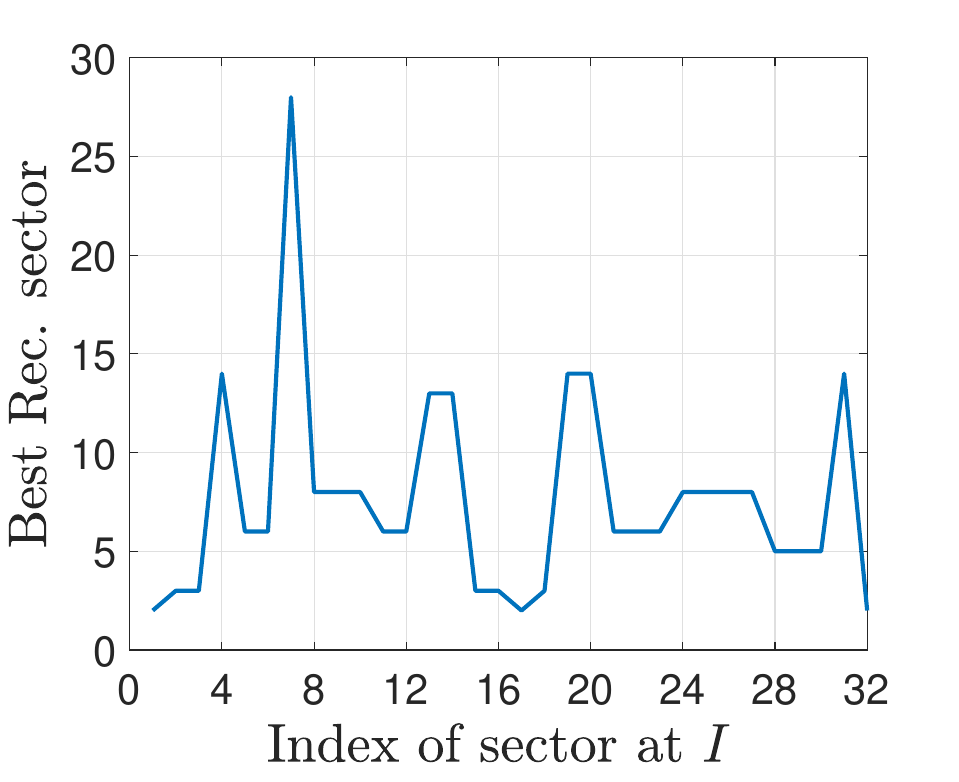} &
  \includegraphics[width=0.55\columnwidth]{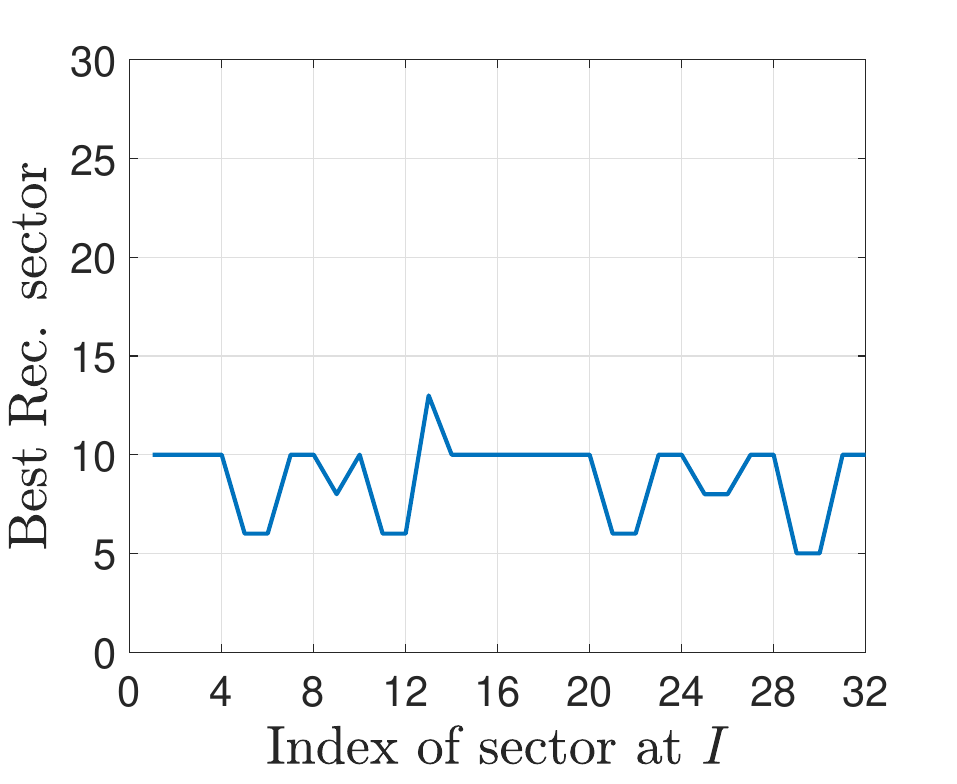} & 
  \includegraphics[width=0.55\columnwidth]{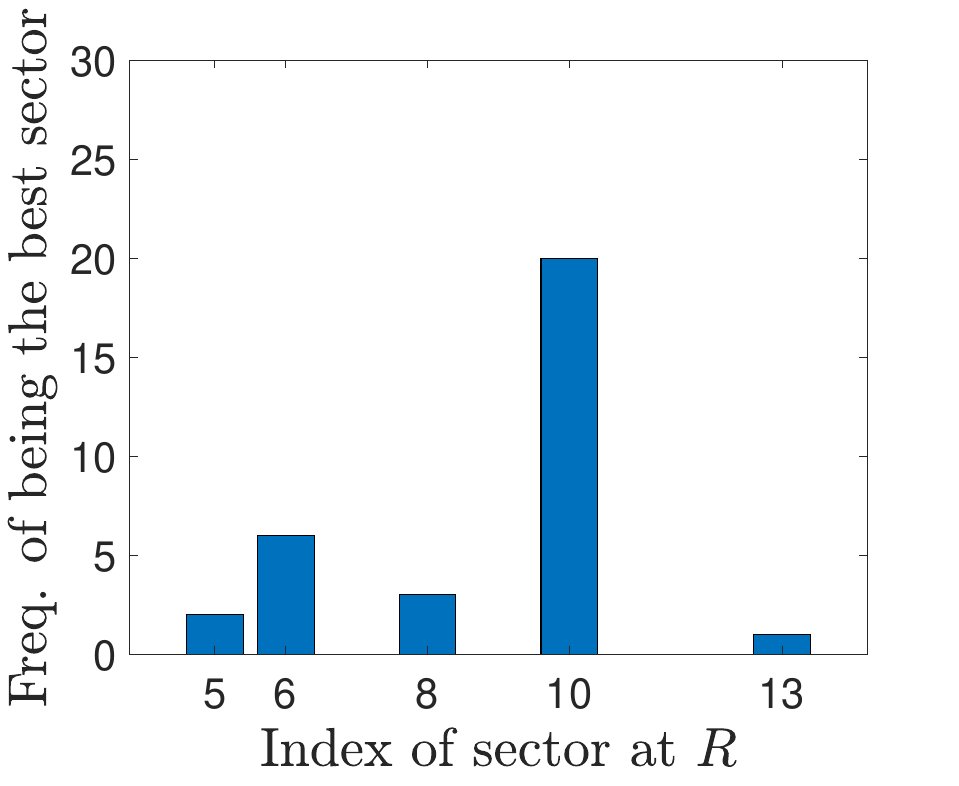} &
  \includegraphics[width=0.45\columnwidth]{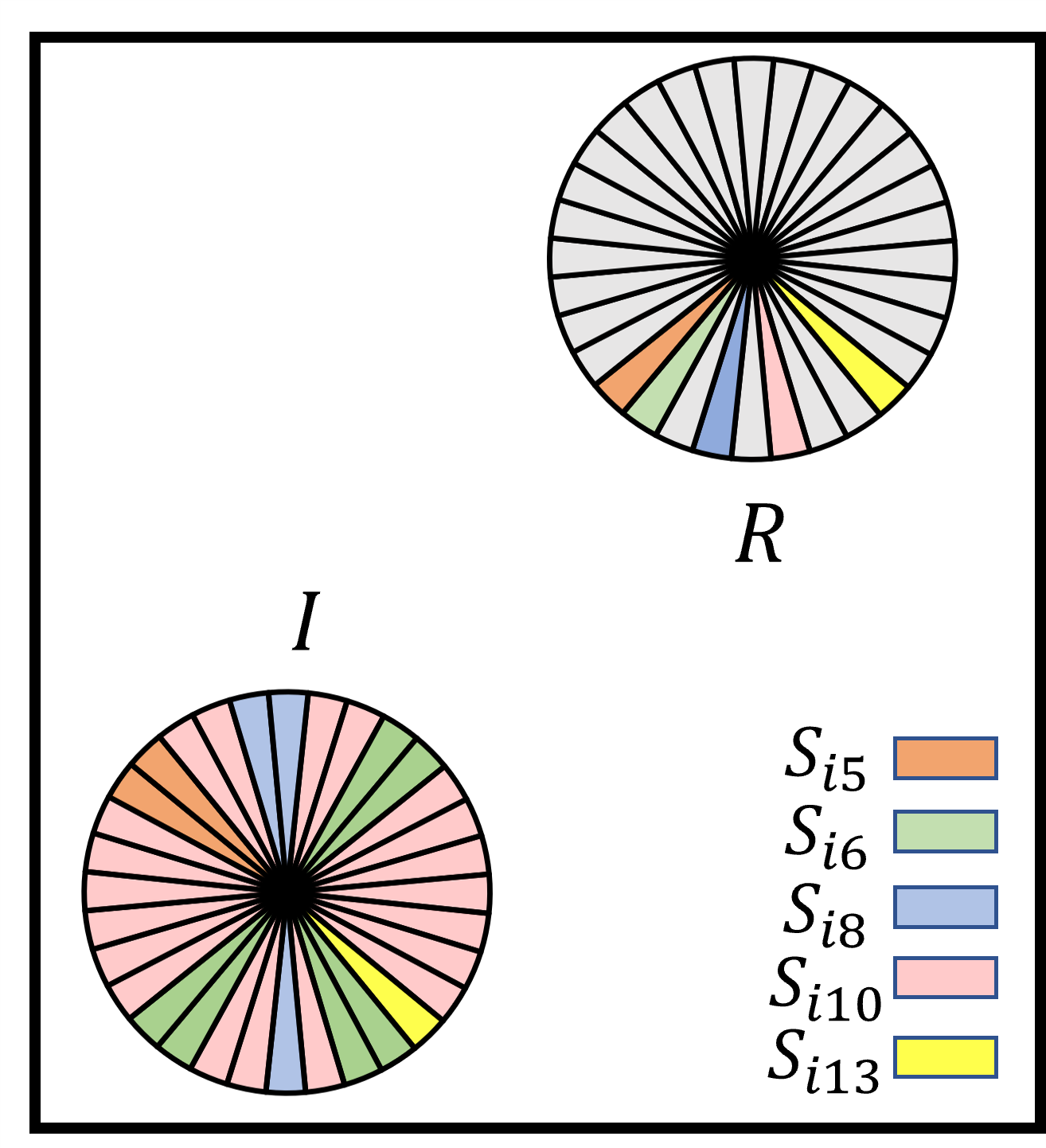} \\
  (a)  & (b)  & (c) & (d)
\end{tabular}
\caption{When LoS between $I$ and $R$ is unobstructed:(a) The best Tx and Rx pair $S_{ij}$ without relay attack. (b) $S_{ij}$ under relay attack. (c) the frequency of a Rx sector being chosen as the best receive sector under relay attack. (d) $S_{ij}$ in a real setting.}
\label{fig:AoA test LOS}
\vspace{-0.20in}
\end{figure*}

\begin{figure*}%
\centering
\setlength{\tabcolsep}{-3pt}
\begin{tabular}{cccc}
  \includegraphics[width=0.55\columnwidth]{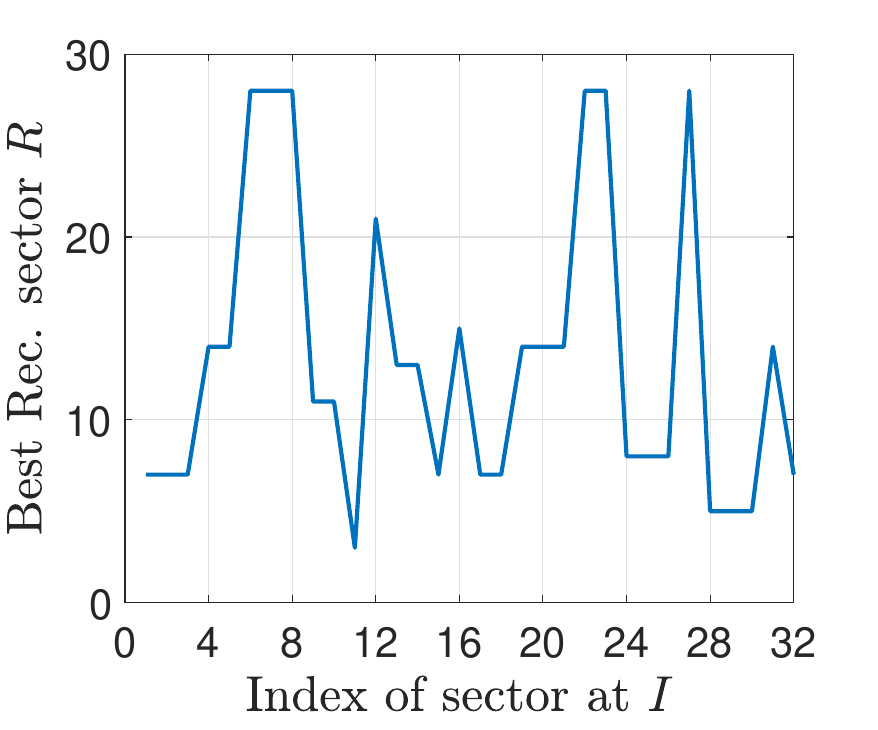} &
  \includegraphics[width=0.55\columnwidth]{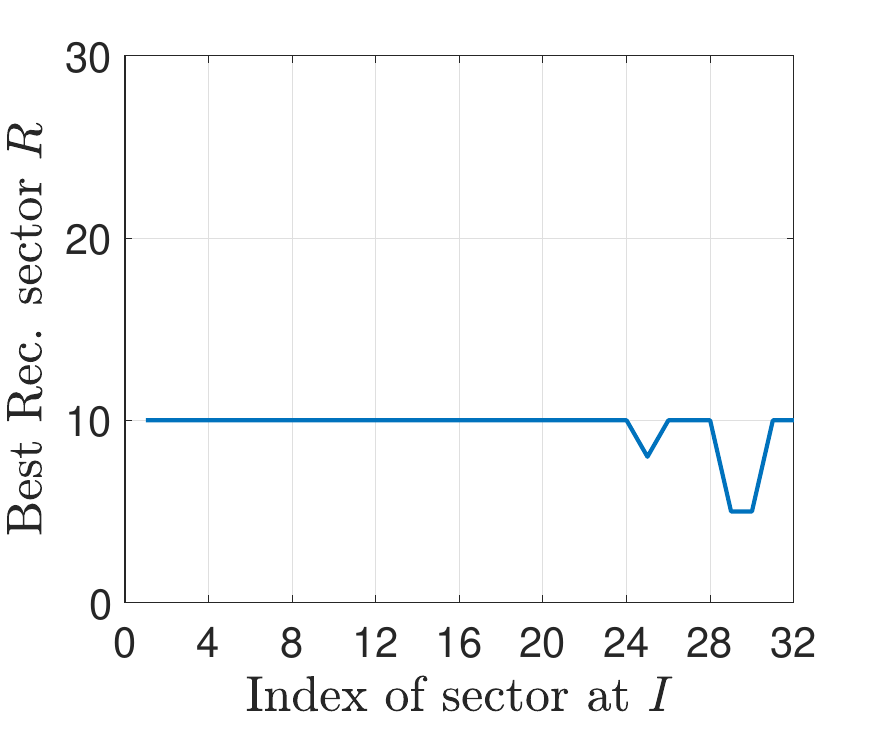} & 
  \includegraphics[width=0.55\columnwidth]{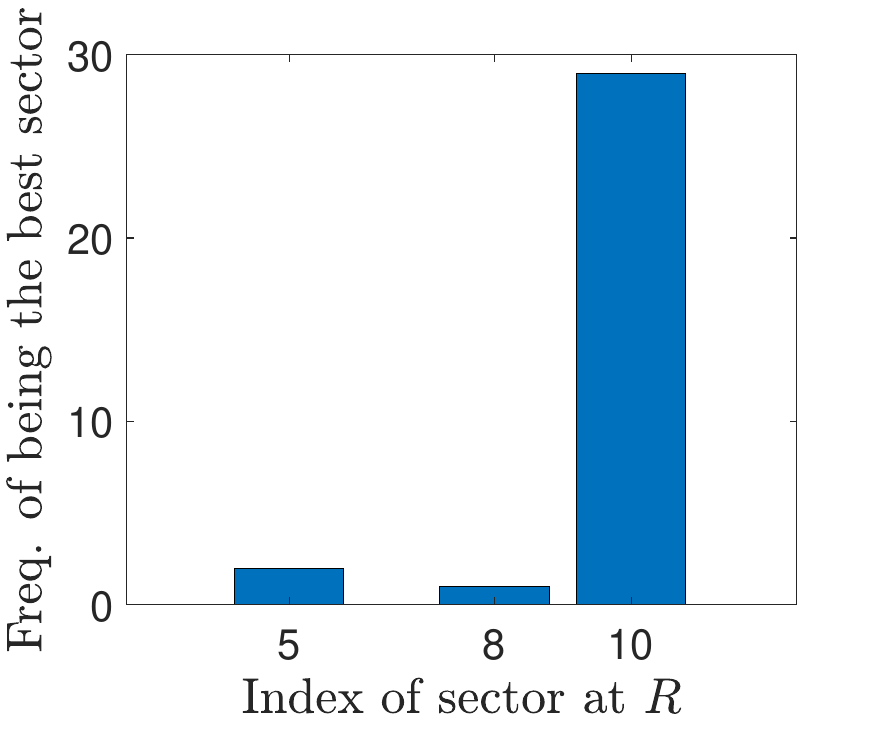} &
  \includegraphics[width=0.45\columnwidth]{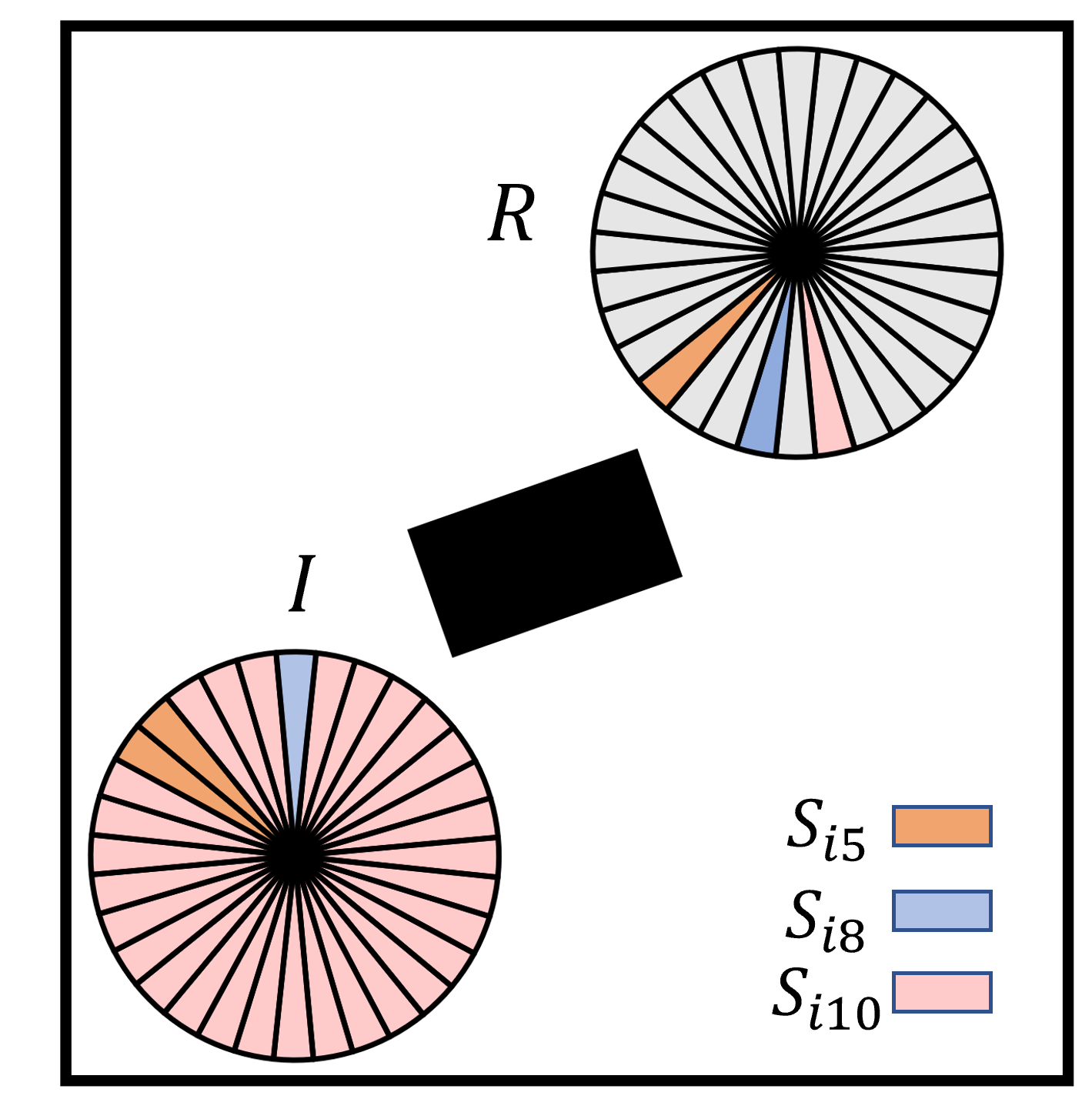} \\
  (a)  & (b)  & (c) & (d)
\end{tabular}
\caption{When LoS between $I$ and $R$ is blocked:(a) The best Tx and Rx pair $S_{ij}$ without relay attack. (b) $S_{ij}$ under relay attack. (c) the frequency of an Rx sector being chosen as the best receive sector under relay attack. (d) $S_{ij}$ in a real setting. Sectors at RX are the best receive sectors for the transmit sectors indicating in the same color in Tx.}
\label{fig:AoA test NLOS}
\vspace{-0.20in}
\end{figure*}

\subsection{Parameter Selection}
The SecBeam protocol is parameterized by the threshold of path loss difference $\epsilon$ and the frequency of one sector being chosen as the best receive sector $\beta$. In the following, we use simulations to  help with selection of these parameters. 

{\bf Selecting $\epsilon$.} The value of $\epsilon$ needs to be selected high enough to prevent false alarms when $\mathcal{M}$ is not present. The path loss can be generally written as, 
\begin{equation}
    PL(dB)=G_T(dBi) +G_R(dBi) - L_{ch}(dB)+ X_{\sigma}(dBm),
    \label{eq:path loss}
\end{equation}
where $P^T$ is the transmit power, $G^T$ and $G^R$ are the respective antenna gains, $L_{ch}$ is the channel attenuation according to some unknown channel model, and $X_{\sigma}$ is a zero-mean Gaussian noise with a standard deviation $\sigma$. Assuming fixed locations and paths, the path loss difference $\Delta P(dB)$ between the two rounds is,
\begin{equation}
    \Delta P=\lvert PL_i(2)-PL_i(1) \rvert=\lvert X_{\sigma}(2) - X_{\sigma}(1) \rvert.
\end{equation}

The path loss difference $\Delta P$ follows a folded normal distribution with a $\frac{\sigma \sqrt{2}}{\sqrt{\pi}}$ mean and $\sigma^2(1-\frac{2}{\pi})$ variance. A typical value for $\sigma$ in the indoor environment is  1.8 \cite{wu201528}. The probability density function (PDF) and the cumulative distribution function (CDF) of $\Delta P$ is shown in Fig.~\ref{PDF and CDF}. In order to ensure a specific passing rate of valid devices, we can select a suitable value for $\epsilon$ based on the CDF of $\Delta P$. For instance, in our study, we choose $\epsilon=6.6$dB to ensure that under benign conditions, the path loss difference due to noise is less than $\epsilon$ with probability 99$\%$.

{\bf Selecting $\beta$.} We then used our simulation results to determine the threshold $\tau$. The optimal value of $\beta$  depends on the physical environment. We listed five different simulation setups in Table. \ref{diff simu setup} where we vary  the room size, locations of $I$ and $\mathcal{M}$, and the amount of amplification (Amp.) used by $\mathcal{M}$. The total number of sectors is  $N=32$ in all scenarios. We extract the highest frequency $f_j, j=1, \dots, N$ in each setup with and without a relay attack, and the results are shown in Fig. \ref{diff setup}. In all the scenarios, there are a clear distinction between the beam selection frequencies of benign and adversarial settings, which shows the feasibility of the coarse AoA test. 

In addition, we performed 1,000 simulations for   $N=16, 32$, and 64 with different combinations of the locations of $I$ and $R$, room size, and amplification gain. The passing rate of legitimate devices and $\mathcal{M}$ versus $\beta$ is shown in Fig. \ref{passing rate beta}. We can see that the passing rate of $I$ and $FA$ attacker ($P_I$ and $P_{FA}$) both increase with $\beta$. When $\beta=0.5$, $P_I=1$ while $P_{FA}=0$. So we conclude that $\beta=0.5$ is sufficient to detect the attack for all values of $N$. It is worth mentioning that when the LoS between $I$ and $R$ is blocked, the frequency of one sector being selected as the best receive sector is higher than those shown in Fig. \ref{diff setup}, since the LoS between $\mathcal{M}$ and victims is always unobstructed. We discuss this in Sec. \ref{against fixed amp. attack}.

\subsection{Simulation Results} 
\label{Simulation results of the second method}
\subsubsection{Fixed power relay attack} We use Setup 1 to demonstrate the security of SecBeam against the fixed power relay attack. We show the results of the RSS measured by $R$ in two rounds and the path loss difference when LoS is unobstructed and blocked in Fig. \ref{fig:RSS and path loss difference}. From Fig. \ref{fig:RSS and path loss difference}(a),   sector 22 is chosen as the best sector in the first round, while in Fig. \ref{fig:RSS and path loss difference}(c), the best sector is   30 because the LoS is blocked. 
We can also see that the path loss difference across two rounds is smaller than the threshold $\epsilon$ for all sectors, no matter whether the LoS is blocked or not.

When the adversary is present, the results are shown in Fig. \ref{path loss diff with M}. When $\mathcal{M}$ chooses different sectors in two SSW rounds, e.g. sector 30 in the first round and sector 20 in the second round (Fig. \ref{path loss diff with M}(b)), the path loss difference exceeds $\epsilon$ in both cases. Even if $\mathcal{M}$ could choose the same sector in two SSW rounds coincidentally, (e.g. Sector 20 in Fig. \ref{path loss diff with M}(a)), the pass loss difference is still higher than  $\epsilon$ due to the relay at fixed power and the power randomization by $I$.

\subsubsection{Fixed amplification attack}
\label{against fixed amp. attack}
Next, we  show the simulation results when $\mathcal{M}$ uses the fixed amplification relay strategy under simulation setup 2. We  assume that $\mathcal{M}$ amplifies all the received beams by 60 dB and relays  them to $R$. When LoS is unobstructed, the results are shown in Fig. \ref{fig:AoA test LOS}. From Fig. \ref{fig:AoA test LOS}(a), we can see that in the benign case, even if some receive sectors, e.g. Sectors 6 and 8, are chosen more than once, $f_6=7/32$ and $f_8=7/32$ are less than the threshold $\beta=0.5$. However,  in Fig. \ref{fig:AoA test LOS}(b) and (c), $f_{10}=20/32$ under the FA attack, which is higher than the threshold, thus   $\mathcal{M}$ is detected. 
When the LoS is blocked, the results are shown in Fig. \ref{fig:AoA test NLOS}. In this case, there is no dominant path between $I$ and $R$, and the obstacle  acts as a reflector. The best receive sector does not present any strong directionality. Comparing with Fig. \ref{fig:AoA test LOS}(b) and (c), we can observe sector 10 is the best receive sector for 29 transmit sectors in Fig. \ref{fig:AoA test NLOS}(b) and \ref{fig:AoA test NLOS}(c) as expected.

\section{Experimental Evaluation}

\subsection{Experimental Setup}
\label{Exp setup}
To demonstrate the effectiveness and security of SecBeam in the real world, we implemented it using our mmWave testbed as described in Section~\ref{feasibility amplify attack}. 

Figure \ref{overall_setup} provides an overview of the experimental setup, where $d_{IM}$ was 3m and $d_{MR}$ was 2.4m. To increase the number of reflection paths in the environment, in certain setups, we added a whiteboard and a stop sign in the room. 
The relay attack was implemented by synchronizing $\mathcal{M},$ $I$ and $R$ via the Octo-clock module. This way the adversary could transmit the same SSF frame synchronously with the legitimate devices, without incurring the hardware delay of USRPs for switching between reception and transmission.

During the experiments, the transmitter sent 10,000 packets in each sector with a 5 milliseconds delay between packets to facilitate data collection (In reality,  only one packet will be sent in each sector). In the first round, $I$ swept in its maximum transmit power $P_{max}=30$dBm. In the second round, $I$ randomly chose the transmit power sequence for each sector. For the fixed power relay attack, the adversary chose signals from one Tx sector and relayed them to the receiver using $P_{max}$. For the fixed amplification relay attack, $\mathcal{M}$ amplified its received signals from $I$ with a fixed amount and then relayed them to $R$.

\begin{figure}[t]
\centering
\includegraphics[width=0.8\linewidth]{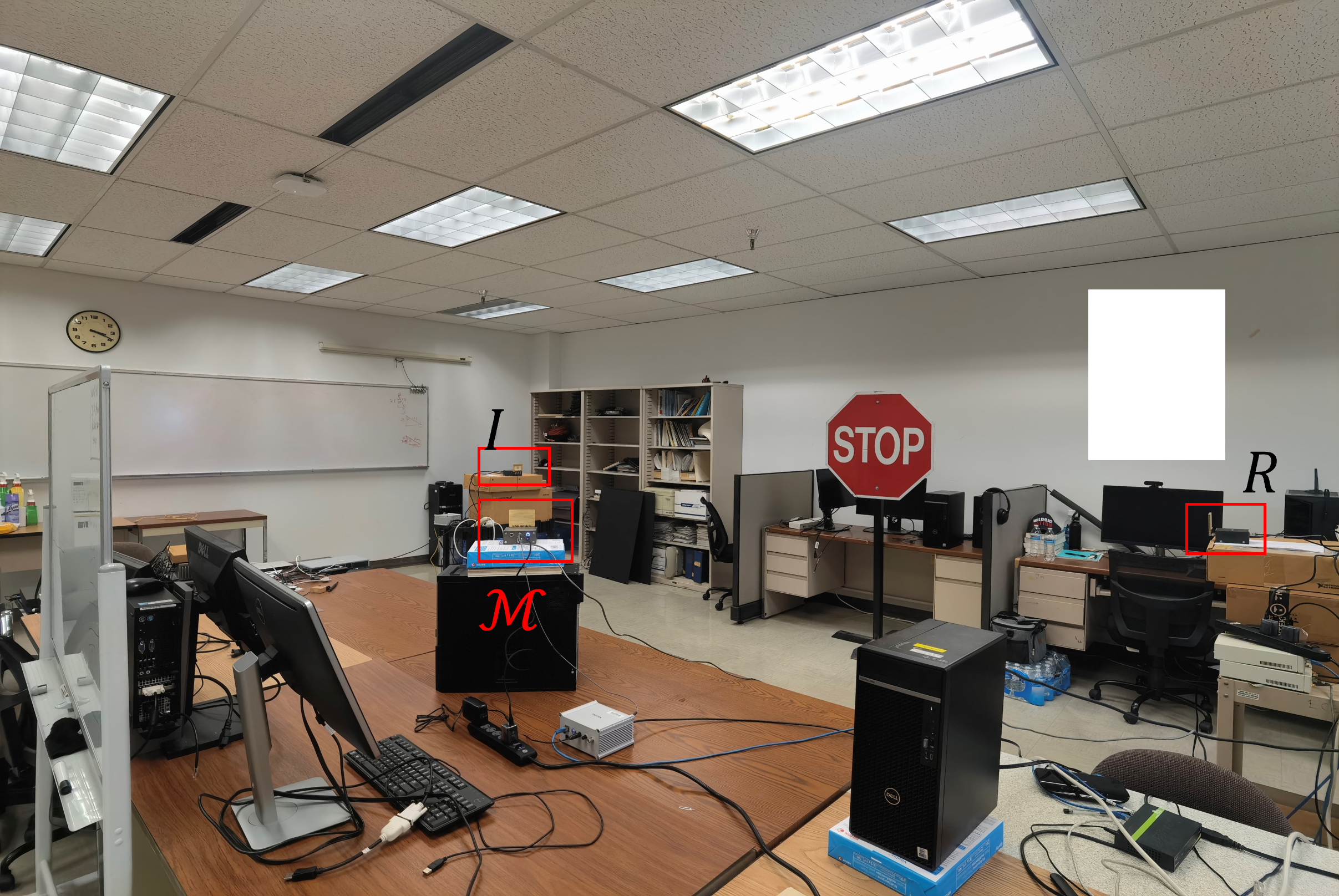}
\caption{Experimental setup.}
\label{overall_setup}
\vspace{-0.20in}
\end{figure}


\begin{figure}[t]%
\centering
\setlength{\tabcolsep}{-3pt}
\begin{tabular}{cc}
  \includegraphics[width=0.55\columnwidth]{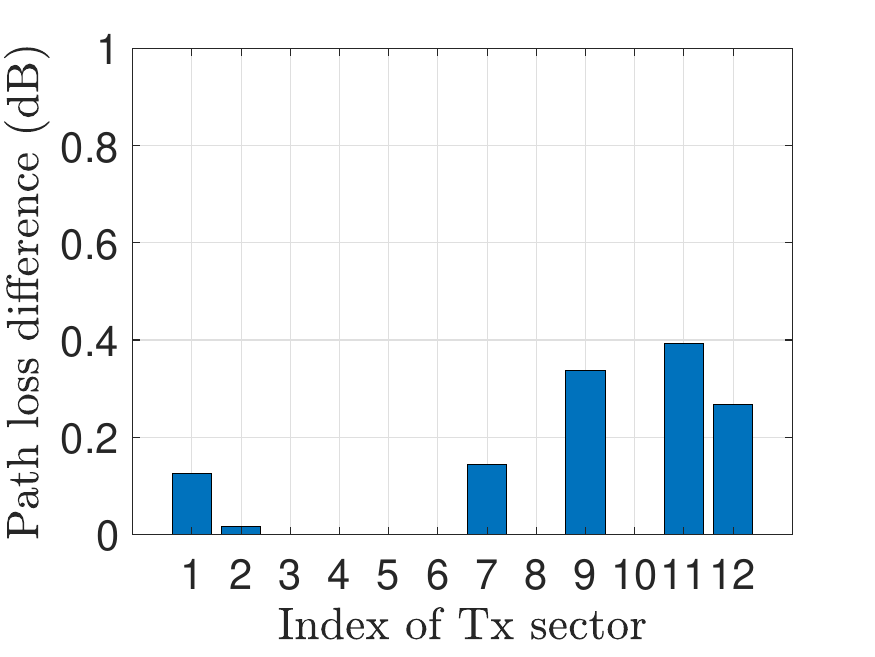} &
  \includegraphics[width=0.55\columnwidth]{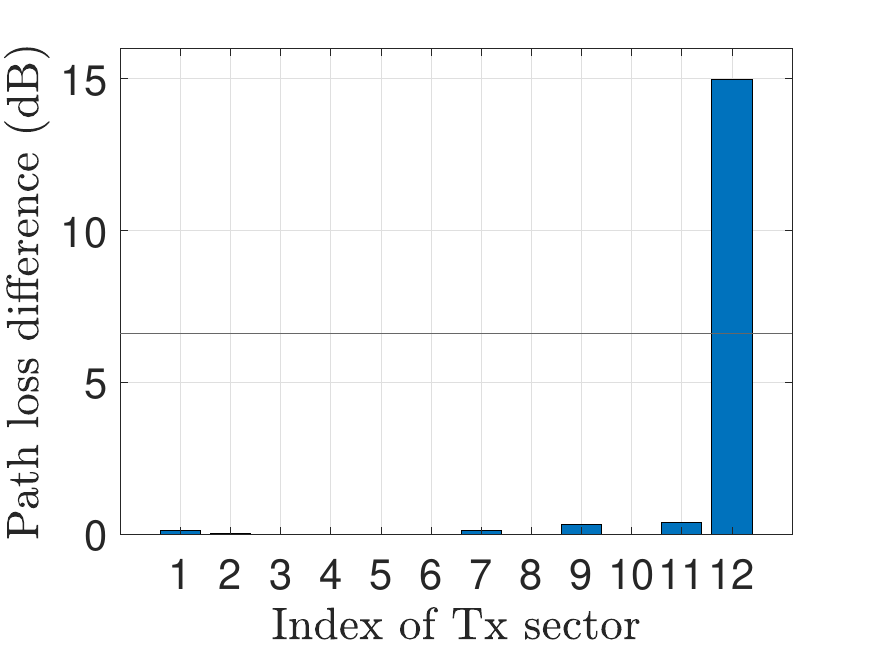} \\ 
  (a) without attack & (b) under attack
\end{tabular}
\caption{Path loss difference without and under attack during the initiator sector sweep phase.}
\label{exp path loss difference}
\vspace{-0.20in}
\end{figure}

\begin{figure}[t]%
\centering
\setlength{\tabcolsep}{-3pt}
\begin{tabular}{cc}
  \includegraphics[width=0.55\columnwidth]{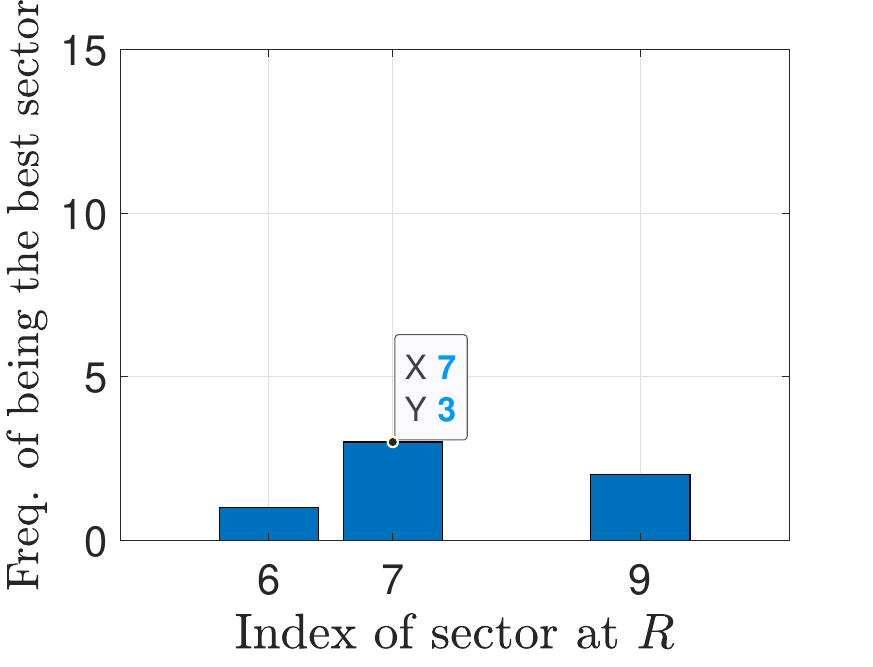} &
  \includegraphics[width=0.55\columnwidth]{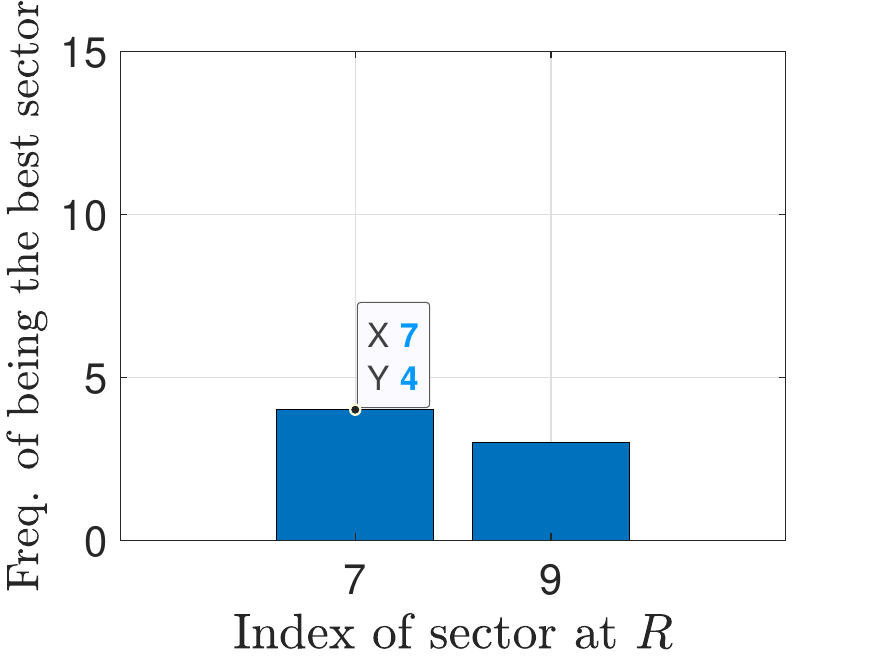} \\ 
  (a) & (b) 
\end{tabular}
\caption{Experimental results without $\mathcal{M}$: (a) frequency of a sector being chosen as the best receive sector with artificial reflectors, (b) Frequency of a sector being chosen as the best receive sector without the artificial reflectors.}
\label{fig:exp rss and freq}
\vspace{-0.20in}
\end{figure}

\subsection{Security Evaluation of SecBeam}

\subsubsection{Security  against the fixed power relay attack}

In our experimental setup, there are 7 transmit sectors that can be detected by $R$ in the first round of sector sweep. The transmit powers for each sector are randomly chosen as $P_i^T(2)=\{16, 7, 22, 0, 10, 9, 18, 19, 4, 11, 10, 15\}$dBm in the second round of sector sweep. The SNR measured at $R$ in two rounds of SSW consist of two matrices of dimension  $10,000 \times 7$. Using these SNR data and the known transmit power values, we calculated the average path loss differences over 10000 measurements for each sector, shown in Fig. \ref{exp path loss difference}. We observe that when there is no adversary present, the average path loss difference for all sectors remains under 0.5 dB, and all 10000 path loss differences for each transmit sector are smaller than the threshold $\epsilon=6.6$. However, when the $FP$ attack is carried out, the average path loss difference for the affected sector ($r_{12}$ in this example) is 14.9875 dB, which is greater than $\epsilon=6.6$.

\subsubsection{Security  against fixed amplification relay attack}
We first present  the  frequency of any receive sector being chosen as the best sector for legitimate devices, as illustrated in Figure~\ref{fig:exp rss and freq}. From Fig.~\ref{fig:exp rss and freq}(a), when the artificial reflectors are present in the environment, the frequency of choosing sector 7 at $R$ as the best sector is found to be $f_7=3/12$. However, after removing the reflectors, in Figure~\ref{fig:exp rss and freq}(b), sector 7 at $R$ was chosen 4 times out of 12 as the best receiving sector, resulting in $f_7=4/12$. The values of $f_7$ in both environments are lower than the  threshold $\beta$, which indicates that the legitimate devices can pass the coarse AoA test.

In the presence of an adversary, we conduct  an experiment where the adversary attempts to relay all  the received signals to the intended receiver $R$. Our experimental results, as shown in Fig. \ref{fig:exp rss and freq no ref}, revealed that the probability of the adversary successfully relaying signals to $R$ from sector 3 was high ($f_3=9/12$) and exceeded the  threshold value of $\beta=0.5$, indicating that the $FA$ attack is detected. We show the best transmit and receive beam pair $S_{ij}$ for this setup in Fig. \ref{fig:exp rss and freq no ref}(c). The black sectors   in Fig. \ref{fig:exp rss and freq no ref}(c) means that signals sent using these sectors cannot be received by $R$. When we removed the reflectors, $f_3$ further increased to 11/12 as shown in Figs. \ref{fig:exp rss and freq no ref}(b) and (d), which is above the threshold value of $\beta$. The adversary is detected. 


\begin{figure}[t]%
\centering
\setlength{\tabcolsep}{-3pt}
\begin{tabular}{cc}
  \includegraphics[width=0.5\columnwidth]{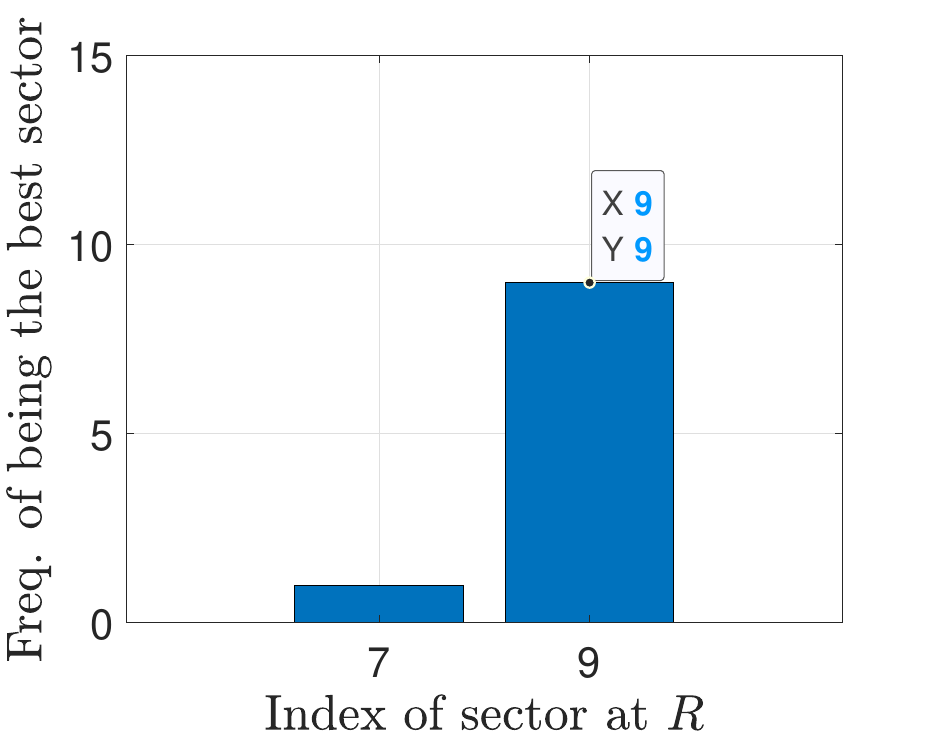} &
  \includegraphics[width=0.5\columnwidth]{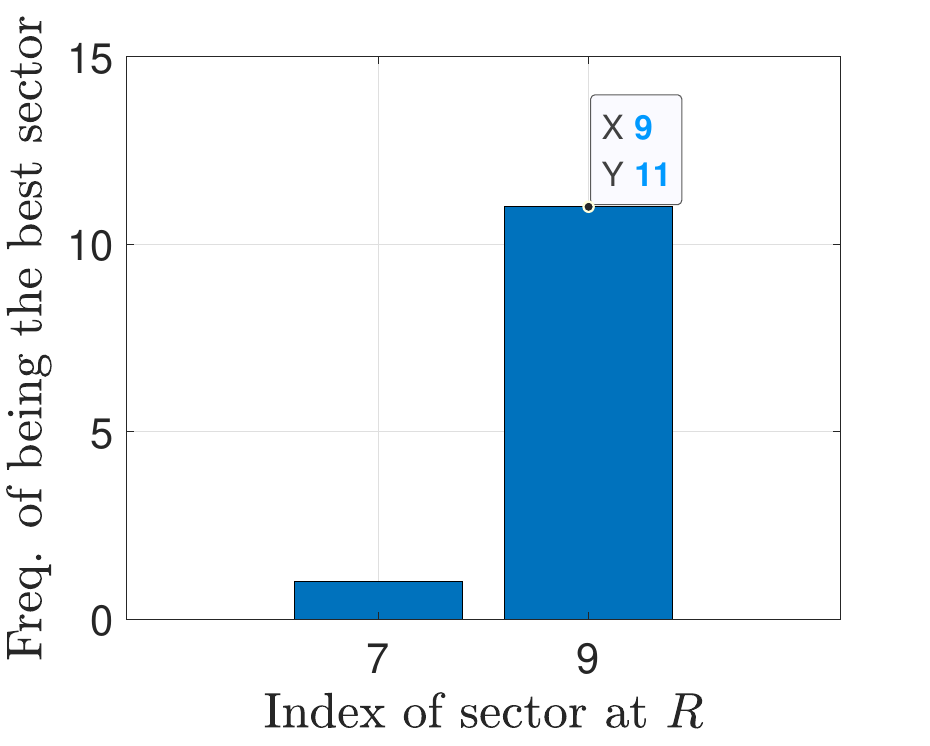} \\ 
  (a) & (b) \\
  \includegraphics[width=0.45\columnwidth]{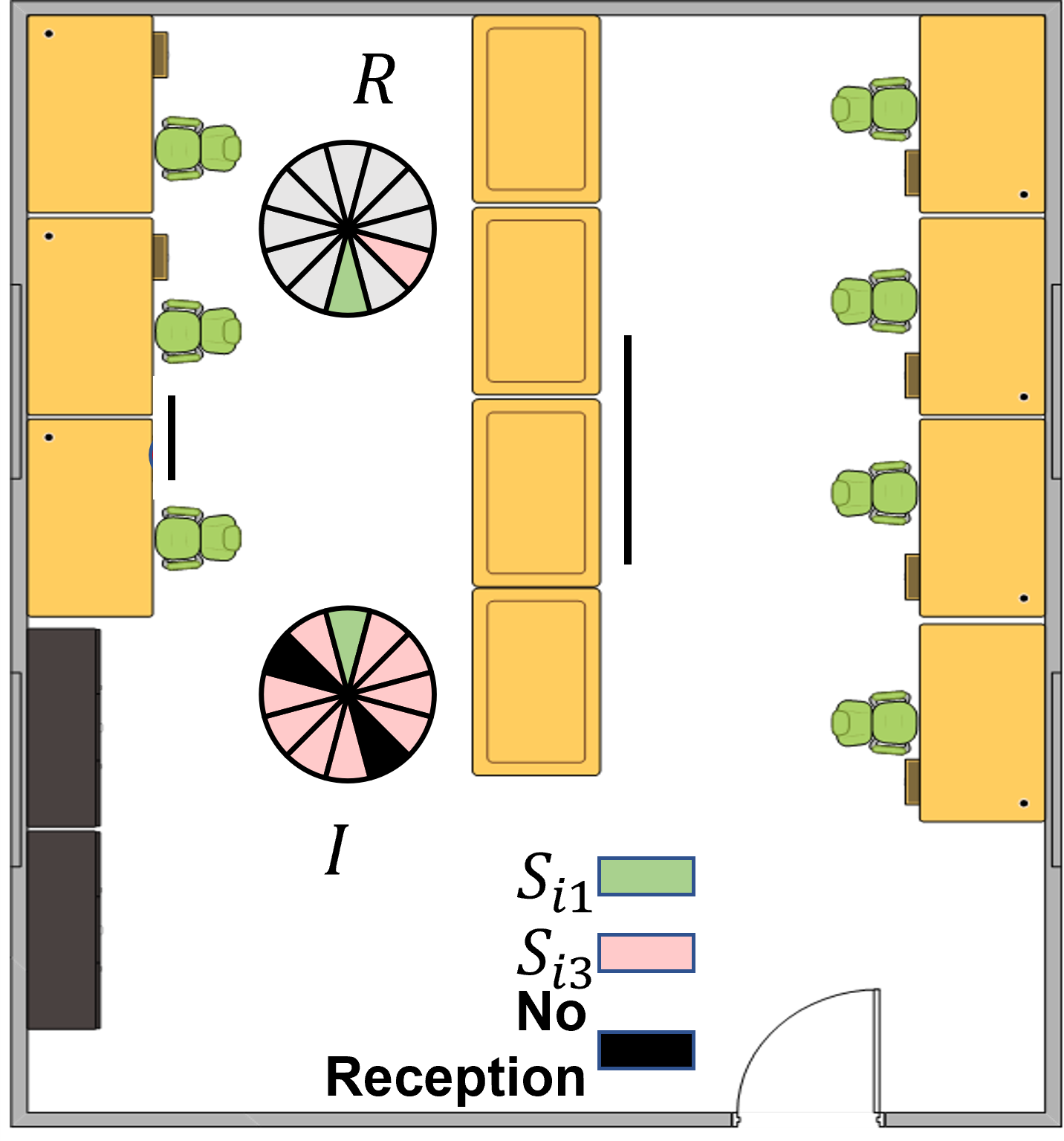} &
  \includegraphics[width=0.45\columnwidth]{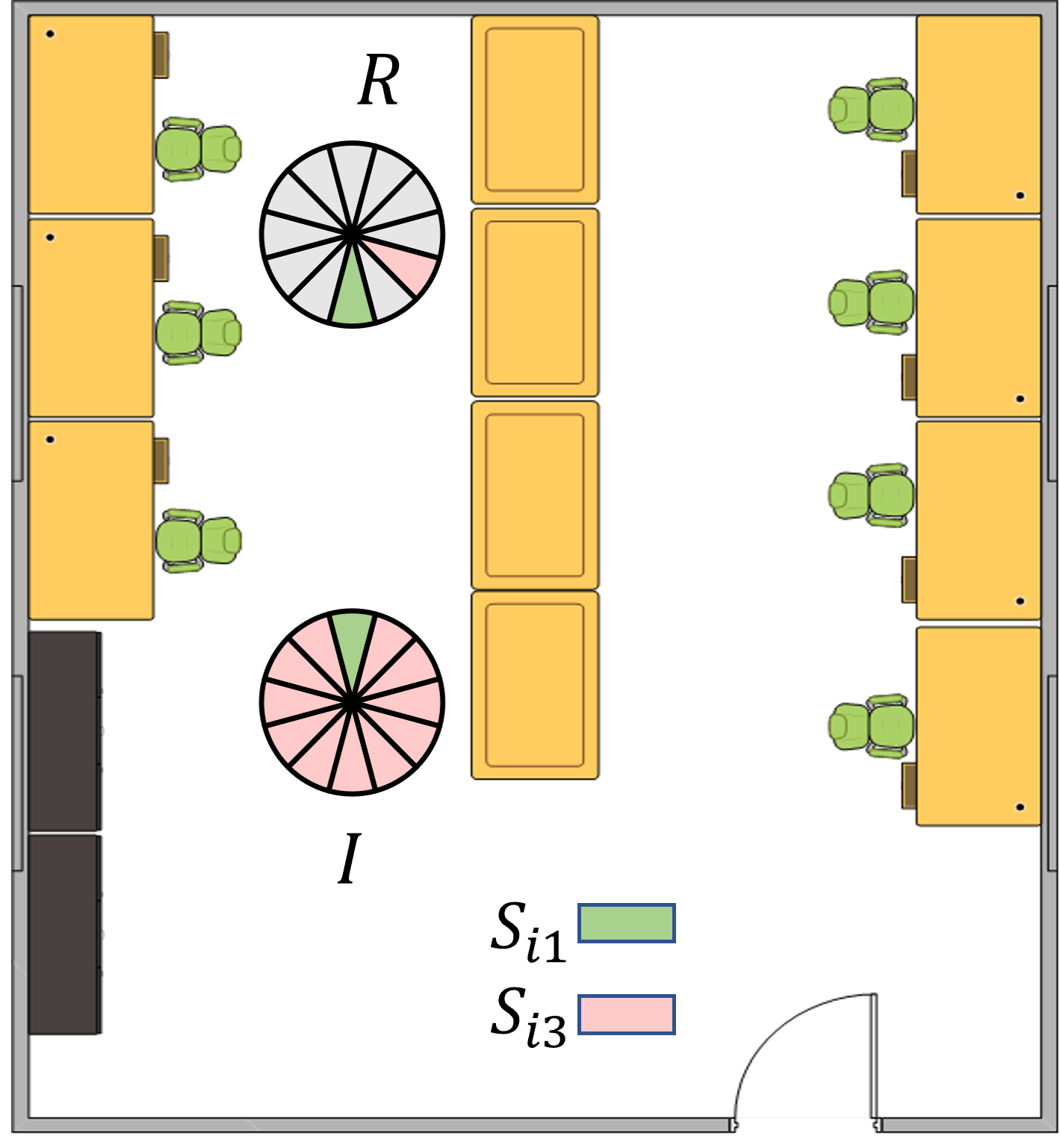} \\
  (c) & (d)
\end{tabular}
\caption{Experimental results under $FA$ relay attack: (a) Frequency of a sector being chosen as the best receive sector when artificial reflectors are in the environment. (b) Frequency of a sector being chosen as the best receive sector without artificial reflectors. (c) $S_{ij}$ on real mmWave devices with artificial reflectors. (d) $S_{ij}$ on real mmWave devices without artificial reflectors.}
\label{fig:exp rss and freq no ref}
\vspace{-0.20in}
\end{figure}

\subsection{Time Efficiency of SecBeam} 
The run time of SecBeam is dominated by the number of sectors at the devices. The time required for RSS and AoA measurements is negligible. Let the size of each sector sweep frame be $F$ bytes, the number of frames sent by each sector be $Y$, bandwidth be $BW$, both $I$ and $R$ have $N_q$ quasi-omni patterns, the time interval between frames sent by the same antenna, Short Beamforming Interframe Spacing (SBIFS), be $T_{S}$ and the time consumed to switch quasi-omni pattern, Long Beamforming Interframe Spacing (LBIFS), be  $T_{L}$. The total time for protocol execution is computed as $2N_q(N(\frac{YF}{BW} +(Y-1) T_S)+(N-N_q)T_S+(N_q-1)T_L)+2(N_q-1)T_L$. If we assume  $F=26$, $Y=1$, $T_{S}=1 \mu s$, and $T_{L}=18 \mu s$ as defined in the  802.11ad standard \cite{IEEE:802.11ad}, $N=32$, $N_q=6$, and $BW=160$ MHz, the total time consumed by SecBeam is 2.07 ms.

\section{Related Work}

\textbf{Relay attacks.} 
The amplify-and-forward beam stealing attacks is a relay type of attack. However, existing relay attacks against wireless systems have different purpose. The majority of relay attacks aim at violating the proximity/distance constraint  between two devices (a prover and a verifier), for applications where proximity-based access control is needed. For example, keyless entry and start systems for cars \cite{francillon2011relay, francis2011practical}, contactless payment systems with smartcards/RFID \cite{kfir2005picking}, etc. We emphasize that the beam-stealing attack differs from the above relay attacks, since it aims at distorting the path-loss property rather than proximity or distance. In addition, there is no location restriction imposed on the attacker.

There is also another type of relay attack  in the literature, namely signal cancellation/annihilation attack  \cite{vcapkun2008integrity, popper2011investigation, moser2019digital,  ghose2017help}, where  the  attacker aims at  causing destructive interference at the intended receiver(s) of a wireless communication link, which either leads to denial-of-service, or can be used to achieve message modification via arbitrary bit flips. However, signal cancellation is generally hard to carry out since it requires precise timing and phase synchronization, which is also  more challenging  in mmWave communication systems due to the directionality of the beams. In contrast,  the amplify-and-forward beam stealing attack is a form of MitM attack which disrupts beam selection so that the attacker has full control of the communication link between the legitimate devices (it can eavesdrop, modify, or perform DoS against the messages exchanged). In addition, it 
is easier to implement than signal cancellation.

\textbf{Relay Attack Defenses.} 
Existing defenses against relay attacks are not applicable to defend against beam-stealing attacks in general. For example,   distance bounding protocols \cite{drimer2007keep, abidin2019quantum, cremers2012distance} aim at enforcing a lower-bound of the distance from a prover to the verifier. However, we target general communication scenarios where no proximity/distance bound is imposed on the communicating devices. Although one can also  use distance bounding as a secure signal path length measurement to indirectly prove that no longer paths   than the dominating path (e.g., LoS) has a lower path-loss, implementing distance bounding typically requires advanced hardware capabilities \cite{tippenhauer2015uwb}.  Moreover, $I$ and $R$ have no knowledge of the wireless environment. If a stronger signal appears on a longer path, this could be because the shorter path experiences higher absorption or diffraction if it is not the LoS.

Existing relay attacks can also be detected by verifying the co-presence of devices via some common context, such as ambient conditions (e.g., sound, lighting, motion, RF signals, etc) \cite{ma2012location, urien2014elliptic, truong2014comparing, Xu_2022, mathur2011proximate}.  However, as stated earlier, they are not applicable to defend beam-stealing attacks due to the different purpose and assumptions of the attack.

On the other hand, another category of defenses utilize physical-layer identification/RF fingerprinting methods to identify the transmitting devices (or fingerprint the channel between the devices) via their signal characteristics, in order to detect impersonation attacks or distinguish   different devices. For example,  Balakrishnan $et\ al.$ \cite{balakrishnan2019physical} fingerprint mmWave devices based on    features extracted from their spatial-temporal beam patterns and demonstrated the resistance of their protocol against impersonation attacks. Wang $et\ al.$ \cite{wang2021exploiting, wang2020machine} adopted a similar idea, and utilized machine learning algorithms to identify each device based on its unique beam pattern. 
Although these methods can be applicable to detect beam-stealing attacks, they all require  a prior (secure) training  phase to extract beam patterns of known devices, which not only brings extra overhead but also does not work well when the channel condition changes or new devices are introduced to the network. In contrast, our proposed defense, SecBeam does not require any prior knowledge about the devices or channel environment.

Finally, we are aware of only one work that  proposed a  specific defense  against the beam-stealing attack under the IEEE 802.11ad protocol (Steinmetzer $et\ al.$ \cite{steinmetzer2018authenticating}). Since the original SLS protocol in IEEE 802.11ad  is unauthenticated which allows beam-stealing by simply modifying the SSW messages, they integrated   crypto machinery into the sector sweep process, in order to authenticate the SSW messages, which prevents the simple form of beam-stealing attacks. However, as we demonstrated in this paper, it is still vulnerable to amplify-and-relay beam-stealing attacks. It is an open problem to defend against such attacks.

\section{Conclusions and Future Work}
In this paper, we studied the security of the 802.11ad  beam alignment protocol for mmWave communications. Although prior work have added cryptographic protection to secure beam alignment messages, we   demonstrate that this protocol is still vulnerable to a new amplify-and-relay beam stealing attack that does not require message forging and bypasses cryptographic protections. We then propose a new secure beam sweeping protocol, SecBeam, which exploits power/sector randomization and coarse angle-of-arrival information to detect  amplify-and-relay attacks. Essentially, SecBeam constructs a physical layer commitment scheme that commits to the path loss of each beam. SecBeam does not require any prior knowledge of the physical environment and is compatible with the current 802.11ad standard. We theoretically analyzed the security of SecBeam, and used both ray-tracing simulations and real-world experiments on a mmWave testbed to evaluate the security of our protocol. Results show that SecBeam can detect two different types of amplify-and-relay attacks under realistic scenarios.
Future works will include extending our method to other efficient beam alignment protocols beyond 802.11ad, for example,  protocols with sub-linear complexity.

\bibliographystyle{IEEEtranS}
\bibliography{ref}

\begin{figure*}%
\centering
\begin{tabular}{ccc}
  \includegraphics[width=0.65\columnwidth]{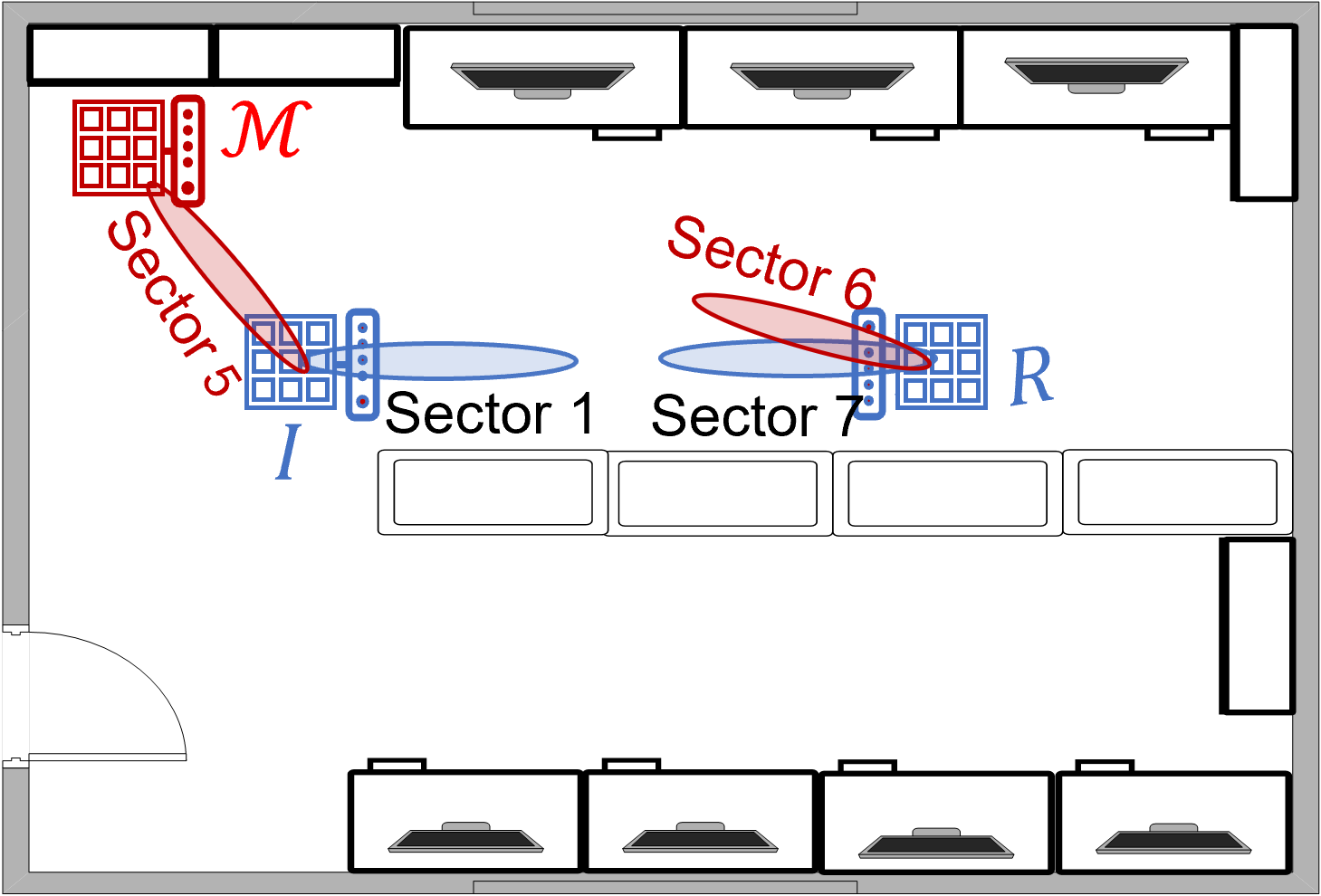} &
  \hspace{-0.1in}
  \includegraphics[width=0.75\columnwidth]{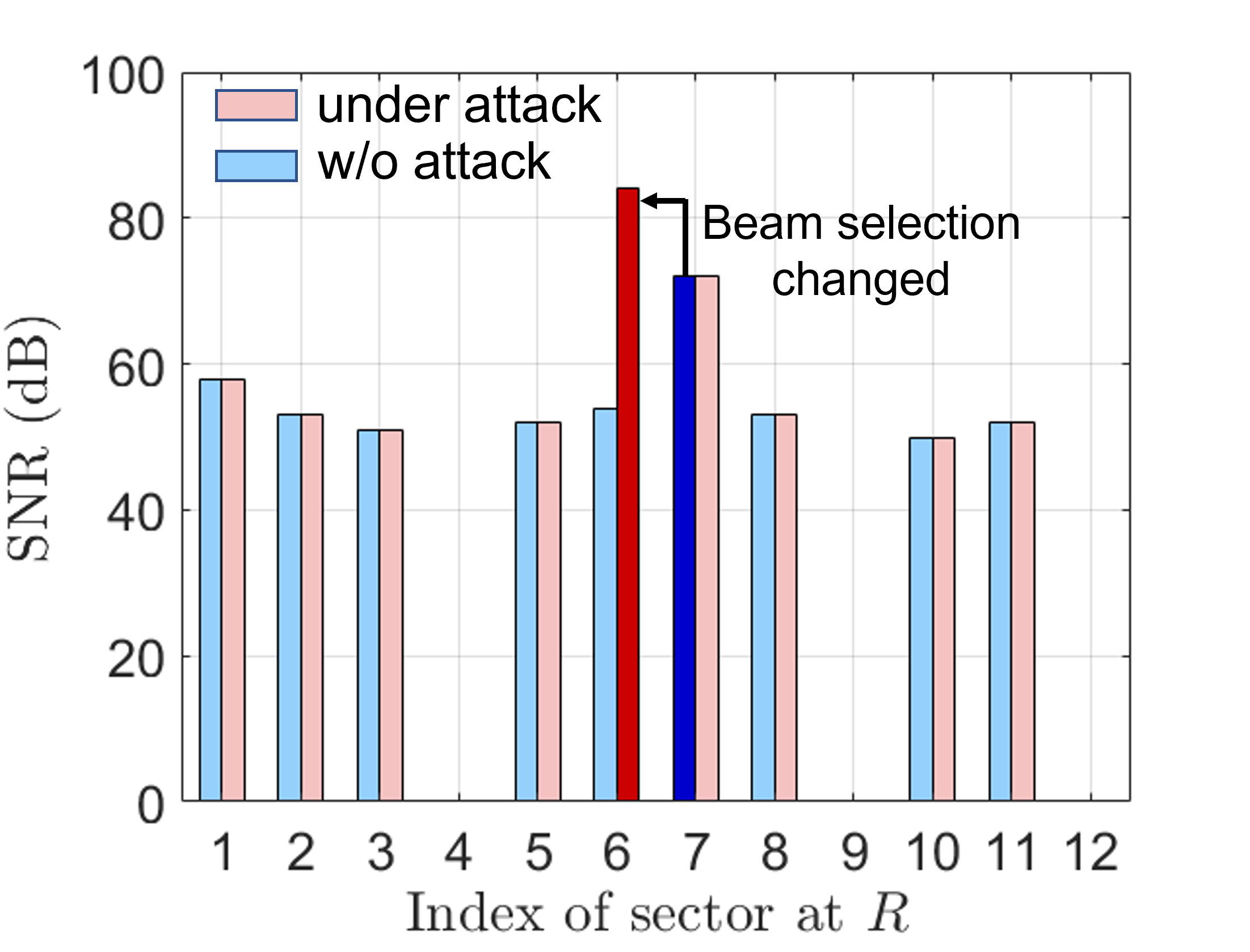} & 
  \hspace{-0.35in}
  \includegraphics[width=0.75\columnwidth]{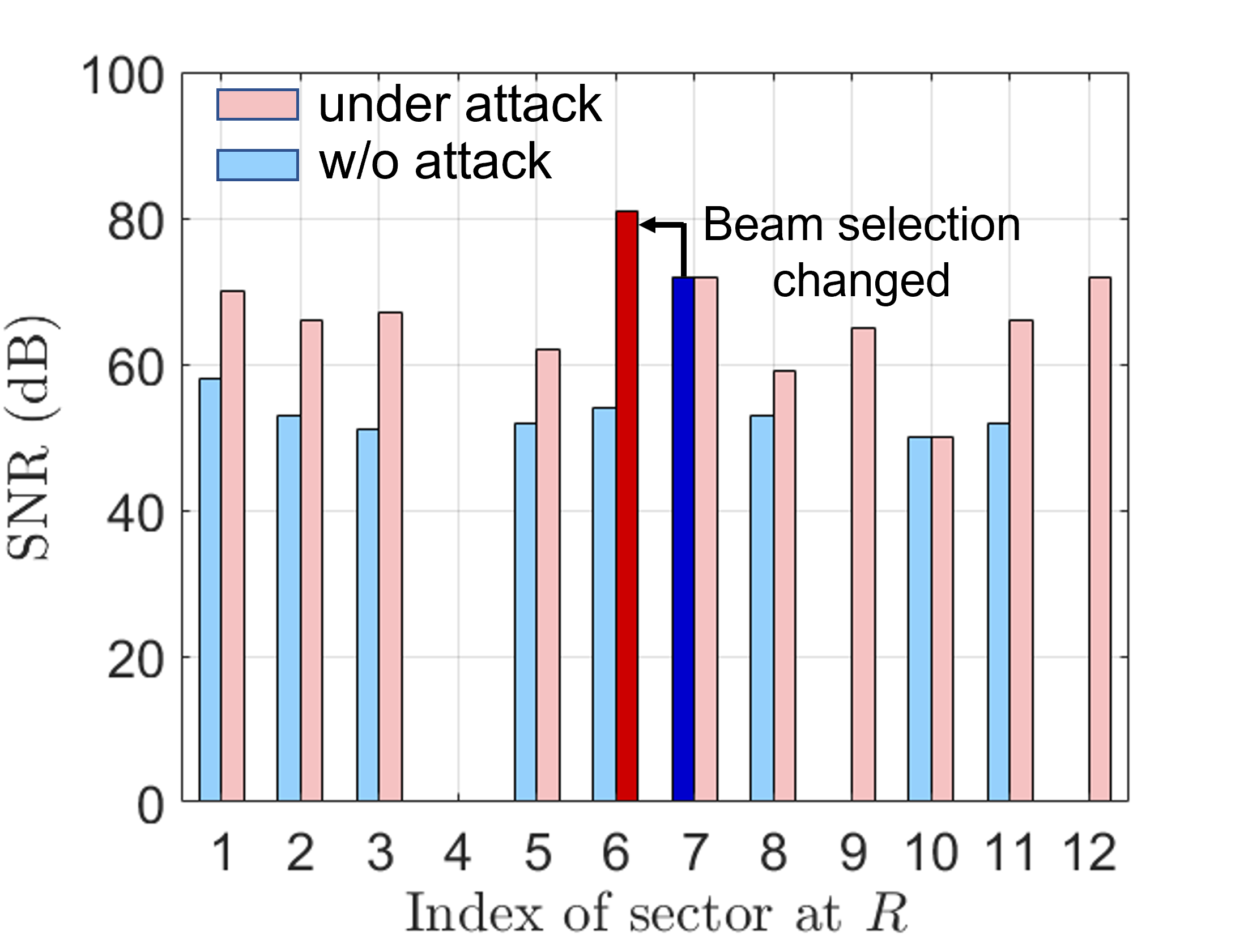} \\
  (a)Experimental setup  & (b)$\mathcal{M}$ relays in $FP$ strategy  & (c)$\mathcal{M}$ relays in $FA$ strategy
\end{tabular}
\caption{Amplify-and-relay beam-stealing attack setup and experimental results when $\mathcal{M}$ uses $FA$ and $FP$ relay strategies. $R$ chooses sector 7 without $\mathcal{M}$, and sector 6 when $\mathcal{M}$ relays in either $FA$ or $FP$ strategies.}
\label{fig:setup1 and result}
\vspace{-0.2in}
\end{figure*}

\newpage

\appendix
\subsection{Amplfy-and-Relay Beam-Stealing Attack in Setup 2}
\label{setup1}
The experimental setup, illustrated in Fig. \ref{fig:setup1 and result}(a), consists of devices that are located at distances of $d_{IM}=2.7 m$, $d_{IR}=4.3 m$, and $d_{MR}=6 m$. To achieve synchronization between devices $I$ and $\mathcal{M}$, an external clock (Ettus OctoClock-G) is used. In this configuration, $\mathcal{M}$ and $R$ are situated in different quasi-omni patterns of $I$, enabling $\mathcal{M}$ to relay sector sweep frames from $R$ to $I$ while $I$ is receiving in the quasi-omni pattern towards itself. Consequently, there is no need to synchronize $\mathcal{M}$ and $R$. The results of this setup are shown in Fig. \ref{fig:setup1 and result}(b)(c) and Fig. \ref{setup1 results on I}. When there is no adversary, $I$ and $R$ select sector 1 and sector 7 as the best beam pair.

Fig. \ref{fig:setup1 and result}(b) and Fig. \ref{setup1 results on I}(a) demonstrate the impact of $\mathcal{M}$ relaying using $FP$ strategy, leading to $R$ and $I$ switching to sector 6 and sector 5, respectively. In Fig. \ref{fig:setup1 and result}(c) and Fig. \ref{setup1 results on I}(b), $\mathcal{M}$ relays using $FA$ strategy, resulting in the selection of sector 5 at $I$ and sector 6 at $R$ again. Ultimately, $\mathcal{M}$ successfully directs $I$ and $R$ to choose beams pointing towards him in both strategies. We further verified that without $\mathcal{M}$, $I$ and $R$ were unable to communicate using sector 5 and sector 6. These results are in line with those presented in Sec. \ref{feasibility amplify attack}.

\begin{figure}[t]%
\centering
\setlength{\tabcolsep}{-5pt}
\begin{tabular}{cc}
 \includegraphics[width=0.55\columnwidth]{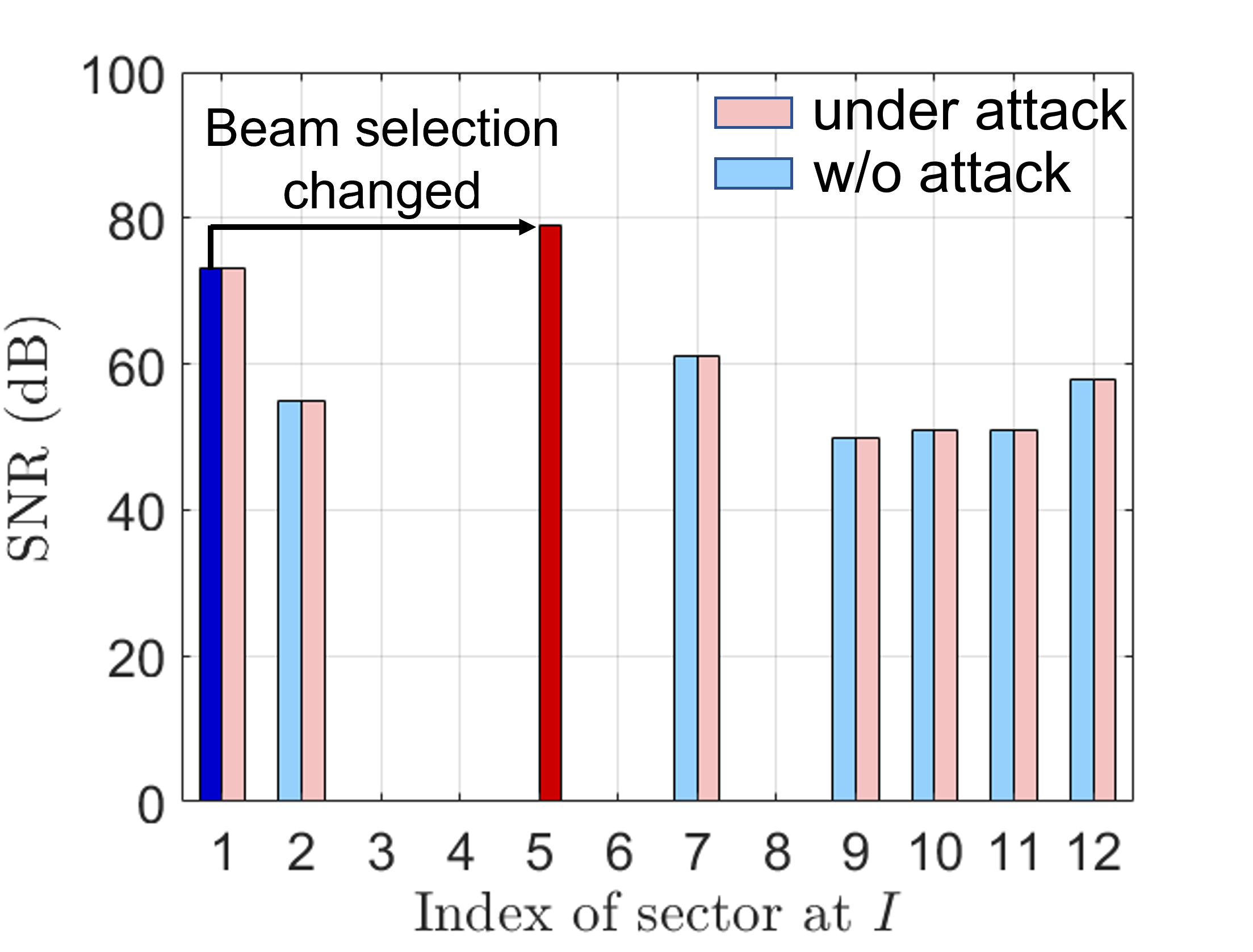} &
  \includegraphics[width=0.55\columnwidth]{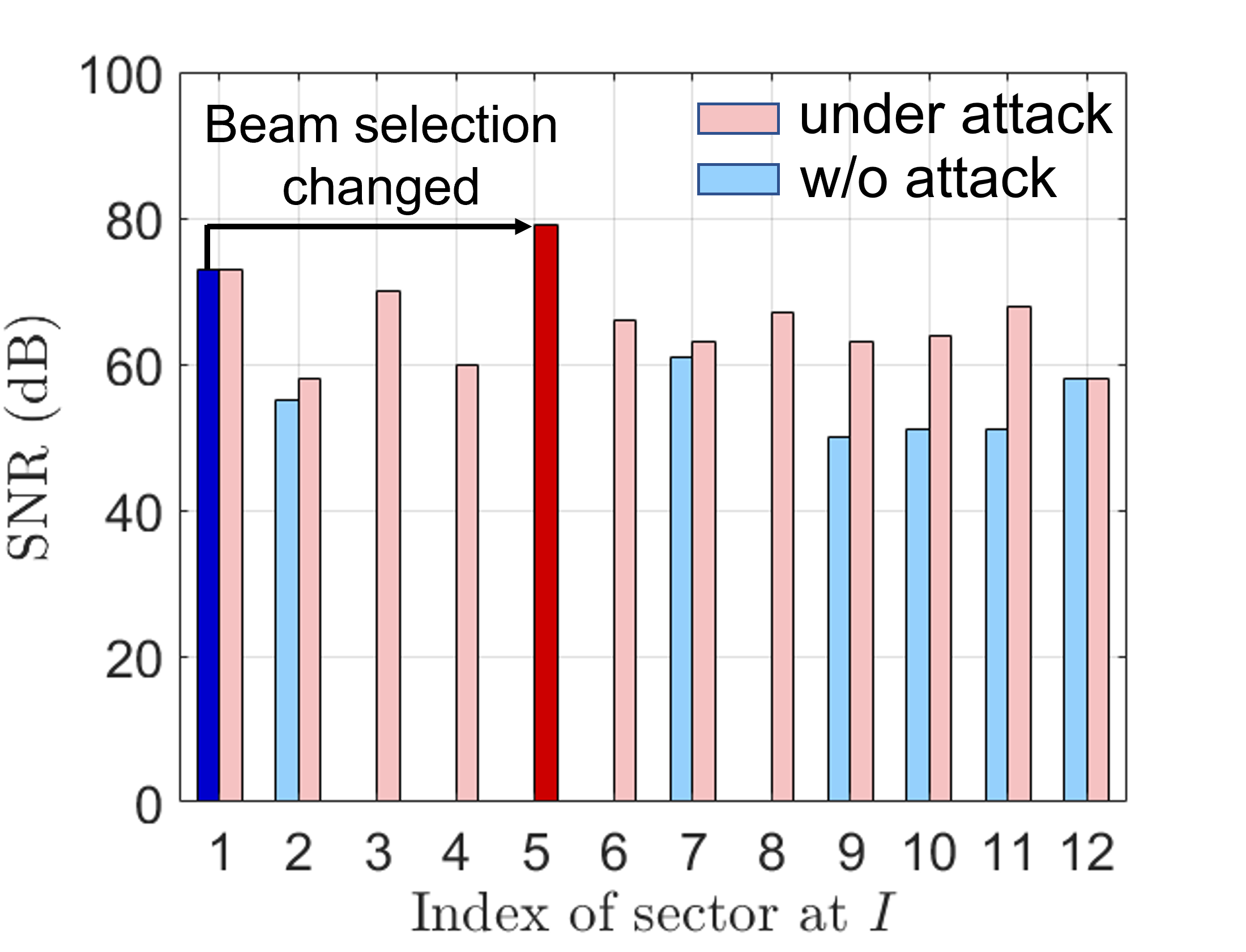} \\   
  (a)$\mathcal{M}$ relays in FP strategy & (b) $\mathcal{M}$ relays in FA strategy
\end{tabular}
\caption{Amplify-and-relay beam stealing attack on $I$ using $FA$ and $FP$ strategies. $I$ chooses sector 1 without $\mathcal{M}$, and sector 5 when $\mathcal{M}$ relays in either $FA$ or $FP$ strategies.}
\label{setup1 results on I}
\vspace{-0.15in}
\end{figure}

\subsection{Proof of Proposition 1}

\noindent
{\bf Proposition 1.} { \em A fixed power attack on sector $S_i$ is detected by the path loss test if the transmission power $P_i(T)(2)$  in the second round satisfies 
\[
P_i^T(2) < P_i^{T}(1) - \epsilon,
\]
where $\epsilon$ is a noise-determined parameter.}

\begin{proof}
Consider the transmission of an SSF at sector $S_i$ by $I$ during the first round with power $P_i^T(1)$. The path loss calculated by $R$ is $PL_i(1)(dB) = P_i^T(1)(dBm) - P_i^R(1)(dBm).$ Similarly, $R$ calculates the path loss of the second round as $PL_i(2)(dB) = P_i^T(2)(dBm) - P_i^R(2)(dBm).$ Taking the absolute difference yields
\begin{eqnarray}
|PL_i(1) - PL_i(2)| &=& PL_i(1) - PL_1(2) \notag\\
&=& P_i^T(1) - P_i^R(1) - P_i^T(2) + P_i^R(2) \notag \\
&>& P_i^T(1) - P_i^T(1) + \epsilon \notag \\
&>&  \epsilon.
\label{epsilon}
\end{eqnarray}
In \eqref{epsilon}, we used the fact that $PL_i(1) > PL_1(2)$ because $I$'s transmission power is reduced in the second round whereas Mallory transmits at fixed power thus leading to the calculation of lower path loss. Moreover, because Mallory's transmission is at fixed power, it follows that  $P_i^R(1) = P_i^R(2).$ Finally, by design,  $P_i^T(2) < P_i^{T}(1) - \epsilon.$
\end{proof}

\subsection{Proof of Proposition 2}

\noindent
 {\bf Proposition 2.} {\em Given the number of sectors $N_M$ heard at $\mathcal{M}$ and the number of sectors $K$  in $N_M$ with path loss smaller than the best path without attack under a fixed amplification strategy, the probability $P_{BS}$ of a successful beam stealing attack is given by 
 \[
    Pr_{BS}=\frac{1}{\binom{N_M}{x}^4} \cdot \left (\sum_{i=1}^{\min(x, K)}\binom{K}{i}\binom{N_M-K}{x-i}\right)^2.
\] }

\begin{proof}
   To prove Proposition 2, we first describe the adversary's strategy. During the first round, Mallory must choose which SSFs out of the $N_M$ ones that she receives should be amplified and forwarded to $R.$ Given the online nature of the attack, this decision is made in real time. Let Mallory decide to forward $x$ out of the $N_M$ sectors, without being able to identify them since the sector IDs are encrypted. The probability of choosing at least one sector that can defeat the optimal legitimate sector without attack (and therefore steal the beam) is given by the probability that any of the $K$ amplified sectors that defeat the optimal one are chosen. We compute this probability as
\begin{equation}
    Pr_1 = \sum_{i=1}^{\min\{x,K\}} \frac{\binom{K}{i} \cdot \binom{N_M-K}{x-i}}{\binom{N_M}{x}}.
    \label{p1}
\end{equation}
Equation~\eqref{p1} is the enumeration of all possible ways of choosing exactly $i$ sectors out of $K$, when Mallory chooses to amplify $x$ out of $N_M$ sectors at random. In the second round, $\mathcal{M}$ has to amplify exactly the same $x$ sectors to pass the path loss test. If any sector out of the $x$ amplified in round 1 is not amplified in the second round, the path loss calculated for that sector will differ by the amplification gain applied in the first round. 

However, the sweep sequence is randomized in the second round by the application of the pseudorandom permutation $\Pi.$ Moreover, Mallory cannot decrypt SSFs to recover the sector ID. Even if she could, the online nature of the attack does not permit enough time for frame decoding. Therefore, the probability of passing the path loss test is given by the probability of selecting the same set of $x$ SSFs, given by 
\(
    Pr_2=1/ \binom{N_M}{x}.
\)
Combining the events of choosing a sector that defeats the optimal one in the first round and choosing the same sectors to amplify in the second round yields Mallory's success probability in stealing the initiator's beam 
\begin{equation}
    \label{PM}
    Pr_{FA}=\frac{1}{\binom{N_M}{x}^2} \cdot \left (\sum_{i=1}^{\min(x, K)}\binom{K}{i}\binom{N_M-K}{x-i}\right)
\end{equation}
To successfully steal the beams in both directions, $\mathcal{M}$ must succeed in both the initiator and responder sector sweep processes. This probability is given by $Pr_{BS}=\left(Pr_{FA}\right)^2$.  
\end{proof}

\end{document}